\documentclass[numberwithinsect, a4paper, USenglish, cleveref, autoref, thm-restate]{lipics-v2021}
\usepackage{cite}
\usepackage[most]{tcolorbox}

\newtheorem{hypothesis}[theorem]{Hypothesis}

\crefname{maintheorem}{Main Theorem}{Main Theorems}

\pdfoutput=1
\hideLIPIcs

\graphicspath{{figures/}}
\title{The Price of Being Partial: Complexity of Partial Generalized Dominating Set on Bounded-Treewidth Graphs}

\titlerunning{Complexity of Partial Generalized Dominating Set on Bounded-Treewidth Graphs}

\author{Jakob Greilhuber}{CISPA Helmholtz Center for Information Security, Saarbrücken, Germany}{jakob.greilhuber@cispa.de}{https://orcid.org/0009-0001-8796-6400}{Member of the Saarbrücken Graduate School of Computer Science, Germany}

\author{D\'{a}niel Marx}{CISPA Helmholtz Center for Information Security, Saarbrücken, Germany}{marx@cispa.de}{https://orcid.org/0000-0002-5686-8314}{}

\authorrunning{J. Greilhuber and D. Marx}
\Copyright{Jakob Greilhuber and D\'{a}niel Marx}


\nolinenumbers

\let\originalleft\left
\let\originalright\right
\renewcommand{\left}{\mathopen{}\mathclose\bgroup\originalleft}
\renewcommand{\right}{\aftergroup\egroup\originalright}
\renewcommand{\epsilon}{\varepsilon}

\makeatletter
\let\ams@underbrace=\underbrace
\def\underbrace#1_#2{\setbox0=\hbox{$\displaystyle#1$}\ams@underbrace{#1}_{\parbox[t]{\the\wd0}{#2}}}
\makeatother

\newcommand{\problem}[1]{\textnormal{\textsc{#1}}}
\newcommand{\uptextsf}[1]{\textnormal{\textsf{#1}}}

\newcommand{\sigStatesNonPartial}{\mathbb{S}}
\newcommand{\rhoStatesNonPartial}{\mathbb{R}}
\newcommand{\allStatesNonPartial}{\mathbb{A}}

\newcommand{\sigStatesPartial}{\mathbb{S}_{\textup{\uptextsf{p}}}}
\newcommand{\rhoStatesParital}{\mathbb{R}_{\textup{\uptextsf{p}}}}
\newcommand{\allStatesPartial}{\mathbb{A}_{\textup{\uptextsf{p}}}}

\newcommand{\sigLargestPartial}{s_\sigma^\textup{\uptextsf{p}}}
\newcommand{\rhoLargestPartial}{s_\rho^{\textup{\uptextsf{p}}}}
\newcommand{\allLargestPartial}{s_{\sigma,\rho}^{\textup{\uptextsf{p}}}}
\newcommand{\tauLargestPartial}{s_\tau^{\textup{\uptextsf{p}}}}

\newcommand{\sigLargestNonPartial}{s_\sigma}
\newcommand{\rhoLargestNonPartial}{s_\rho}
\newcommand{\allLargestNonPartial}{s_{\sigma,\rho}}
\newcommand{\tauLargestNonPartial}{s_{\tau}}

\newcommand{\portalSelectionGadgetConstant}{c_{\textup{\uptextsf{t}}}}
\newcommand{\penaltySym}{\delta}

\newcommand{\genDomSetRel}{$(\sigma,\rho)$-$\problem{DomSet}_\mathcal{R}$} \newcommand{\minGenDomSet}{$(\sigma,\rho)$-$\problem{MinDomSet}$} \newcommand{\minGenDomSetRel}{$(\sigma,\rho)$-$\problem{MinDomSet}_{\mathcal{R}_{\supseteq}}$} \newcommand{\parGenDomSetRel}{$(\sigma,\rho)$-$\problem{ParDomSet}_{\mathcal{R}}$} \newcommand{\minParGenDomSet}{$(\sigma,\rho)$-$\problem{MinParDomSet}$} \newcommand{\minParGenDomSetRel}{$(\sigma,\rho)$-$\problem{MinParDomSet}_{\mathcal{R}_\supseteq}$} 

\newcommand{\CSP}[1][2]{$#1$-\problem{CSP}-$B$}
\newcommand{\kSAT}[1][k]{$#1$-\problem{SAT}}

\newcommand{\hwRelation}{\mathtt{HW}_{=1}}
\newcommand{\hwGeqOneRelation}{\mathtt{HW}_{\geq 1}}

\newcommand{\scp}[1]{\uptextsf{scp}\left(#1\right)}
\newcommand{\acc}[1]{\uptextsf{acc}\left(#1\right)}
\newcommand{\complAcc}[1]{\overline{\uptextsf{acc}}\left(#1\right)}

\newcommand{\Oh}{O}

\newcommand{\tw}{{\uptextsf{tw}}}
\newcommand{\pw}{{\uptextsf{pw}}}

\newcommand{\range}[2][1]{[#1,#2]}

\newcommand{\nat}{\mathbb{Z}_{\geq 0}}

\newcommand{\usedStates}{A}
\newcommand{\blocksPerVertexGadget}{g}

\newcommand{\selStateVertexPerGadget}{f_{\sigma}}
\newcommand{\unselStateVertexPerGadget}{f_{\rho}}

\newcommand{\weightConstant}{w_{\sigma, \rho}}
\newcommand{\selWeight}{w_{\sigma}}
\newcommand{\unselWeight}{w_{\rho}}

\newcommand{\numOverflowSigma}{c_\sigma}
\newcommand{\numOverflowRho}{c_\rho}
\newcommand{\numOverflowAll}{c_{\sigma,\rho}}

\NewDocumentCommand{\weight}{m}{\IfBlankTF{#1}{\textnormal{\uptextsf{w}}}
  {\textnormal{\uptextsf{w}}\left(#1\right)}
}

\NewDocumentCommand{\rhoWeight}{m}{\IfBlankTF{#1}{\textnormal{\uptextsf{w}}_\rho}
  {\textnormal{\uptextsf{w}}_\rho\left(#1\right)}
}

\NewDocumentCommand{\sigWeight}{m}{\IfBlankTF{#1}{\textnormal{\uptextsf{w}}_\sigma}
  {\textnormal{\uptextsf{w}}_\sigma\left(#1\right)}
}

\newcommand{\variable}[1]{x_{#1}}
\newcommand{\constraint}[1]{C_{#1}}
\newcommand{\variableCount}{n}
\newcommand{\constraintCount}{m}

\newcommand{\bag}[1]{X_{#1}}
\newcommand{\bagCount}{t}
\newcommand{\outputBag}[1]{\hat{X}_{#1}}

\newcommand{\state}[1]{\uptextsf{st}\left(#1\right)}
\newcommand{\rightState}[1]{\compl{\uptextsf{st}}\left(#1\right)}

\newcommand{\compl}[1]{\overline{#1}}
\newcommand{\complState}[1]{\uptextsf{compl}\left(#1\right)}

\newcommand{\statevertex}[3]{u^{#1}_{#2,#3}} \newcommand{\leftBlock}[3]{L^{#1}_{#2,#3} }  \newcommand{\rightBlock}[3]{\compl{L}^{#1}_{#2,#3}}

\NewDocumentCommand{\constraintBagMapping}{m}{\IfBlankTF{#1}{\gamma}
  {\gamma(#1)}
}

\NewDocumentCommand{\cspSolution}{m}{\IfBlankTF{#1}{\delta}
  {\delta(#1)}
}

\newcommand{\initialRelation}[2]{\mathtt{I}^{#1}_{#2}} \newcommand{\consistencyRelation}[2]{\mathtt{Con}^{#1}_{#2}}
\newcommand{\finalRelation}[2]{\mathtt{F}^{#1}_{#2}}

\newcommand{\constraintRelation}[1]{\mathtt{C}_{#1}} 

\newcommand{\problembox}[3]{\begin{tcolorbox}[
            enhanced,
            colback=white,
            colframe=black,
            boxrule = 0.5pt,
            coltitle=black,
            title=#1,
            rounded corners,
            attach boxed title to top left={yshift=-10pt, xshift=6pt},
            bottom=-0.5mm,
            boxed title style={
                    interior style={fill=white},
                    frame hidden
                }
        ]
        \begin{tabularx}{12.5cm}{ r X p {0.5cm}}
            Input:     & #2        \\
            Question:  & #3
        \end{tabularx}
    \end{tcolorbox}
}

\newcommand{\simpleManagerLeftSolSize}{M_L}
\newcommand{\simpleManagerRightSolSize}{M_{\compl{L}}}
\newcommand{\simpleManagerPortSolSize}{M_U}
\newcommand{\parityConstant}{{r}}

\newcommand{\pwseth}{{PWSETH}} 
\begin{document}
\maketitle

\begin{abstract}
    For fixed sets $\sigma, \rho$ of non-negative integers,
the $(\sigma, \rho)$-domination framework introduced by Telle [\textit{Nord.\ J.\ Comput.} 1994] captures many classical graph problems. For a graph $G$, 
a $(\sigma,\rho)$-set is a set $S$ of vertices such that for every $v\in V(G)$, we have
\begin{itemize}
    \item[(1)] if $v \in S$, then $|N(v) \cap S| \in \sigma$, and
    \item[(2)] if $v \notin S$, then $|N(v) \cap S| \in \rho$.
\end{itemize}
Algorithms and lower bounds for the decision, optimization, and counting versions of finding $(\sigma,\rho)$-sets on bounded-treewidth graphs were systematically studied [van Rooij et al., \textit{ESA} 2009][Focke et al., \textit{TALG} 2025]. 
We initiate the study of a natural partial variant  \minParGenDomSet{}  of the problem, in which the constraints given by $\sigma, \rho$ need not be fulfilled for all vertices, but we want to find a set of size at most $k$ that maximizes the number of vertices that are satisfied in the sense that they satisfy (1) and (2) above.

Our goal is to understand whether \minParGenDomSet{} can be solved in the same running time as the nonpartial version, or whether it is strictly harder. 
Formally, we consider nonempty finite or simple cofinite sets $\sigma$ and $\rho$ (simple cofinite sets are of the form $\mathbb{Z}_{\geq c}$), and we try to determine the smallest constant $c_{\sigma,\rho}$ such that there is a $c_{\sigma,\rho}^{\tw}\cdot n^{O(1)}$ time algorithm for the problem if a tree decomposition of width $\tw$ is given. 
We obtain matching upper and lower bounds on $c_{\sigma,\rho}$ for every such fixed $\sigma$ and $\rho$ under the Primal Pathwidth Strong Exponential Time Hypothesis, and establish whether the partial problem is harder than the nonpartial variant. 
For some sets $\sigma$ and $\rho$, the more general \minParGenDomSet{} has the same complexity as the nonpartial special case (e.g., for \textsc{Dominating Set}), while for other choices, the partial version is significantly harder (e.g., for \textsc{Perfect Code}).

 \end{abstract}

\newpage
\tableofcontents
\newpage

\section{Introduction}
\label{sec:introduction}
Treewidth is one of the most-studied notions of capturing how tree-like a graph is, with significant algorithmic and graph-theoretic applications (see, for example, \cite[Chapter 7]{cyganParameterizedAlgorithms2015}).
Many graph problems can be solved in polynomial-time on trees, and similarly, many problems admit polynomial-time algorithms on graphs of bounded treewidth
\cite{Bodlaender1997a,bodlaenderCombinatorialOptimizationGraphs2008,bodlaenderTreewidthComputationsUpper2010,courcelleMonadicSecondOrderLogic1990,cyganFastHamiltonicityChecking2018,cyganParameterizedAlgorithms2015,fominExactExponentialAlgorithms2010,RobertsonSeymour1986,tellePracticalAlgorithmsPartial1993,vanrooijDynamicProgrammingTree2009,vanrooijFastAlgorithmsJoin2020}.
In fact, it is well-known that any graph problem expressible in (extensions of) monadic second-order logic can be solved in polynomial time on graphs of bounded treewidth \cite{arnborgEasyProblemsTreedecomposable1991,courcelleMonadicSecondOrderLogic1990}, which captures a large number of natural graph problems such as \problem{Hamiltonian Cycle}, $q$-\problem{Coloring}, and even optimization problems like \problem{Vertex Cover}.

Showing that a problem can be solved in polynomial-time on graphs of bounded treewidth is often a standard application of dynamic programming techniques or algorithmic meta-theorems.
However, improving the running time and obtaining optimal dependence on treewidth can be challenging.
Let us illustrate this using
the case of \problem{Dominating Set}, one of the most fundamental algorithmic problems studied on graphs.
On graphs of bounded treewidth, the initial algorithm by Telle and Proskurowski \cite{tellePracticalAlgorithmsPartial1993} solves the problem in time $9^\tw \cdot |G|^{\Oh(1)}$, where $\tw$ is the width of a tree decomposition that is provided with the input.
Later, the running time was improved to $4^\tw \cdot |G|^{\Oh(1)}$ \cite{alberFixedParameterAlgorithms2002a,alberImprovedTreeDecomposition2002}, and finally even $3^{\tw} \cdot |G|^{\Oh(1)}$ \cite{vanrooijDynamicProgrammingTree2009}  using the technique of Fast Subset Convolution \cite{bjorklundFourierMeetsMobius2007}.
Can we expect further improvements in the future?
More generally, for any graph problem that has single-exponential dependency on the treewidth, what is the smallest constant $c$ such that the problem can be solved in $c^{\tw} \cdot |G|^{\Oh(1)}$?

Lokshtanov, Marx, and Saurabh \cite{lokshtanovKnownAlgorithmsGraphs2018} initiated a line of research aimed at addressing questions of this form, providing the first answers about the optimal dependence on treewidth for problems such as \problem{Dominating Set}, \problem{Independent Set}, and \problem{Max Cut}.
All of these problems can be solved in time $c^{\tw} \cdot |G|^{\Oh(1)}$ for some problem-dependent constant $c$, but they show that an algorithm running in time $(c - \varepsilon)^{\tw} \cdot |G|^{\Oh(1)}$ would refute the
Strong Exponential Time Hypothesis (SETH) \cite{calabroComplexitySatisfiabilitySmall2009,impagliazzoComplexityKSAT2001}.
Many further results showing lower bounds for problems on graphs of bounded treewidth (which are actually often already lower bounds for the parameter pathwidth) under the SETH have appeared in the literature
\cite{andreevParameterizedComplexityPaired2024,bodlaenderParameterizedComplexityConflictfree2021,bojikianFinegrainedComplexityComputing2025,borradaileOptimalDynamicProgram2016,curticapeanTightConditionalLower2016,cyganFastHamiltonicityChecking2018,cyganSolvingConnectivityProblems2022,dubloisNewAlgorithmsMixed2021,dubloisUpperDominatingSet2022,egriFindingListHomomorphisms2018,esmerFundamentalProblemsBoundedtreewidth2024a,esmerListHomomorphismsDeleting2024a,esmerGeneralizedGraphPacking2025a,fockeTightComplexityBoundsLowerBound,fockeCountingListHomomorphisms2024,greilhuberResidueDominationBoundedTreewidth2025,hanakaParameterizedOrientableDeletion2020a,hartmannIndependenceDominationBoundedtreewidth2025,hegerfeldExactStructuralThresholds2022a,jansenComputingChromaticNumber2019a,katsikarelisStructurallyParameterizedDscattered2022,lampisStructuralParameterizationsTwo2024a,lampisStructuralParameterizationsInduced2025a,marxDegreesGapsTight2021,marxAntifactorFPTParameterized2025,meybodiParameterizedComplexity12020,okrasaFullComplexityClassification2020a,okrasaFinegrainedComplexityGraph2021}.
These results usually confirm that the known algorithms for the studied problem are already essentially optimal, or by the systematic study of a problem setting reveal that there are cases where new nontrivial algorithmical insights are needed to match the lower bounds.
In this line of work, it is a standard assumption that the tree decomposition is supplied together with the input, as one wants to separate the complexity of computing such a decomposition from the complexity of solving the graph problem on such a decomposition.

However, even though the SETH is the standard complexity conjecture used in the field, it is not universally accepted to be true.
Recently, Lampis \cite{lampisPrimalPathwidthSETH2025} introduced the Primal Pathwidth Strong Exponential Time Hypothesis (\pwseth{}), partly to obtain more plausible lower bounds for the parameter pathwidth (and hence also for treewidth).
The primal graph of an instance of \kSAT{} contains a vertex for each variable, and two vertices are connected by an edge if they appear together in some clause.
\begin{hypothesis}[\pwseth{} {\cite[Conjecture 1.1]{lampisPrimalPathwidthSETH2025}}]
    \label{conj:pwseth}
    For any $\varepsilon > 0$, no algorithm can solve \kSAT[3]{} in time $(2 - \varepsilon)^{\pw{}} \cdot |I|^{\Oh(1)}$ on instances $I$, even when the input is provided together with a path decomposition of the primal graph of $I$ of width \pw{}.
\end{hypothesis}
Lampis shows that the \pwseth{} is implied by the SETH, the Set Cover Conjecture \cite{cyganProblemsHardCNFSAT2016}, and the $k$-Orthogonal Vectors Assumption (see e.g. \cite{williamsFinegrainedQuestionsAlgorithms}), making lower bounds obtained under it more plausible than those obtained under SETH alone.
Additionally, Lampis proves lower bounds under the \pwseth{} for various problems that can serve as a convenient starting point for further reductions.
Even though the \pwseth{} is quite recent, there are already further publications presenting lower bounds under it.
Hartmann and Marx study distance variants of \problem{Independent Set} and \problem{Dominating Set} \cite{hartmannIndependenceDominationBoundedtreewidth2025}, Lampis and Vasilakis consider \problem{Acyclic Matching} and \problem{Induced Matching} \cite{lampisStructuralParameterizationsInduced2025a}.
Esmer and Marx \cite{esmerGeneralizedGraphPacking2025a} show tight lower bounds for clique partitioning problems.
All novel lower bounds we present in this paper are under the \pwseth{}, and hence also implied by the SETH, Set Cover Conjecture, and $k$-Orthogonal Vectors Assumption.

\subparagraph*{Domination-like Problems.}
The $(\sigma,\rho)$-dominating set framework \cite{telleComplexityDominationtypeProblems1994,tellePracticalAlgorithmsPartial1993}  generalizes a wide range of classical graph problems by specifying desired neighborhood conditions through sets $\sigma$ and $\rho$ of integers.
This flexible formulation captures many well-known problems including \problem{Dominating Set}, \problem{Independent Set}, and \problem{Perfect Code}, making it a powerful unifying abstraction for studying domination-like properties in graphs. By adjusting $\sigma$ and $\rho$, one can model diverse constraints on how vertices dominate their neighborhoods, allowing for the exploration of both established and novel combinatorial problems within a single framework.

Formally, let  $\sigma$ and $\rho$ be two fixed sets of non-negative integers.
For a graph $G$, a \emph{$(\sigma, \rho)$-set} of $G$ is a set $S \subseteq V(G)$, such that
\begin{itemize}
    \item for all $v \in S$ we have that $|N(v) \cap S| \in \sigma$, and
    \item for all $v \in V(G) \setminus S$ we have that $|N(v) \cap S| \in \rho$.
\end{itemize}
Hence, the sets $\sigma$ and $\rho$ constrain how many selected neighbors selected, respectively unselected, vertices are allowed to have.

Many classical notions can be captured by $(\sigma, \rho)$-sets.
For example, a set $S$ is a dominating set if and only if it is a $(\nat,\mathbb{Z}_{\geq 1})$-set, a $p$-dominating set (i.e., every unselected vertex has at least $p$ selected neighbors) if it is a $(\nat,\mathbb{Z}_{\geq p})$-set, a total dominating set (i.e., selected vertices also need to have a selected neighbor) if and only if it is a $(\mathbb{Z}_{\geq1}, \mathbb{Z}_{\geq 1})$-set, and an independent set if and only if it is a $(\{0\},\nat)$-set.
We refer to Telle \cite{telleComplexityDominationtypeProblems1994} for a more extensive list of properties that can be expressed as $(\sigma, \rho)$-sets.
Problems related to this domination concept have played a role in numerous publications since its inception \cite{bodlaenderFasterAlgorithmsBranch2010,bui-xuanFastDynamicProgramming2013,cattaneoParameterizedComplexityDominationType2014a,chapelleParameterizedComplexityGeneralized2010,etesamiWhenOptimalDominating2019,fominSortSearchExact2009,fominBranchRechargeExact2011,fockeTightComplexityBoundsLowerBound,fockeTightComplexityBoundsUpperBound,golovachComputationalComplexityGeneralized2007,golovachParameterizedComplexityGeneralized2012,greilhuberResidueDominationBoundedTreewidth2025,halldorssonIndependentSetsDomination2000,heggernesPartitioningGraphsGeneralized1998,jaffkeMimwidthIIIGraph2019,masarikFairVertexProblems2025,telleAlgorithmsVertexPartitioning1997,vanrooijFastAlgorithmsJoin2020,vanrooijGenericConvolutionAlgorithm2021,vanrooijDynamicProgrammingTree2009}.

In this work, we are interested in the minimization problem \minGenDomSet{}: given a graph $G$ and integer $k$ decide whether $G$ has a $(\sigma, \rho)$-set of size at most $k$.
Chappelle~\cite{chapelleParameterizedComplexityGeneralized2010} shows that the problem can be solved with single-exponential dependency on the treewidth when $\sigma, \rho$ are ultimately periodic sets, but does not explicitly give or optimize the base of the running time.
Ultimately periodic sets capture, for example, finite and cofinite sets.
As pointed out by Focke et al.~\cite{fockeTightComplexityBoundsUpperBound}, for general cofinite sets, the complexity of \minGenDomSet{} on bounded-treewidth graphs is somewhat mysterious: it depends on the power of representative sets, which is not understood very well (see also \cite{groenlandTightBoundsGraph2024,marxAntifactorFPTParameterized2025}).
Therefore, we restrict our attention to finite or \emph{simple cofinite} sets for which representative sets are not relevant.
A simple cofinite set is of the form $\mathbb{Z}_{\geq c}$ for some integer $c$.
Let us also remark that most known classical graph problems that can be expressed using the $(\sigma,\rho)$-domination framework have sets which are finite or simple cofinite.
In fact, all problems given by Telle \cite[Table I]{telleComplexityDominationtypeProblems1994} as examples for classical problems that the framework can describe have such sets.

To explain the running times of the algorithms in more detail, we first need to define some constants based on $\sigma$ and $\rho$, which correspond to the largest state a dynamic programming algorithm requires for each set.

\begin{definition}[Constants $\sigLargestNonPartial$ and $\rhoLargestNonPartial$]
    Define
    \begin{equation*}
        \sigLargestNonPartial = \begin{cases}
            \max \sigma & \text{if $\sigma$ is finite},          \\
            \min \sigma & \text{if $\sigma$ is simple cofinite},
        \end{cases}
        \text{and }
        \rhoLargestNonPartial = \begin{cases}
            \max \rho & \text{if $\rho$ is finite},          \\
            \min \rho & \text{if $\rho$ is simple cofinite}.
        \end{cases}
    \end{equation*}
    Let $\allLargestNonPartial = \max(\rhoLargestNonPartial,\sigLargestNonPartial)$.
\end{definition}

When $\sigma$ and $\rho$ are both finite or simple cofinite, van Rooij \cite{vanrooijGenericConvolutionAlgorithm2021} shows that \minGenDomSet{} can be solved in time $(\sigLargestNonPartial + \rhoLargestNonPartial + 2)^{\tw} \cdot |G|^{\Oh(1)}$, when a tree decomposition of width $\tw$ is provided together with the input.
This running time is seemingly optimal, as the base corresponds exactly to the number of states in the natural dynamic programming approach to the problem.
However, Focke et al.~\cite{fockeTightComplexityBoundsUpperBound} show that if the sets $\sigma$ and $\rho$ satisfy some additional combinatorial properties, then improved algorithms are possible.

We say that the pair $(\sigma, \rho)$ is \emph{$m$-structured} if $\sigma$ is a subset of some residue class modulo $m$, and $\rho$ is a subset of some residue class modulo $m$.
While every pair $(\sigma, \rho)$ is $1$-structured by definition, Focke et al.~\cite{fockeTightComplexityBoundsUpperBound} show that improvements over the previously best known algorithms are possible when $(\sigma, \rho)$ is $m$-structured for some $m \geq 2$.

\begin{theorem}[Due to results of {\cite{fockeTightComplexityBoundsUpperBound, vanrooijGenericConvolutionAlgorithm2021}}]
    \label{thm:non_partial_upper_bound}
    Let $\sigma, \rho$ be nonempty finite or simple cofinite sets.
    When a tree decomposition of width $\tw$ is provided together with the input, the problem \minGenDomSet{} can be solved in time
    \begin{itemize}
        \item $(\sigLargestNonPartial + \rhoLargestNonPartial + 2)^{\tw} \cdot |G|^{\Oh(1)}$ when $(\sigma, \rho)$ is not $m$-structured for any $m \geq 2$,
        \item $(\max(\sigLargestNonPartial, \rhoLargestNonPartial) + 1)^{\tw} \cdot |G|^{\Oh(1)}$ when $(\sigma, \rho)$ is $m$-structured for some $m \geq 3$, or $(\sigma, \rho)$ is $2$-structured and it is not the case that $\sigLargestNonPartial = \rhoLargestNonPartial$ is even,
        \item $(\sigLargestNonPartial + 2)^{\tw} \cdot |G|^{\Oh(1)}$ otherwise, that is when $(\sigma, \rho)$ is $2$-structured but not $m$-structured for any $m \geq 3$, and $\sigLargestNonPartial = \rhoLargestNonPartial$ is even.
    \end{itemize}
\end{theorem}

Observe that no cofinite set is a subset of a residue class modulo $m$ for any $m \geq 2$, hence the peculiar cases of
\cref{thm:non_partial_upper_bound} only occur when $\sigma, \rho$ are both finite sets.
The algorithms are based on a parity argument that severely restricts the number of combinations of states that can appear in the dynamic programming algorithm when $(\sigma,\rho)$ is $m$-structured.
This argument is robust in the sense that it
works equally well for various optimization and counting versions of the problem.

Are the upper bounds in \cref{thm:non_partial_upper_bound} tight for every $\sigma$ and $\rho$?
Focke et al. \cite{fockeTightComplexityBoundsLowerBound} show that these algorithms are tight under the SETH when $\sigma, \rho$ are finite sets (and $0 \notin \rho$).
Later, the same result was proven under the \pwseth{} in the thesis of Schepper \cite{schepperFasterHigherEasier}.
We complete the picture by extending the lower bound to every finite and simple cofinite set. While this result establishes tight lower bounds for some fundamental problems where no such bounds were known (e.g., \problem{$p$-Dominating Set}, \problem{Independent Dominating Set} --- see also \cite{meybodiParameterizedComplexity12020} for other open questions that we answer), it is secondary to our main contribution about the partial problem variants.
However, we need to have these results for \minGenDomSet{} at hand in order to be able to contrast them with our tight bounds for the partial problems.

\begin{restatable}{maintheorem}{mtheoremLowerBoundNonPartial}
    \label{thm:main_theorem_lower_bound_nonpartial}
    Let $\sigma, \rho$ be nonempty finite or simple cofinite sets with $0 \notin \rho$. Then, unless the \pwseth{} is false, for any $\varepsilon > 0$ and even when the input is provided together with a path decomposition of width $\pw$ there is no algorithm that can solve the problem \minGenDomSet{} in time
    \begin{itemize}
        \item $(\sigLargestNonPartial + \rhoLargestNonPartial + 2 - \varepsilon)^{\tw} \cdot |G|^{\Oh(1)}$ when $(\sigma, \rho)$ is not $m$-structured for any $m \geq 2$,
        \item $(\max(\sigLargestNonPartial, \rhoLargestNonPartial) + 1 - \varepsilon)^{\tw} \cdot |G|^{\Oh(1)}$ when $(\sigma, \rho)$ is $m$-structured for some $m \geq 3$, or $(\sigma, \rho)$ is $2$-structured and it is not the case that $\sigLargestNonPartial = \rhoLargestNonPartial$ is even,
        \item $(\sigLargestNonPartial + 2 - \varepsilon)^{\tw} \cdot |G|^{\Oh(1)}$ otherwise, that is when $(\sigma, \rho)$ is $2$-structured but not $m$-structured for any $m \geq 3$, and $\sigLargestNonPartial = \rhoLargestNonPartial$ is even.
    \end{itemize}
\end{restatable}

\subparagraph*{Partial Optimization Problems.}
For a minimization problem where the goal is to find a solution of at most a certain size that satisfies every constraint, we can consider the natural relaxation where the goal is no longer to satisfy every constraint, but rather a specified portion of it. This relaxation introduces new algorithmic and combinatorial challenges.
Various partial optimization problems have been studied both from the perspective of exact and parameterized algorithms, as well as approximation and heuristics
\cite{bera2014approximation,bshoutyMassagingLinearProgramming1998,charikarAlgorithmsFacilityLocation2001,chekuriAlgorithmsCoveringMultiple2022a,fominSubexponentialAlgorithmsPartial2011a,halperinImprovedApproximationAlgorithms2002,inamdarPartialCoveringGeometric2018,mccutchenStreamingAlgorithmsKCenter2008,slavikImprovedPerformanceGreedy1997,watson-gandyHeuristicProceduresMpartial1982}.
For example, for \problem{Dominating Set} the partial problem means looking for a set of size at most $k$ that dominates at least $\ell$ vertices, or for \problem{Vertex Cover} finding a set of vertices of size at most $k$ that covers at least $\ell$ edges.
These problems known as \problem{Partial Dominating Set} and \problem{Partial Vertex Cover} have received considerable attention \cite{aminiImplicitBranchingParameterized2011,blaserComputingSmallPartial2003,caiParameterizedComplexityCardinality2008,caskurluPartialVertexCover2014,gandhiApproximationAlgorithmsPartial2004,golovachParameterizedComplexityDomination2008,guoParameterizedComplexityGeneralized2005,ishiiSubexponentialFixedparameterAlgorithms2016,kneisPartialVsComplete2007a,koutisLIMITSApplicationsGroup2016,manurangsiNoteMaxKVertex2019}.

Interestingly, the complexity of partial problems and their nonpartial variant can differ significantly.
For example, for the parameter solution size, the \problem{Vertex Cover} problem is one of the quintessential fixed-parameter tractable problems, but \problem{Partial Vertex Cover} is $W[1]$-hard \cite{guoParameterizedComplexityGeneralized2005}, implying that we cannot expect it to be fixed-parameter tractable.
The problem \problem{Partial Dominating Set} parameterized by solution size is easily seen to be $W[2]$-hard, as this already holds for \problem{Dominating Set}.
On the other hand, the input value $\ell$ itself can be seen as a handle towards tractability, and both \problem{Partial Vertex Cover} and \problem{Partial Dominating Set} are fixed-parameter-tractable for parameter $\ell$ \cite{blaserComputingSmallPartial2003,kneisPartialVsComplete2007a}.
In fact, Bl\"aser \cite{blaserComputingSmallPartial2003} shows that even the more general \problem{Partial Set Cover} admits such an algorithm.

\subparagraph*{Our Contribution.}

In this paper, we initiate the study of the natural partial variant of the \minGenDomSet{} problem on graphs of bounded treewidth.
We begin by defining the appropriate notion of domination, which we refer to as satisfying vertices.

\begin{definition}
    Let $\sigma$ and $\rho$ be two sets of non-negative integers.
    Let $G$ be a graph and $S \subseteq V(G)$.
    A vertex $v \in V(G)$ is \emph{satisfied} (by $S$ relative to $(\sigma,\rho)$) if:
    \begin{enumerate}
        \item when $v \in S$, then $|N(v) \cap S| \in \sigma$,
        \item when $v \notin S$, then $|N(v) \cap S| \in \rho$.
    \end{enumerate}
    A vertex $v \in V(G)$ that is not satisfied is called \emph{violated} (by $S$ relative to $(\sigma,\rho)$).
\end{definition}

Next, we formally introduce the partial problem variant \minParGenDomSet{}.

\problembox{\minParGenDomSet{}}{Graph $G$, integers $k$ and $\ell$}{Is there a set $S \subseteq V(G)$ such that $S$ has size at most $k$ and at most $\ell$ vertices are violated by $S$ relative to $(\sigma,\rho)$?}

The problem is equivalent to the problem of finding a set $S$ of size at most $k$ that satisfies \emph{at least} $\ell$ vertices.
We study \minParGenDomSet{} on graphs of bounded treewidth for sets $\sigma$, $\rho$ which are nonempty finite sets, or simple cofinite sets.
Observe that the problem generalizes \minGenDomSet{}, in which one looks for a set that satisfies all vertices.
So, \minParGenDomSet{} is at least as hard as \minGenDomSet{}.
The main question we want to answer is whether this hardness is strict, that is, whether there is a price to pay if we want to solve the more general problem.

\begin{quote}
    \centering
    \emph{For a given $\sigma$ and $\rho$, is the partial problem \minParGenDomSet{} strictly harder on bounded-treewidth graphs than the corresponding \minGenDomSet{}  problem?}
\end{quote}

We answer this question by a complete characterization of the complexity of \minParGenDomSet{} for every nonempty $\sigma$ and $\rho$ that are finite or simple cofinite:  we provide an algorithm and show that it is tight under the \pwseth{}.
Hence, with this work, there is a complete characterization of the complexity of the partial and the nonpartial minimization problems with nonempty finite and simple cofinite sets.
In particular, we provide the first tight lower bounds for the nonpartial and partial variants of \problem{Independent Dominating Set}, $p$-\problem{Dominating Set}, \problem{Perfect Dominating Set}, and \problem{Total Dominating Set}.

Moving from \minGenDomSet{} to \minParGenDomSet{} does not always result in a harder problem.
Let us illustrate this for \problem{Partial Dominating Set}.
We can observe  that, similarly to \problem{Dominating Set}, a dynamic programming algorithm on a tree decomposition needs to consider three states for a vertex $v$: (1) $v$ is unselected with no selected neighbor so far, (2) $v$ is unselected with a selected neighbor, and (3) $v$ is selected (see also \cite{ishiiSubexponentialFixedparameterAlgorithms2016}).
Furthermore, Fast Subset Convolution can be performed similarly to \problem{Dominating Set} \cite{vanrooijDynamicProgrammingTree2009}, resulting in running time $3^\tw\cdot |G|^{\Oh(1)}$ \cite{liuwdominatingSetProblem2021} (earlier, a $4^\tw\cdot |G|^{\Oh(1)}$ algorithm was presented by Roayaei and Razzazi \cite{roayaeiFPTalgorithmModifyingGraph2016a}).
The observation about the number of states can be generalized to the $p$-\problem{Dominating Set} problem, which is \minGenDomSet{} with $\sigma = \nat$ and $\rho = \mathbb{Z}_{\geq p}$.
Both for the partial and nonpartial variant, we need to keep track of $p+2$ states: vertex $v$ is selected, vertex $v$ is unselected with $i=0,\dots,p-1$ neighbors, or vertex $v$ is unselected with at least $p$ neighbors.
With these states it is possible to solve the problem in time $(p + 2)^{\tw} \cdot |G|^{\Oh(1)}$ using fast convolution algorithms.
Note that the important observation that the number of states for $p$-\problem{Dominating Set} is the same in the partial and nonpartial versions was already made by Liu and Lu~\cite{liuwdominatingSetProblem2021}, even though the running time they obtain is worse than  $(p + 2)^{\tw} \cdot |G|^{\Oh(1)}$ due to inefficient handling of join nodes.

The situation changes when at least one of $\sigma$ or $\rho$ is a finite set.
If, say, $\sigma$ is finite, then in the nonpartial problem a selected vertex $v$ can have $ \sigLargestNonPartial+1=\max \sigma+1$ possible states: vertex $v$ can have $0,\dots,\max \sigma$ selected neighbors.
On the other hand, in the partial problem, we need an additional ``overflow'' state to express the case when  the vertex has already at least $\max \sigma+1$ selected neighbors, making it violated (and no amount of further selected neighbors can make it satisfied).
Thus, it seems that whenever $\sigma$ or $\rho$ is finite, we need one extra state and have a corresponding increase in the running time.
We confirm this intuition.

Let us now formally state our main results.
To describe the running time of our algorithm, we need to define some constants based on $\sigma$ and $\rho$, which again have a direct correspondence to the largest states required for $\sigma, \rho$ in the dynamic programming algorithm. The superscript $\sf p$ refers to ``partial'', in order to distinguish the constants from  $\sigLargestNonPartial$ and $\rhoLargestNonPartial$.

\begin{definition}[Constants $\sigLargestPartial$ and $\rhoLargestPartial$]
    Define
    \begin{equation*}
        \rhoLargestPartial = \begin{cases}
            \max \rho + 1 & \text{if $\rho$ is finite},          \\
            \min \rho     & \text{if $\rho$ is simple cofinite},
        \end{cases}
        \text{and }
        \sigLargestPartial = \begin{cases}
            \max \sigma + 1 & \text{if $\sigma$ is finite},          \\
            \min \sigma     & \text{if $\sigma$ is simple cofinite}.
        \end{cases}
    \end{equation*}
    Let $\allLargestPartial = \max(\rhoLargestPartial,\sigLargestPartial)$.
\end{definition}

Our first main result about the partial problems provides an algorithm with a running time depending on these constants that can solve \minParGenDomSet{}.
\begin{restatable}{maintheorem}{mtheoremUpperBoundPartial}
    \label{thm:main_theorem_algorithm}
    Let $\sigma, \rho$ be fixed nonempty finite or simple cofinite sets.
    Then, there is an algorithm that can, when given a graph $G$ and a tree decomposition of $G$ of width $\tw$ as input, decide, for every $k, \ell \in \range[0]{|V(G)|}$ at the same time, whether there exists a set $S \subseteq V(G)$ of size exactly $k$ violating exactly $\ell$ vertices relative to $(\sigma,\rho)$, in time $(\sigLargestPartial + \rhoLargestPartial + 2)^\tw \cdot |G|^{\Oh(1)}$.
\end{restatable}

We provide the proof of \cref{thm:main_theorem_algorithm} in \cref{sec:algorithms}.
Note that the algorithm is a fairly standard dynamic program utilizing fast convolution algorithms.
We then show that under the \pwseth{} these algorithms cannot be improved, unless $0 \in \rho$, in which case the empty set is a trivial optimal solution to the problem.
\begin{restatable}{maintheorem}{mtheoremLowerBoundPartial}
    \label{thm:main_theorem_lower_bound_partial_version}
    Let $\sigma, \rho$ be nonempty finite or simple cofinite sets with $0 \notin \rho$. Then, unless the \pwseth{} is false, there is no algorithm that can solve the problem \minParGenDomSet{} in time $(\sigLargestPartial + \rhoLargestPartial + 2 - \varepsilon)^\pw \cdot |G|^{\Oh(1)}$ for any $\varepsilon > 0$, even when a path decomposition of width $\pw$ is provided together with the input.
\end{restatable}

When both $\sigma$ and $\rho$ are simple cofinite, for example for the $p$-\problem{Dominating Set} problem, we see that the algorithms for the partial and nonpartial variant have the same running time.
Hence, in this case the partial and nonpartial variant are equally hard.

On the other hand, comparing \cref{thm:non_partial_upper_bound} and \cref{thm:main_theorem_lower_bound_partial_version}, we can see that when at least one of $\sigma$, $\rho$ is finite, then the running time for the partial version is strictly higher than the running time for the nonpartial version.
We provide a comparison of the obtained running times for a variety of classical problems in \cref{table:running_time_comp}.

\begin{table}
    \centering
    \begin{tabular}{c|c|c|c|c}
        Problem Name                                           & $\sigma$              & $\rho$                & Nonpartial    & Partial       \\ \hline \hline
        \problem{Dominating Set}                               & $\mathbb{Z}_{\geq 0}$ & $\mathbb{Z}_{\geq 1}$ & $3^\tw$       & $3^\tw$       \\
        $p$-\problem{Dominating Set} ($p \geq 1$)              & $\mathbb{Z}_{\geq 0}$ & $\mathbb{Z}_{\geq p}$ & $(p+2)^\tw$   & $(p+2)^\tw$   \\
        \problem{Perfect Code}                                 & $\{0\}$               & $\{1\}$               & $2^{\tw}$     & $5^{\tw}$     \\
        \problem{Perfect Dominating Set}                       & $\mathbb{Z}_{\geq 0}$ & $\{1\}$               & $3^\tw$       & $4^\tw$       \\
        \problem{Total Dominating Set}                         & $\mathbb{Z}_{\geq 1}$ & $\mathbb{Z}_{\geq 1}$ & $4^\tw$       & $4^\tw$       \\
        \problem{Total Perfect Dominating Set}                 & $\{1\}$               & $\{1\}$               & $2^\tw$       & $6^\tw$       \\
        \problem{Weakly Perfect Dominating Set}                & $\{0,1\}$             & $\{1\}$               & $4^\tw$       & $6^\tw$       \\
        \problem{Independent Dominating Set}                   & $\{0\}$               & $\mathbb{Z}_{\geq 1}$ & $3^\tw$       & $4^\tw$       \\

        \problem{Dominating $p$-Regular subgraph} ($p \geq 0$) & $\{p\}$               & $\mathbb{Z}_{\geq 1}$ & $(p + 3)^\tw$ & $(p + 4)^\tw$ \\ \hline
    \end{tabular}
    \caption{Comparison of the tight running times for the nonpartial and partial minimization problems for some classical problems.
        Polynomial factors in the input size are hidden.
        The bounds for the nonpartial variant stem from \cref{thm:non_partial_upper_bound,thm:main_theorem_lower_bound_nonpartial}, the bounds for the partial variant from \cref{thm:main_theorem_algorithm,thm:main_theorem_lower_bound_partial_version}. Most problems of this table are taken from \cite[Table I]{telleComplexityDominationtypeProblems1994}.}
    \label{table:running_time_comp}
\end{table}

A particularly striking example is the case of $\sigma = \{1\}, \rho = \{1\}$ (also called \problem{Total Perfect Dominating Set}), where the problem \minGenDomSet{} can be solved in time $2^{\tw} \cdot |G|^{\Oh(1)}$ \cite{fockeTightComplexityBoundsUpperBound}, but \minParGenDomSet{} requires time $6^{\tw} \cdot |G|^{\Oh(1)}$.
Hence, one sometimes has to pay a significant price for considering the partial problem variant instead of the nonpartial one.
We identify two main reasons for these differing running times.
\begin{itemize}
    \item \textbf{Overflow states are needed for finite sets.}
          As discussed above, if a set $\tau \in \{\sigma, \rho\}$ is simple cofinite, there is no difference in the state set needed for the partial and nonpartial variants (we have $\tauLargestNonPartial = \tauLargestPartial$).
          On the other hand, if $\tau$ is finite, then the size of the state set for the partial variant is larger than for the nonpartial variant (we have $\tauLargestPartial > \tauLargestNonPartial$).
          This indicates that we have to pay an additional price for each finite set when moving from the nonpartial to the partial variant.
    \item  \textbf{Structured sets are irrelevant.} For $m$-structured sets $(\sigma, \rho)$ (with $m \geq 2$) the algorithm by Focke et al. \cite{fockeTightComplexityBoundsUpperBound} is based on a parity argument  that can significantly reduce the number of subproblems that need to be considered.
          However, if some vertices of the graph are violated in the solution, then this combinatorial argument no longer holds, and we have to fall back on the more naive algorithm.
          Hence, the running time difference is especially large when moving to the partial problem from $m$-structured sets $(\sigma, \rho)$.
\end{itemize}

\cref{thm:main_theorem_lower_bound_partial_version} is proven in \cref{sec:intermediate_lower_bound,sec:realizing_relations} in technical reduction chains, and these proofs represent the main bulk of the work done for this paper.
In \cref{sec:technical_overview}, we present an overview of the used techniques.
The proof of \cref{thm:main_theorem_lower_bound_nonpartial} is given in \cref{sec:nonpartial}.

\section{Technical Overview}
\label{sec:technical_overview}
We will now give an overview of the technical ideas used to obtain our results.

\subsection{The Algorithm.}

For our algorithmic result, we utilize standard dynamic programming procedures, similar to those given in \cite{fockeTightComplexityBoundsUpperBound,greilhuberResidueDominationBoundedTreewidth2025} for other variants of $(\sigma, \rho)$-dominating set.
Using typical notation for these types of problems, we introduce state sets, which will be used for both the algorithm and the lower bound.

\begin{restatable}[State Sets for the Partial Problem]{definition}{defStates}
    \label{def:states_partial}
    Let $\sigma, \rho$ be nonempty finite or simple cofinite sets.
    We define the set of $\rho$-states for the partial problem as $\rhoStatesParital = \{\rho_0,\dots,\rho_{\rhoLargestPartial}\}$,
    and the set of $\sigma$-states for the partial problem as $\sigStatesPartial = \{\sigma_0,\dots,\sigma_{\sigLargestPartial}\}$.
    We denote the set of all states for the partial problem as $\allStatesPartial = \sigStatesPartial \cup \rhoStatesParital$.

    When $\tau$ is a finite set, the state $\tau_{\tauLargestPartial}$ is called an overflow-state.
\end{restatable}
In a partial solution, each vertex is assigned one of these states.
If a vertex is selected it has a $\sigma$-state, if it is unselected it has a $\rho$-state.
The subscript of the state describes the number of selected neighbors the vertex has.
Observe that we have $|\allStatesPartial| = \sigLargestPartial + \rhoLargestPartial + 2$, which forms the base of the running time of \cref{thm:main_theorem_algorithm}.

We proceed by dynamic programming on a nice tree decomposition, for a node $t$, $X_t$ is the bag of $t$, and $V_t$ is the set of vertices introduced in and below node $t$.
For a node $t$ and a vector $\vec{u} \in \allStatesPartial^{X_t}$, we say that $\vec{u}$ has a matching partial solution of size $k$ violating $\ell$ vertices if there is a set $S \subseteq V_t$ such that:
\begin{itemize}
    \item $|S \setminus X_t| = k$, and
    \item $S$ violates exactly $\ell$ vertices of $V_t \setminus X_t$, and
    \item for any $v \in X_t$, $\vec{u}[v]$ is a $\sigma$-state if and only if $v \in S$, and when $\vec{u}[v] = \tau_c$, we have that $c = \min(x, \tauLargestPartial)$, where $x$ is the number of selected neighbors that vertex $v$ has in $V_t \setminus X_t$.
\end{itemize}

For each node $t$, and each possible solution size and violation number, we have a table entry $T_t[k, \ell]$, whose task is to store the set of all vectors of $\allStatesPartial^{X_t}$ that have a matching partial solution for $t$ of size $k$ violating $\ell$ vertices.
Using fast convolution algorithms by van Rooij \cite{vanrooijGenericConvolutionAlgorithm2021}, one can compute these table entries $T_t[k,\ell]$ correctly in time $|\allStatesPartial|^{\tw} \cdot |G|^{\Oh(1)}$, which yields \cref{thm:main_theorem_algorithm}.

Compared to the known algorithms for the nonpartial problem we need to keep track of the number of violated vertices.
While this requires careful bookkeeping, we feel that the main algorithmic contribution we make is establishing that an additional state is needed exactly for each finite set.

\subsection{The Lower Bound}

Next we describe the main framework of the lower bound, which is the far more technical part of our contribution.

\subsubsection{The Three-Step Framework.}
On a very high level, our strategy follows a fairly long line of work in the domain of tight lower bound for dynamic programming algorithms \cite{curticapeanTightConditionalLower2016,fockeTightComplexityBoundsLowerBound,greilhuberResidueDominationBoundedTreewidth2025,marxDegreesGapsTight2021,marxAntifactorFPTParameterized2025}.
In these results, a certain design pattern has emerged, which we also follow.
Typically, the proof is by reduction from \kSAT{}, or an appropriate \problem{CSP} variant, to the target problem. If our goal is to rule out the existence of $(c-\epsilon)^\pw\cdot |V(G)|^{O(1)}$ time algorithms, then it is convenient to start the reduction from a \problem{CSP} problem with domain size $c$, for which Lampis \cite{lampisFinerTightBounds2020} provides a lower bound of the correct form under SETH.

In these reductions, the pathwidth $\pw$ of the output instance needs to be roughly the number $n$ of variables of the input \problem{CSP} instance.
Typically, the construction starts by introducing a $n\times m$ grid of \emph{state-vertices,} where $m$ can be arbitrarily large. Any two consecutive columns of state vertices are connected by low-pathwidth gadgets, ensuring that the constructed graph has pathwidth roughly $n$.

Given a solution of the constructed instance a state is assigned to each state vertex $v$ based on the role that $v$ plays in the solution.
For the problem \minParGenDomSet{} the state is
determined by whether $v$ is in the solution and by the number of selected neighbors it has in the gadgets connecting the column of $v$ to the \emph{previous} column.
Thus, the states of the $n$ vertices in a column can be interpreted as a variable assignment of the source \problem{CSP} instance with $n$ variables and domain size $|\allStatesPartial|$.
To obtain a correct reduction, the gadgets connecting the columns need to enforce two properties:
\begin{enumerate}
    \item Each column must express exactly the same assignment.
    \item The assignment satisfies the input instance.
\end{enumerate}
To ensure the first property, the gadgets between columns $i$ and $i+1$ need to ``copy'' the states of the vertices in column $i$ to column $i+1$.
Furthermore, given that the number $m$ of columns can be made arbitrary large without changing the pathwidth, it is convenient to require that the gadgets between columns $i$ and $i+1$ ensure that the $i$-th clause of the source \problem{CSP} instance is satisfied.
That is, the gadgets between two columns are responsible for copying/propagating states across the construction and for ensuring that a single clause is satisfied.

For every fixed $\sigma$ and $\rho$, we need to be able to construct various gadgets
that are responsible for the tasks described above.
However, when designing gadgets of different types, at some point it is easier to observe and prove that we are able to construct \emph{every} gadget (in a well-defined sense).
Formally, we introduce a new, more general problem where the input contains a set of ``relations'': additional constraints on certain subsets of vertices.
By exploiting the additional power of these relations the lower bound is proved for this more general problem.
Then, the relations are systematically replaced by gadgets.

Instead of having one large reduction, we can use a framework that splits the reduction into three steps, each having a clear individual goal.
Let $\mathcal{P}$ be the problem one wants to show a lower bound for.
\begin{enumerate}
    \item A new problem $\mathcal{P}'$ is defined.
          Problem $\mathcal{P}'$ is a variant of $\mathcal{P}$, in which the input instance contains  list of \emph{relations} that introduce additional constraints a solution must fulfill.
          For example, in the $(\sigma,\rho)$-domination problem, a relation can constrain a subset $X$ of vertices, restricting which combinations of vertices of $X$ can be selected in a solution.
          We refer to problem $\mathcal{P}'$ as the \emph{intermediate problem} or the \emph{problem with relations}.
          A lower bound is shown for problem $\mathcal{P}'$ using a grid-like construction.
    \item In a second step, the relations of the problem $\mathcal{P}'$, which can be complex, are replaced by a set of much simpler relations.
          For example, in many cases, a single type of relation ensuring that exactly one vertex of the constrained set is selected is sufficient to simulate arbitrary relations.
    \item Finally, the simple relations are replaced by gadgets that simulate their behavior in a reduction to problem $\mathcal{P}$. Note that this last step often requires case distinctions based on the concrete problem one considers, as this is the point where gadgets need to be crafted.
\end{enumerate}

This framework was used by Curticapean and Marx\cite{curticapeanTightConditionalLower2016} to show a lower bound for counting the number of perfect matchings of a graph.
Later, the framework was used for edge-selection problems \cite{marxDegreesGapsTight2021,marxAntifactorFPTParameterized2025}.
Finally, Focke et al. \cite{fockeTightComplexityBoundsLowerBound} and Greilhuber et al. \cite{greilhuberResidueDominationBoundedTreewidth2025} utilize this approach for $(\sigma,\rho)$-domination problems.

\subsubsection{Challenges}
In this work, we are using the same framework for our problem. We can profit from the fact that the first two steps are already somewhat well-understood for $(\sigma,\rho)$-domination problems.
However, we still face multiple challenges which make the application of the framework far from straightforward, one of them requires a novel spin on step 1.

\subparagraph*{Challenges in step 1.}
Previous applications of the framework for $(\sigma,\rho)$-domination problems use the same definition for the intermediate problem.
In this definition, a relation is a pair $(\scp{R},\acc{R})$, where $\scp{R} \subseteq V(G)$ is the scope of the relation, and $\acc{R} \subseteq 2^{\scp{R}}$ the set of accepted selections.
A solution $S$ to the problem with relations must not only be a $(\sigma,\rho)$-set, but also for each relation $R$ it must be the case that $S \cap \scp{R} \in \acc{R}$.
In other words, the relations can explicitly forbid certain selections from being allowed in a solution.
A path decomposition of the intermediate problem is a path decomposition of the graph such that additionally for each relation $R$ there exists a bag containing $\scp{R}$.
Moreover, the size of the scope of each relation that is used in the reductions should be bounded by a constant.
These two properties ensure that steps 2 and 3 do not increase the pathwidth too much.

The first major challenge we face is about the difference between the partial and nonpartial problem.
The constructions used for step 1 in the literature \cite{fockeTightComplexityBoundsLowerBound,greilhuberResidueDominationBoundedTreewidth2025} are for the nonpartial problem, and thus a solution cannot violate any vertices.
When dealing with the partial problem, we have an increase in the running times due to the fact that overflow-states are required for finite sets.
We need to ensure that some of our state vertices can have overflow-states in a solution,
otherwise we would not get a higher base in the lower bound than for the nonpartial problem.
Hence, we must carefully create a construction in which some state vertices are violated and correctly propagate their state.
In particular, a priori it is not clear why a violated state vertex should have its state propagate correctly across the construction.

The second major problem is that the intermediate problem defined above is too  powerful when $\rho$ is a simple cofinite set.
It turns out that when $\rho$ is simple cofinite, there are relations which cannot be simulated by gadgets in a straightforward way. Therefore, we need to work with a more restricted problem where $\acc{R}$ is assumed to be superset-closed for each relation $R$. 
We can prove the required lower bound even for this significantly less expressive problem, and the superset-closed relations can be simulated with appropriate gadgets.

\subparagraph*{Challenges in step 2.}
A $\hwRelation$-relation (hamming weight one relation) accepts a selection from its scope if and only if it has size exactly one.  Focke et al. \cite{fockeTightComplexityBoundsLowerBound} showed that (as long as $\rho \not = \{0\}$) there is a reduction from the \minGenDomSet{} with arbitrary relations with bounded arity to the problem using only $\hwRelation$-relations of bounded arity.
While the reduction also works for the partial problem \minParGenDomSet{}, it would not be helpful when $\rho$ is simple cofinite, since we cannot replace the $\hwRelation$ relations by gadgets.
We instead show that superset-closed relations can be replaced by $\hwGeqOneRelation$-relations, where
a $\hwGeqOneRelation$ relation accepts a selection from its scope if and only if it is nonempty.

\subparagraph*{Challenges in step 3.}
In step 3, we need to replace the simple relations ($\hwRelation$ or  $\hwGeqOneRelation$)  we obtain after step 2 with gadgets.
The construction of these gadgets is very specific to the \minParGenDomSet{} problem and highly depends on the properties of the sets $\sigma$ and $\rho$.

The gadgets given in the literature \cite{fockeTightComplexityBoundsLowerBound,greilhuberResidueDominationBoundedTreewidth2025} cannot directly be used here, since they are for the nonpartial variant.
Hence, the first major challenge we face in step 3 is that we must create gadgets that can handle violations.
In the partial setting, a gadget cannot directly force a vertex set to behave a certain way, as the vertices of the gadget could be violated.
Thus, the first step is to define gadgets such that if the behavior is not as desired, then the gadget has to pay a ``penalty'' either in terms of violated vertices or larger solution size.
The correct form of the definition is not obvious: in fact, we need to use two greatly different definitions for the two cases when $\rho$ is finite or simple cofinite (``robust'' and ``fragile'' realizations, see \cref{def:robust_realization,def:fragile_realization}).
In contrast to the gadgets given in \cite{fockeTightComplexityBoundsLowerBound,greilhuberResidueDominationBoundedTreewidth2025}, our gadgets do not have constant size.
They can be large, but their pathwidth is small.

\subsubsection{The Reduction Chains.}

Let us now present the reduction chain of our lower bound in more detail.
Concretely, we will present a ``plan of attack'' which illustrates which definitions and lemmas we use to handle steps 1-3 of the relational framework.
The reduction chain is different depending on whether $\rho$ is finite or simple cofinite.
We focus on the partial problem variant in this overview, ignoring the result in \cref{thm:main_theorem_lower_bound_nonpartial}.

\subparagraph*{Set $\rho$ is finite.}

We begin by covering the case when $\rho$ is finite, as our reduction chain is closer to the known approaches in this scenario.
We explain why this approach fails for simple cofinite $\rho$ immediately after this section about finite $\rho$.

\subparagraph*{Step 1.}
For our intermediate problem, we use the natural partial variant of the problem with relations.
A relation $R$ is a pair $(\scp{R},\acc{R})$ where $\scp{R}$ is a vertex set and $\acc{R} \subseteq 2^{\scp{R}}$.

\problembox{\parGenDomSetRel{}}{Graph $G$, integer $\ell$, set of relations $\mathcal{R}$}{Is there a set $S \subseteq V(G)$ such that at most $\ell$ vertices of $V(G)$ are violated by $S$ relative to $(\sigma,\rho)$, and for each $R \in \mathcal{R}$ we have $S \cap \scp{R} \in \acc{R}$?}

Observe that the problem \parGenDomSetRel{} is not a minimization problem.
We say that an instance of \parGenDomSetRel{} has arity $d$ (or at most $d$), when $|\scp{R}| \leq d$ for each $R \in \mathcal{R}$.
A path decomposition of the problem \parGenDomSetRel{} is a path decomposition such that additionally for each relation $R$ there is a bag that contains $\scp{R}$.

Since \parGenDomSetRel{} is a partial variant of the intermediate problem used in the literature \cite{fockeTightComplexityBoundsLowerBound}, we can propagate the states of \emph{satisfied} state vertices in essentially the same way in our grid-like construction.
However, changes are required since we also need to be able to propagate the states of \emph{violated} state vertices across the construction.
This can be achieved by only allowing certain vectors of states to be used, when done carefully this ensures propagation while keeping the desired properties.
Formally, we show the following intermediate lower bound for \parGenDomSetRel{}.

\begin{restatable}{lemma}{thmDecisionHighLvl}
    \label{thm:decision_high_lvl}
    Let $\sigma, \rho$ be nonempty finite or simple cofinite sets.
    For any $\varepsilon > 0$, there is a constant $d$ depending only on $\varepsilon$, $\sigma$, and $\rho$, such that if there is an algorithm that can solve the problem \parGenDomSetRel{} on instances of arity at most $d$ provided together with a path decomposition of width \pw{} in time $(|\allStatesPartial| - \varepsilon)^\pw \cdot |I|^{\Oh(1)}$, then the \pwseth{} is false.
\end{restatable}

To prove the lemma above, we provide a suitable reduction for each fixed $\varepsilon$.
The bound $d$ on the arity size is required so that we can later replace the relations.

\subparagraph*{Step 2.}
Next, we want to replace the arbitrary relations with simpler relations.
Using a reduction by Focke et al. \cite[Corollary 8.8]{fockeTightComplexityBoundsLowerBound}, which also works for the partial problem variant, we get the following result.\footnote{We remark that \cref{thm:replacing_arbitrary_with_hw_1} indeed only requires $\rho \neq \{0\}$. Even if $0 \in \rho$ the problem remains hard due to the relations in the input instance.}

\begin{restatable}{lemma}{thmReplacingArbitraryHWOne}
    \label{thm:replacing_arbitrary_with_hw_1}
    Let $\sigma, \rho$ be nonempty sets with $\rho \neq \{0\}$.
    There is a pathwidth-preserving reduction from arbitrary instances \parGenDomSetRel{} with arity at most $d$ to instances \parGenDomSetRel{} with arity at most $2^d + 1$ that only use $\hwRelation{}$-relations.
\end{restatable}

\subparagraph*{Step 3.}
Now all that remains for us is to find gadgets that can replace the $\hwRelation{}$-relation.
For the nonpartial problem and finite $\sigma,\rho$, Focke et al. \cite{fockeTightComplexityBoundsLowerBound} show that there are graphs of constant size which do the job.
That is, to replace a relation $R$ on the vertices $\scp{R}$, they add a small gadget to the graph that is connected to $\scp{R}$.
This gadget does not provide any further selected vertices to $\scp{R}$ in any $(\sigma,\rho)$-set, but it ensures that exactly one vertex of $\scp{R}$ must be selected.
Intuitively, this can be achieved by ensuring that $\scp{R}$ receives a further unselected neighbor with $\max \rho - 1$ selected neighbors, and a further unselected neighbor with $\min \rho - 1$ selected neighbors.
Together, these neighbors guarantee that exactly one vertex of $\scp{R}$ is selected.
The task of the added gadgets is to ensure that the new neighbors of $\scp{R}$ are indeed unselected, and that they have these $\max \rho - 1$, respectively $\min \rho - 1$ selected neighbors.
Observe that this approach exploits that $\rho$ has a maximum value.
Since we are in the case of finite $\rho$, we can also use this fact.
However, we cannot directly use the gadgets given in the literature, as they do not function correctly when vertices can be violated.

To overcome this, we need gadgets that are robust against violations.
Intuitively, this can be achieved by ensuring that if the gadget does not behave correctly, then there is some additional price to pay.
We formally define the gadgets we utilize next.

\begin{restatable}[Robust Realization]{definition}{defRobustRealization}
    \label{def:robust_realization}
    Let $\sigma, \rho$ be sets of non-negative integers.
    A pair $(G,U)$, where $G$ is a graph and $U \subseteq V(G)$ realizes the $\hwRelation{}$-relation of arity $|U|$ with penalty $\penaltySym$, cost tradeoff $\beta$, and cost $\gamma$ if the following properties all hold:
    \begin{enumerate}
        \item For any $u \in U$, there exists a set $S_u \subseteq V(G)$ such that $S_u \cap U = \{u\}$, $S_u \cap N(U) = \emptyset$, $|S_u \setminus U| = \gamma$, and $S_u$ violates no vertex of $V(G) \setminus U$.
        \item For any set $S \subseteq V(G)$ that violates at most $\ell$ vertices of $V(G) \setminus U$, it holds that $S \setminus U$ has size at least $\gamma - \ell \cdot \beta$.
        \item Any set $S \subseteq V(G)$ such that $|S \cap U| \not = 1$ violates at least one vertex of $V(G) \setminus U$ or fulfills $|S \setminus U| > \gamma$.
        \item Any set $S \subseteq V(G)$ such that $S \cap N(U) \not = \emptyset$ violates more than $\penaltySym$ vertices of $V(G) \setminus U$ or fulfills $|S \setminus U| > \gamma + \penaltySym$.
        \item $U$ is an independent set of $G$.
    \end{enumerate}
\end{restatable}

When replacing a $\hwRelation{}$-relation the relation scope is identified with the set $U$ of a fresh copy of the robust realization gadget.
So, the goal of such a gadget is ensuring that exactly one vertex of $U$ is selected, and that no vertex of $N(U)$ is selected.
Let us elaborate on the definition.
First, there is a small solution (selecting at most $\gamma$ vertices outside $U$), whenever exactly one vertex $u$ of $U$ is selected.
The solution also has the property that it selects no vertex of $N(U)$.
This guarantees that whenever exactly one vertex of $U$ is selected we can extend the selection of $U$ to this solution of the robust realization gadget without changing the number of selected neighbors of vertices of $U$.
On the other hand, if a solution is not of this form, then we have to pay for it either in terms of larger solution size or number of violated vertices.
Specifically, the tradeoff value $\beta$ of the second property guarantees that decreasing the solution size is possible only at the cost of some number of violated vertices.
The third property says that if not exactly one vertex of $U$ is selected, then either there is a violation or the solution size strictly increases. The fourth property states that selecting a vertex in the neighborhood of $U$ is even worse: it creates more than $\delta$ violations or increases solution size by more than $\delta$.\footnote{Note that in principle the idea of using gadgets that simulate relations with some notion of penalty for undesired behaviors is not new.
    For instance, Marx et al. \cite{marxDegreesGapsTight2021} use this approach for a class of edge selection problems.}

We show that these types of gadgets exist for arbitrarily large penalties $\delta$.
Our constructions can be seen as extensions of those given in the literature \cite{fockeTightComplexityBoundsLowerBound} that take solution sizes into account and can also handle violations.

\begin{restatable}{lemma}{thmRobustRelationGadgetsExist}
    \label{thm:robust_relations_exist}
    Let $\sigma$ be a nonempty finite or simple cofinite set, and $\rho$ be a finite set with $0 \notin \rho$.
    Then, there is a non-negative constant $\beta$, such that for any $\penaltySym,d > 0$ there exists a pair $(G,U)$ that realizes the $\hwRelation{}$-relation of arity $d = |U|$ and cost $\gamma = \Oh(\penaltySym)$,
    with selection penalty $\delta$ and cost tradeoff $\beta$.
    Moreover, such a pair $(G,U)$ together with a path decomposition of $G$ of width $\Oh(d)$ can be computed in time polynomial in $d + \penaltySym$.
\end{restatable}

Now, we have all tools at our disposal for our final reduction to replace the $\hwRelation$-relations.

\begin{restatable}{lemma}{thmRhoFinitePartialRealizeRelationsReduction}
    \label{thm:rho_finite_partial_realize_relations_reduction}
    Let $\sigma$ be a nonempty finite or simple cofinite set, and $\rho$ be a finite set with $0 \notin \rho$.
    There is a pathwidth-preserving reduction from \parGenDomSetRel{} using only $\hwRelation{}$ relations of arity bounded by a constant $d$ to \minParGenDomSet{}.
\end{restatable}
\begin{proof}[Proof Sketch.]
    Let the input instance of \parGenDomSetRel{} be $(G,\ell,\mathcal{R})$.
    We utilize the robust realization gadgets for a penalty value of $\delta = |V(G)| + \ell + \ell \cdot \beta + 1$, where $\beta$ is the cost tradeoff value.
    Let $\gamma$ be the cost constant of the gadgets.
    For any relation $\mathcal{R}$ we add $\delta$ copies of the robust realization gadget (with $|U| = |\scp{R}|$) and identify the set $U$ of each gadget with $\scp{R}$.
    For such a gadget $(H,U)$ we refer to the vertices of $V(H) \setminus U$ as the internal vertices of the gadget.
    The solution size of the output instance is $k = |V(G)| + |\mathcal{R}| \cdot \delta \cdot \gamma$, and the output instance is $(G',k,\ell)$, where $G'$ is the constructed graph.

    The forward direction of correctness is straightforward.
    By the first property of the robust realization gadgets, there is a suitable extension within such a gadget that adds no further violations to the gadget and selects exactly $\gamma$ internal vertices of the gadget.
    We utilize this extension in each of the added gadgets.

    Now, let $S$ be a solution to the output instance.
    Due to the tradeoff value $\beta$ the set $S$ must select at least $|\mathcal{R}| \cdot \delta \cdot \gamma - \ell \cdot \beta$ vertices within the added gadgets.
    Then, the fourth property of the gadget and our chosen values for $\delta$ and $k$ ensure that is not possible that a vertex of $N(U)$ is selected for the set $U$ of any added gadget.
    This guarantees that $S \cap V(G)$ violates at most $\ell$ vertices of $G$.
    Recall that we added $\delta$ gadgets to replace each relation $R$.
    By the third property of the gadgets, either exactly one vertex of $\scp{R}$ is selected, or each of these $\delta$ gadgets has more than $\gamma$ selected internal vertices, or a violated vertex.
    Our choices of $k$ and $\delta$ mean that also this is impossible, this would either yield more than $\ell$ violated vertices, or that $|S| > k$.

    Finally, recall that a path decomposition of the input instance is such that for each $R \in \mathcal{R}$ there is a bag containing $\scp{R}$.
    To fit a gadget $H$ for $R$ into such a path decomposition, we can copy the bag $X_R$ that contains $\scp{R}$ arbitrarily many times, and then insert the path decomposition of the gadget $H$ into these copies.
    This will only increase the width of the resulting path decomposition by the width of $H$, which is a constant.
    We can repeat this step for each relation while only increasing the pathwidth by a constant overall.
\end{proof}

Combining \cref{thm:decision_high_lvl,thm:replacing_arbitrary_with_hw_1,thm:rho_finite_partial_realize_relations_reduction} now yields the result of \cref{thm:main_theorem_lower_bound_partial_version} for the case where $\rho$ is finite.

\subparagraph*{Set $\rho$ is simple cofinite.}
Unfortunately, we cannot use this approach for the case when $\rho$ is simple cofinite.
The problem is that it is \emph{impossible} to build robust realization gadgets for relations of arity 2 or larger when $\rho$ is simple cofinite.
Indeed, by the first property of these gadgets $(G,U)$, there is a set $S$ selecting exactly one vertex of $U$, no vertex of $N(U)$, and exactly $\gamma$ vertices of $V(G) \setminus U$ that violates no vertex of $V(G) \setminus U$.
Because $\rho$ is simple cofinite, we could simply take $S$, and add all vertices of $U$ to it.
This set $S'$ would then still satisfy all vertices of $V(G) \setminus U$ because the unselected vertices of $N(U)$ cannot become violated due to having even more selected neighbors, and clearly $S'$ still selects exactly $\gamma$ vertices of $V(G) \setminus U$.
However, this violates the third property of the definition of these gadgets, which states that any set that selects more than one vertex of $U$ must violate a vertex of $V(G) \setminus U$ or select more than $\gamma$ vertices of $V(G) \setminus U$.
Hence, we need a new approach to handle the case when $\rho$ is simple cofinite. 
We present the central ideas of this approach next.

\subparagraph*{Step 1.}

Since we are unable to later replace $\hwRelation{}$-relations with gadgets, we cannot work with arbitrary relations for the intermediate problem.
Hence, we need to restrict the types of relations we use.
A relation $R$ is superset-closed if for any $r \in \acc{R}$ also any superset $r' \subseteq \scp{R}$ of $r$ is in $\acc{R}$.
We only allow such relations in the intermediate problem.

\problembox{\minParGenDomSetRel{}}{Graph $G$, integers $k, \ell$, set of superset-closed relations $\mathcal{R}$}{Is there a set $S \subseteq V(G)$ such that $|S| \leq k$ and at most $\ell$ vertices of $V(G)$ are violated by $S$ relative to $(\sigma,\rho)$, and for each $R \in \mathcal{R}$ we have $S \cap \scp{R} \in \acc{R}$?}

Even when only allowing superset-closed relations, we still cannot easily replace them with gadgets.
We define a property of problem instances that we additionally exploit to replace superset-closed relations.

\begin{restatable}[Good Instances]{definition}{defGoodInstances}
    \label{def:partial_good_instance}
    Let $(G,k,\ell,\mathcal{R})$ be an instance of the problem \minParGenDomSetRel{}.
    The instance is \emph{good} if any set $S \subseteq V(G)$ that fulfills all relations of $\mathcal{R}$ fulfills
    \begin{enumerate}
        \item $|S| \geq k$, and
        \item there is a set $D \subseteq S$ of size at least $\ell$, such that any $v \in D$ has more than $\max \sigma$ neighbors that are also selected by $S$.
    \end{enumerate}
    Note that if $\sigma$ is cofinite, then $\max \sigma=\infty$ and hence the second condition can be satisfied only if $\ell=0$, in which case the condition is vacuously true.
\end{restatable}

We then show the following lower bound for the intermediate problem.
\begin{restatable}{lemma}{thmMinimizationHighLvl}
    \label{thm:minimization_high_lvl}
    Let $\sigma$ be a nonempty finite or simple cofinite set, and $\rho$ be a simple cofinite set.
    For any $\varepsilon > 0$, there is a constant $d$, such that if there is an algorithm that can solve the problem \minParGenDomSetRel{} on good instances $I$ of arity at most $d$ provided with a path decomposition of width \pw{} in time $(|\allStatesPartial| - \varepsilon)^\pw \cdot |I|^{\Oh(1)}$, then the \pwseth{} is false.
\end{restatable}

In the reduction we utilize that superset-closed relations can be controlled well for the minimization problem.
Indeed, using the bound on the solution size for the minimization variant, one can ensure that the fact that the relations of \minParGenDomSetRel{} are superset-closed is not problematic: in the instance we create, a solution of size $k$ can afford to satisfy each relation using only an inclusionwise-minimal satisfying subset.
That the output instance of our reduction is good is then achieved by employing relations.\footnote{
    \cref{thm:minimization_high_lvl} also works if $\rho = \nat$, that is, even when $0$ is in the set $\rho$.
    This is in stark contrast to the problem without relations, where the empty set is always an optimal solution.
}

\subparagraph*{Step 2.}
Next, we show that we can reduce to problems that use simpler relations.
A $\hwGeqOneRelation{}$-relation accepts a selection from its relation scope if and only if it is nonempty.

\begin{restatable}{lemma}{thmReplaceSupersetByHwGeqOne}
    \label{thm:partial_replace_superset_by_hw_geq_1}
    There exists a pathwidth-preserving reduction from the problem \minParGenDomSetRel{} on instances with arity bounded by some constant $d$ to the problem \minParGenDomSetRel{} in which each used relation is a $\hwGeqOneRelation{}$-relation of arity at most $d$.
    Moreover, if the input instance is a good instance, then the output instance is a good instance.
\end{restatable}

\subparagraph*{Step 3.}
Finally, we replace $\hwGeqOneRelation$-relations with appropriate gadgets.
Ideally, we would like to have gadgets which are similar to the robust realization gadgets we used for the case when $\rho$ is finite.
In particular, the gadgets should have some vertex set $U$ such that selecting vertices of $N(U)$ can be made arbitrarily bad.
Unfortunately, it is not possible to create such gadgets.
For example, if $\sigma$ and $\rho$ are simple cofinite, then selecting a vertex $v$ compared to not selecting it is never bad for the neighbors of $v$.
Hence, we cannot have a large penalty for selecting vertices of $N(U)$.
On the flip side, we can at least ensure that selecting a vertex of $N(U)$, or no vertex of $U$ is somewhat bad:
either a vertex of the gadget that is not in $U$ is violated, or the solution size is larger than it needs to be.
This intuition is captured by the next definition.

\begin{restatable}[Fragile Realization]{definition}{defFragileRealization}
    \label{def:fragile_realization}
    Let $\sigma$ and $\rho$ be sets of non-negative integers.
    A pair $(G,U)$, where $G$ is a graph and $U \subseteq V(G)$ fragilely realizes the $\hwGeqOneRelation{}$-relation of arity $|U|$ with cost $\gamma$, if it has all the following properties:
    \begin{enumerate}
        \item Any set $S \subseteq V(G)$ that violates no vertex of $V(G) \setminus U$ fulfills $|S \setminus U| \geq \gamma$.
        \item For any nonempty $U' \subseteq U$, there is a set $S \subseteq V(G)$ such that $S \cap U = U'$, $S$ violates no vertex of $V(G) \setminus U$, $S \cap N_G(U) = \emptyset$, and $|S \setminus U| = \gamma$.
        \item Any set $S \subseteq V(G)$ which selects no vertex of $U$, or selects a vertex of $N_G(U)$ violates a vertex of $V(G) \setminus U$, or has the property that $|S \setminus U| > \gamma$.
        \item The set $U$ is an independent set of $G$.
    \end{enumerate}
\end{restatable}

We construct such gadgets (\cref{thm:fragile_realization_gadgets_exist}), and then use them in the following reduction.
Note that the reduction relies on the input instance being good, otherwise the much weaker guarantees provided by a fragile realization gadget would not be sufficient.

\begin{restatable}{lemma}{thmPartialRhoCofiniteReplaceRelations}
    \label{thm:partial_rho_cofinite_replace_relations}
    Let $\sigma$ be a nonempty set, and $\rho$ a simple cofinite set with $0 \notin \rho$.
    There is a pathwidth-preserving reduction from good instances of \minParGenDomSetRel{} with arity bounded by some constant $d$ using only $\hwGeqOneRelation{}$-relations to \minParGenDomSet{}.
\end{restatable}
\begin{proof}[Proof Sketch.]
    Let $(G,k,\ell,\mathcal{R})$ be the input instance of \minParGenDomSetRel{}, and let $\gamma$ be the cost of the fragile realization gadget.
    Set $t = \ell + \ell \cdot \gamma + k + 1$, and $k' = k + \gamma \cdot |\mathcal{R}| \cdot t$.
    We replace each relation of $\mathcal{R}$ with $t$ fragile realization gadgets by adding the gadgets to the graph and identifying the vertices $U$ of the gadgets with the vertices in the relation scope.
    For a fixed realization gadget $(H,U)$ we call the vertices of $V(H) \setminus U$ the internal vertices of the gadget.
    Let $G'$ be the resulting graph.
    The output instance is $(G',\ell,k')$.

    The forward direction of correctness follows from property 2 of realization gadgets, which allows us to extend solutions for the input instance to solutions to the output instance.

    For the backwards direction consider a solution $S'$ of the output instance.
    By the third property of the fragile realization, each time a selection is not the intended selection, the solution either violates an internal vertex of the gadget, or selects more than $\gamma$ internal vertices.
    Fix a relation $R \in \mathcal{R}$.
    Since at most $\ell$ vertices can be violated by $S'$, $S'$ must select at least $\gamma \cdot (t \cdot |\mathcal{R}| - \ell)$ internal vertices of the added gadgets.
    In particular, by the budget $k'$, at most $\ell \cdot \gamma + k$ gadgets for relation $R$ can be such that $S'$ selects more than $\gamma$ internal vertices of them.
    But, then at most $\ell + \ell \cdot \gamma + k = t - 1$ gadgets for $R$ have a violated internal vertex, or are such that the solution selects more than $\gamma$ internal vertices.
    Hence, by the third property of the realization gadgets, at least one gadget behaves as desired, and at least one vertex of each relation scope is selected by $S'$.

    However, it would be problematic if the vertices of the input graph receive additional selected neighbors from the attached gadgets.
    This is where the fact that the input instance is good comes in.
    From this property we obtain that the solution cannot afford to violate a single internal vertex of the added gadgets, and that it must select exactly $\gamma$ internal vertices of each gadget.
    By the third property of the realization gadgets this implies that not only must all relations of $\mathcal{R}$ be satisfied by $S'$, but also that none of the vertices in the input instance receive additional neighbors from any of the added gadgets.
    Hence, $S' \cap V(G)$ is a solution to the input instance.
    The reduction is pathwidth-preserving because the gadgets we add have constant pathwidth, and are attached to vertices which appear in the same bag of a path decomposition of the input graph.
\end{proof}

By combining \cref{thm:minimization_high_lvl,thm:partial_replace_superset_by_hw_geq_1,thm:partial_rho_cofinite_replace_relations}, we obtain \cref{thm:main_theorem_lower_bound_partial_version} when the set $\rho$ is simple cofinite.

\section{Preliminaries}
\label{sec:preliminaries}
For integers $a,b$ we write $\range[a]{b}$ for the set $\{a \leq x \leq b \mid x \in \nat\}$.
Let $\tau$ be a set of integers, and $k$ an integer.
Then, $\tau + k = \{x + k \mid x \in \tau\}$.
A set of integers $\tau$ is simple cofinite if $\tau = \{x \mid x \in \mathbb{Z}, x \geq c\}$ for some integer $c$, in other words, when $\tau = \mathbb{Z}_{\geq c}$.

\subsection{Graphs}

We use standard notation.
All graphs considered in this work are finite and contain no multi-edges or self-loops.
Given a graph $G$, and a set $U \subseteq V(G)$, we denote the graph induced by $U$ as $G[U]$.
For a graph $G$, and a set $U \subseteq V(G)$, we use $G - U$ to denote the graph $G[V(G) \setminus U]$.
For a graph $G$ and a vertex $v$, the neighborhood of $v$ in $G$ is $N_G(v) = \{u \mid uv \in E(G)\}$.
Given a set of vertices $S$, the neighborhood of this set is $N_G(S) = \bigcup_{v \in S} N_G(v) \setminus S$.
The closed neighborhood of a vertex $v$ is $N_G[v] = N_G(v) \cup \{v\}$.
The closed neighborhood of a vertex set $S$ is $N_G[S] = N_G(S) \cup S$.
We may drop the subscript $G$ if it is clear from the context which graph is meant.

\subsection{Treewidth and Pathwidth}

We introduce the standard notions of treewidth and pathwidth, which are, for example, also introduced in \cite[Chapter 7]{cyganParameterizedAlgorithms2015}.

\begin{definition}[Tree Decomposition, Treewidth]
    \label{def:treewidth}
    Let $G$ be a graph.
    A tree decomposition of $G$ is a pair $(T,\{X_t\}_{t \in V(T)})$, where $T$ is a tree and for any $t \in V(T)$, $X_t$, called the bag of $t$, is a subset of $V(G)$, which fulfills the following properties:
    \begin{enumerate}
        \item For any $v \in V(G)$ there exists a $t \in V(T)$ such that $v \in X_t$.
        \item For every edge $uv \in E(G)$ there exists a $t \in V(T)$ such that $\{u,v\} \subseteq X_t$.
        \item For every vertex $v \in V(G)$, let $T_v = \{t \in V(t) \mid v \in X_t\}$. Then, the graph $T[T_v]$ must be connected.
    \end{enumerate}
    The width of the tree decomposition is $\max_{t \in V(T)} |X_t| - 1$.
    The treewidth of $G$ is the minimum width of any tree decomposition of $G$.
\end{definition}

\begin{definition}[Path Decomposition, Pathwidth]
    \label{def:pathwidth}
    Let $G$ be a graph.
    A path decomposition of $G$ is a tree decomposition $(T,\{X_t\}_{t \in V(T)})$ where $T$ is a path graph.
    As the structure given by $T$ is just an order of the nodes, one can equivalently treat a path decomposition as a sequence of bags $(X_1,\dots,X_r)$.
    The pathwidth of a graph $G$ is the minimum width of any path decomposition of $G$.
\end{definition}

Note that for any problem we consider, we utilize the standard assumption that a tree decomposition (path decomposition) of width $\tw$ ($\pw$) is provided together with the input.
For the purpose of our dynamic programming algorithms, we use the notion of a nice tree decomposition (see, for example, \cite[Chapter 7]{cyganParameterizedAlgorithms2015}).

\begin{definition}[Nice Tree Decomposition]
    A tree decomposition $(T,\{X_t\}_{t \in V(t)})$ is nice if it is rooted at some $r \in V(T)$, $X_r = \emptyset$, and each node is one of the following types:
    \begin{description}
        \item[Leaf node:] A leaf $t$ of $T$ where $X_t = \emptyset$.
        \item[Introduce node:] A node of $t$ with exactly one child $t'$ such that $X_t = X_{t'} \cup \{v\}$ for some vertex $v \notin X_{t'}$.
        \item[Join node:] A node $t$ with exactly two children $t'$ and $t''$ and $X_t = X_{t'} = X_{t''}$.
        \item[Forget node:] A node $t$ with exactly one child $t'$ such that $X_t = X_{t'} \setminus \{v\}$ for some vertex $v \in X_{t'}$.
    \end{description}
\end{definition}

Observe that we can assume that the root is a join node or forget node.
It is well-known that one can compute a nice tree decomposition quickly given a tree decomposition.

\begin{lemma}[{\cite[Lemma 7.4]{cyganParameterizedAlgorithms2015}}]
    \label{thm:transform_into_nice_decomposition}
    Given a tree decomposition $(T,\{X_t\}_{t \in V(t)})$ of a graph $G$ of width at most $k$, one can compute a nice tree decomposition of $G$ of width at most $k$ that has $\Oh(k \cdot |V(G)|)$ nodes in time $\Oh(k^2 \cdot \max(|V(T)|,|V(G)|))$.
\end{lemma}

\subsection{Reductions}

We now introduce the type of reduction we utilize.

\begin{definition}[Parameter-Preserving Reduction]
    A parameter-preserving reduction from parameterized problem $\mathcal{A}$ to parameterized problem $\mathcal{B}$ is a polynomial-time algorithm that, when given an instance $(I,k)$ of problem $\mathcal{A}$, outputs an equivalent instance $(I',k')$ of problem $\mathcal{B}$, where $k' = k + \Oh(1)$.

    In particular, if the parameter of problems $\mathcal{A}$ and $\mathcal{B}$ is pathwidth, we refer to such a reduction as a pathwidth-preserving reduction, and we always assume that the input instance is provided with a path decomposition of width $\pw$, and then the reduction has to produce a path decomposition of width $\pw + \Oh(1)$ of the output instance in polynomial-time.\footnote{This definition is similar to \cite[Definition 3.1]{fockeTightComplexityBoundsLowerBound}, although we use many-one reductions instead of general Turing reductions.}
\end{definition}

We now state a simple observation about these special types of reductions.
\begin{observation}
    Assume that there is a parameter-preserving reduction from $\mathcal{A}$ with parameter $k$ to parameterized problem $\mathcal{B}$ with parameter $k'$.
    If there is an algorithm that can solve instances $I'$ of $\mathcal{B}$ in time $c^{k'} \cdot |I'|^{\Oh(1)}$ for some constant $c$, then one can also solve instances $I$ of $\mathcal{A}$ in time $c^{k} \cdot |I|^{\Oh(1)}$.
\end{observation}

In particular, if some conjecture implies that there is no algorithm that can solve problem $\mathcal{A}$ in time $c^k \cdot |I|^{\Oh(1)}$,
then the conjecture implies that there is no algorithm that can solve problem $\mathcal{B}$ in time $c^{k'} \cdot |I|^{\Oh{(1)}}$.

\subsection{States}

We now describe sets of states which we utilize both for the upper and lower bound.
Note that these types of state sets are standard in the literature for $(\sigma,\rho)$-dominating set problems.

\defStates*{}

Generally, each vertex in some (partial) solution can be associated with one of these states.
The state of a vertex has the symbol $\sigma$ if the vertex is selected, and the symbol $\rho$ otherwise.
The subscript represents how many selected neighbors the vertex has.
For a set $\tau$, the state $\tau_{\tauLargestPartial}$ is a shorthand for all the states $\tau_{\tauLargestPartial}, \tau_{\tauLargestPartial+1},\dots$, which behave identically (as a vertex with any those states is either satisfied or violated, depending on whether $\tau$ is finite or simple cofinite).
Note that for a state $\sigma$ and $\rho$ refer to the selection status of the vertex, and not to the sets we use for the problem.
So even if the sets $\sigma$ and $\rho$ are identical, the states $\sigma_c$ and $\rho_c$ are still different.

Also observe that $|\allStatesPartial| = \rhoLargestPartial + \sigLargestPartial + 2$, and hence the running times of our algorithm and lower bounds in \cref{thm:main_theorem_algorithm,thm:main_theorem_lower_bound_partial_version} can equivalently be expressed in terms of $|\allStatesPartial|$ rather than $\rhoLargestPartial$ and $\sigLargestPartial$.
 
\section{The Algorithm}
\label{sec:algorithms}
In this section, we cover the upper bounds of the paper.
We begin by providing a simple observation about the problem.

\begin{observation}
    \label{obs:zero_in_rho_trivial}
    If $0 \in \rho$, then $\emptyset$ is a set of minimum size that violates no vertices.
\end{observation}

Hence, all cases which are interesting for the problem \minParGenDomSet{} fulfill $\min \rho \geq 1$.
Before giving the details of the algorithm we use to handle the interesting cases, we define a special addition operation, and explain how we can utilize fast convolution algorithms by van Rooij~\cite[Theorem 2]{vanrooijGenericConvolutionAlgorithm2021} in our setting.

For any two $\rho$-states $\rho_x$ and $\rho_y$ we define $\rho_x \oplus \rho_y = \rho_{\min(x + y, \rhoLargestPartial)}$, similarly, for any two $\sigma$-states $\sigma_x$ and $\sigma_y$ we define $\sigma_x \oplus \sigma_y = \sigma_{\min (x + y, \sigLargestPartial)}$.
The operation $\oplus$ is not defined if the states are not both of the same type.
The operator $\oplus$ to join states incorporates the natural behavior that two states should only be joined whenever they are either both a $\sigma$-state or both a $\rho$-state, and if that is fulfilled, the sum should use capped addition to compute the resulting state.

Note that with these operators, the sets $\rhoStatesParital$ and $\sigStatesPartial$ can be seen as independent copies of the set $\range[0]{\rhoLargestPartial}$, respectively $\range[0]{\sigLargestPartial}$, and the set $\allStatesPartial$ as the disjoint union of these independent copies.
The framework of van Rooij \cite{vanrooijGenericConvolutionAlgorithm2021} for fast generic convolutions considers convolution operations over various domains and operations $\oplus$. In particular, ``addition with maximum value'' over disjoint unions of integers explicitly appears in this framework, covering our operation $\oplus$ over $\allStatesPartial$. Therefore, we obtain the following by applying \cite[Theorem 2]{vanrooijGenericConvolutionAlgorithm2021}.

\begin{lemma}[Follows from {\cite[Theorem 2]{vanrooijGenericConvolutionAlgorithm2021}, \cite[Main Theorem I.1]{schepperFasterHigherEasier}}]
    \label{thm:fast_convolution}
    Let $\sigma, \rho$ be fixed nonempty finite or simple cofinite sets and $\allStatesPartial$ the corresponding set of states for the partial problem.
    There is an algorithm that takes explicitly given functions
    $f,g: \allStatesPartial^k \times \range[0]{n} \rightarrow \{0,1\}$ as input and computes function
    \begin{align*}
        h(c, \kappa) = \sum_{c_f \oplus c_g = c} \sum_{\kappa_f + \kappa_g = \kappa} f(c_f,\kappa_f) g(c_g, \kappa_g),
    \end{align*}
    in $|\allStatesPartial|^{k} \cdot (n + k)^{\Oh(1)}$ arithmetic operations.
\end{lemma}

Note that the $O(1)$ term in the running time of \cref{thm:fast_convolution} is allowed to depend on $\sigma, \rho$ since we treat these sets as fixed.
Now, we are ready to state our dynamic programming algorithm.
\mtheoremUpperBoundPartial*
\begin{proof}
    This dynamic program is fairly standard, and in particular resembles the one used by Focke et al.~\cite{fockeTightComplexityBoundsUpperBound}, Greilhuber et al.~\cite{greilhuberResidueDominationBoundedTreewidth2025}, and utilizes ideas of van Rooij~\cite{vanrooijGenericConvolutionAlgorithm2021}.

    \subparagraph*{The algorithm.}

    We begin by transforming the given tree decomposition into a nice  tree decomposition with the same width in polynomial time
    using \cref{thm:transform_into_nice_decomposition}.
    So, we shall assume that the used tree decomposition is nice for the remainder of the proof.

    Given a node $t$ of the tree decomposition, we let $V_t$ be the set of vertices introduced in the bags of the nodes under and including $t$, and we denote the bag of $t$ by $X_t$.
    For a node $t$ of the tree decomposition and a vector $\vec{u} \in \allStatesPartial^{X_t}$ (where we index positions of $\vec{u} $ by vertices of $X_t$), we say that $\vec{u} $ has a matching partial solution of size $k$ violating $\ell$ vertices if there is a set $S \subseteq V_t$ such that:
    \begin{itemize}
        \item $|S \setminus X_t| = k$, and
        \item $S$ violates exactly $\ell$ vertices of $V_t \setminus X_t$, and
        \item for any $v \in X_t$, $\vec{u}[v]$ is a $\sigma$-state if and only if $v \in S$, and when $\vec{u}[v] = \tau_c$, we have that $c = \min(x, \tauLargestPartial)$, where $x$ is the number of selected neighbors that vertex $v$ has in $V_t \setminus X_t$.
    \end{itemize}

    For any node $t$ of the tree decomposition and any $k, \ell \in \range[0]{|V(G)|}$, we want $T_t[k,\ell]$ to contain exactly those vectors of $\allStatesPartial^{X_t}$ which have a matching partial solution of size $k$ violating $\ell$ vertices.
    We proceed in a bottom-up fashion along the rooted nice tree decomposition, and make different computations depending on which type of vertex the considered node is.

    \begin{description}
        \item[Leaf nodes:] In leaf nodes $t$, we set $T_t[0,0] = \{\varepsilon\}$ (where $\varepsilon$ denotes the empty string) and $T_t[x,y] = \emptyset$ if $(x,y) \not = (0,0)$.

        \item[Introduce nodes:] Let $t$ be a introduce node and $t'$ be the unique child of $t$. Moreover, let $\{v\} = X_t \setminus X_{t'}$.
              For a vector $\vec{u} \in \allStatesPartial^{X_{t'}}$, for $\tau \in \{\sigma,\rho\}$, we define the $\tau$-extension as the unique string $\vec{u}_\tau$ indexed by $X_t$ with the property that $\vec{u}_\tau[X_{t'}] = \vec{u}[X_{t'}]$ and $\vec{u}_\tau[v] = \tau_0$.
              Set $T_t[k,\ell] = \{\vec{u}_\sigma \mid \vec{u} \in T_{t'}[k,\ell]\} \cup \{\vec{u}_\rho \mid \vec{u} \in T_{t'}[k,\ell]\}$.

              The correctness of this computation is justified by the fact that the new vertex is either selected or not, and that it has no neighbors which are forgotten, and that it is itself not forgotten yet, so even it if is selected, it does not change the state of the other vertices of the bag.

        \item[Join nodes:]
              Let $t$ be a join node with children $t_1$ and $t_2$.
              We now intend to use \cref{thm:fast_convolution} to quickly handle the node.
              For this purpose, for any $k, \ell \in \range[0]{|V(G)|}$, and $\vec{u} \in \allStatesPartial^{X_{t}}$ compute the functions
              \begin{align*}
                  f^{\ell}_1(\vec{u},k) = \begin{cases}
                                              1 & \text{if $\vec{u} \in T_{t_1}[k,\ell]$}, \\
                                              0 & \text{otherwise,}
                                          \end{cases} \quad \text{and} \quad
                  f^{\ell}_2(\vec{u},k) = \begin{cases}
                                              1 & \text{if $\vec{u} \in T_{t_2}[k,\ell]$}, \\
                                              0 & \text{otherwise.}
                                          \end{cases}
              \end{align*}

              Then, for any non-negative $\ell_1 + \ell_2 = \ell \in \range[0]{|V(G)|}$, $k \in \range[0]{|V(G)|}$, and $\vec{u} \in \allStatesPartial^{X_{t}}$, compute
              \begin{align*}
                  S_t^{\ell_1,\ell_2}[\vec{u},k] = \sum_{\vec{u}_1 \oplus \vec{u}_2 = \vec{u}} \sum_{k_1 + k_2 = k} f_1^{\ell_1}(\vec{u}_1,k_1) \cdot f_2^{\ell_2}(\vec{u}_2,k_2)
              \end{align*}
              using \cref{thm:fast_convolution}.
              Finally, compute
              \[T_t[k,\ell] = \{\vec{u} \in \allStatesPartial^{X_t} \mid \exists \ell_1, \ell_2: \ell_1 + \ell_2 = \ell \land S_t^{\ell_1,\ell_2}[\vec{u},k] > 0\}.\]

              In the convolution, the outer sum ensures that only strings $\vec{u}_1, \vec{u}_2 $ which actually join to vector $\vec{u}$ are considered.
              The inner sum ensures that the values $k_1$ and $k_2$, which represent the solution size in $V_{t_1}$ and $V_{t_2}$, properly sum up to $k$.
              Finally, the product inside the inner sum is only non-zero if there are actually matching strings in the tables $T_{t_1}[k_1,\ell_1]$ and $T_{t_1}[k_2,\ell_2]$.
              The convolution is computed separately for any number of violations $\ell_1, \ell_2$ which sum up to $\ell$.
              Then finally, the table entry $T_t[k,\ell]$ is set to contain all those strings for which there exist values $\ell_1, \ell_2$ that sum up to $\ell$, such that the convolution for the string had a non-zero value, i.e., such that there are two strings which can be joined to obtain $\vec{u}$, and such that their solution sizes add up to $k$.

              Observe that no overcounting occurs, because the states only keep track of the number of selected forgotten vertices, and we have that $(V_{t_1} \setminus X_t) \cap (V_{t_2} \setminus X_t) = \emptyset$.
              Similarly, we do not overcount the number of violations or solution size because we only keep track of these quantities for forgotten vertices.

        \item[Forget nodes:]
              Let $t$ be a forget node and $t'$ the unique child of $t$.
              Moreover, let $\{v\} = X_{t'} \setminus X_t$.
              For a string $\vec{u} \in \allStatesPartial^{X_{t'}}$ we define the $\sigma$-drop of $\vec{u}$ as the unique string $\vec{u}_\sigma$ indexed over $X_{t}$  such that:
              \begin{itemize}
                  \item For any $w \in X_{t} \setminus N(v)$ we have $\vec{u}_\sigma[w] = \vec{u}[w]$, and
                  \item for any $w \in X_{t} \cap N(v)$, let $\vec{u}[w] = \tau_c$. Then, $\vec{u}_\sigma[w] = \tau_{c} \oplus \tau_1$.
              \end{itemize}
              Moreover, for a string $\vec{u} \in \allStatesPartial^{X_{t'}}$, we define $h(\vec{u})$ to be the number of $\sigma$-states in $\vec{u}$ of vertices which are neighbors of $v$.

              Then, compute
              \begin{align*}
                  T_t[k,\ell] = & \{\vec{u}[X_{t'}] \mid \vec{u} \in T_t[k,\ell], \vec{u}[v] = \rho_c, c + h(\vec{u}) \in \rho\}                 \\
                                & \cup \{\vec{u}_\sigma \mid \vec{u} \in T_t[k-1,\ell], \vec{u}[v] = \sigma_c, c + h(\vec{u}) \in \sigma\}       \\
                                & \cup \{\vec{u}[X_{t'}] \mid \vec{u} \in T_t[k,\ell-1], \vec{u}[v] = \rho_c, c + h(\vec{u})\notin \rho\}        \\
                                & \cup \{\vec{u}_\sigma \mid \vec{u} \in T_t[k-1,\ell-1], \vec{u}[v] = \sigma_c, c + h(\vec{u}) \notin \sigma\},
              \end{align*}
              where we assume that table entries which are out of bounds contain the empty set.

              This computation essentially accounts for all four possible cases of $v$ being selected, and $v$ being violated or not violated.
              Observe that the state of $v$ does not immediately tell us whether the vertex is violated or not for any string $\vec{u} \in \allStatesPartial^{X_{t'}}$, because it only describes the number of forgotten selected neighbors of $v$, in particular, to figure out how many selected neighbors $v$ actually has in a partial solution, we need to add $h(\vec{u})$ to this amount.
              Then, we can easily deduce whether vertex $v$ is satisfied and or selected, and adjust the tables based on this information.
              The last aspect we need to take care of is that when $v$ is selected and forgotten, this naturally should increase the number of selected neighbors of every neighbor that $v$ has in the bag.
              Hence, we use the $\sigma$-drop strings in this case.
    \end{description}

    At the root node $r$, we can simply extract the result for any value $k,\ell$ by checking whether $T_r[k,\ell] = \emptyset$ or $T_r[k,\ell] = \{\varepsilon\}$.

    \subparagraph*{Correctness.}
    Within each type of node, an informal argument for the correctness was provided.
    Formally, the correctness follows by induction.

    \subparagraph*{Running time.}
    We can clearly handle the leaf nodes in polynomial-time $|G|^{\Oh(1)}$.
    For introduce nodes, we also only need to iterate over $T_t[k,\ell]$ for any $k,\ell$ and perform a simple operation.
    Hence, we can handle these nodes in time $|\allStatesPartial|^{\tw} \cdot |G|^{\Oh(1)}.$
    For forget nodes, we similarly only need to iterate over four different tables $T_t[k,\ell]$, and perform a simple polynomial-time computation for each string.
    Hence, these can be computed in time $|\allStatesPartial|^{\tw} \cdot |G|^{\Oh(1)}$ as well.
    At the root node, we can read the result for any $k,\ell$ in time $|G|^{\Oh{(1)}}$.
    The overall number of nodes of our nice tree decomposition is also in $|G|^{\Oh(1)}$.

    All that remains is arguing that we can also handle the join-nodes quickly enough.
    There, we initially compute the functions $f_1^\ell$ and $f_2^\ell$ for $|V(G)|+1$ values of $\ell$.
    To compute these, we initially set all values of the functions to zero, and then iterate over all strings of $T_{t_1}[k,\ell]$, $T_{t_2}[k,\ell]$, respectively, to set some entries to one.
    This needs time $|\allStatesPartial|^{\tw} \cdot |G|^{\Oh(1)}$.

    Then, for each $\ell_1 + \ell_2$ summing up to $\ell$, we compute a table $S_t^{\ell_1,\ell_2}$.
    For this purpose, we utilize \cref{thm:fast_convolution} which allows us to compute this table in $|\allStatesPartial|^{\tw} \cdot |G|^{\Oh(1)}$ arithmetic operations,
    and as each arithmetic operation is on numbers with polynomial bit-length, also in time $|\allStatesPartial|^{\tw} \cdot |G|^{\Oh(1)}$  overall.
    Finally, recovering $T_t[k,\ell]$ from the tables $S_t^{\ell_1,\ell_2}$ can be done by iterating over at most $|\allStatesPartial|^{\tw + 1}$ strings, and for each string, checking whether there are $\ell_1,\ell_2$ that sum up to $\ell$ and whether $S_t^{\ell_1,\ell_2}[s,k] > 0$ takes polynomial time.

    Overall, we obtain that all nodes can be handled in time $|\allStatesPartial|^{\tw} \cdot |G|^{\Oh(1)}$, concluding the proof.
\end{proof}

\section{Intermediate Lower Bounds}
\label{sec:intermediate_lower_bound}
We now proceed to the tight lower bounds for \minParGenDomSet{}.
Our general strategy is the same as a long line of work in the domain of tight lower bounds for dynamic programming algorithms, and was previously used by \cite{curticapeanTightConditionalLower2016,fockeTightComplexityBoundsLowerBound,greilhuberResidueDominationBoundedTreewidth2025,marxDegreesGapsTight2021,marxAntifactorFPTParameterized2025}.
The idea is to first reduce to a problem which lies in between the final problem one wants to show hardness to, and the input problem (which will usually be some variant of \problem{SAT}, or in our case, a constraint satisfaction problem).

\subsection{The PWSETH and Constraint Satisfaction Problems}

There is a discrepancy between the lower bound provided by the \pwseth{}, which is of the form $(2 - \varepsilon)^{\pw}$ and the lower bounds we seek, which are of the form $(|\allStatesPartial| - \varepsilon)^{\pw}$.
For precisely such situations, Lampis shows that the \pwseth{} is actually equivalent to showing that no faster algorithm for a family of constraint satisfaction problems exists.
As the lower bound provided by these constraint satisfaction problems have a more suitable form for us, we will reduce from these problems instead of going from \kSAT[3]{}.

Lampis \cite{lampisPrimalPathwidthSETH2025} considers the \CSP{} problem, here, the input consists of $n$ variables $x_1,\dots,x_n$, and $m$ constraints $C_1,\dots,C_m$.
Each variable can take values from $\range{B}$, and each constraint $C_i$ (for $i \in \range{m}$) is a pair $(\scp{C_i},\acc{C_i})$, where $\scp{C_i}$ is an ordered pair of two variables of the instance, and $\acc{C_i} \subseteq \range{B} \times \range{B}$ a set of accepted assignments to the variables.
The goal is to decide whether there exists a function $\delta: \{x_1,\dots,x_n\} \rightarrow \range{B}$ such that, for any constraint $C_i$ with $\scp{C_i} = (x_\lambda,x_{\lambda'})$, we have $(\delta(x_\lambda),\delta(x_{\lambda'})) \in \acc{C_i}$.
The primal graph of an instance $I$ is the graph where the vertex set is $\{x_1,\dots,x_n\}$, and two vertices are connected by an edge if and only if both appear in the same constraint.

\begin{theorem}[Follows from {\cite[Theorem 3.2]{lampisPrimalPathwidthSETH2025}}]\label{thm:lampis}
    Let $B \geq 3$ be an integer.
    The \pwseth{} is true if and only if, for all $\varepsilon > 0$, there is no algorithm for \CSP{} that takes an instance $I$ of \CSP{} together with a path decomposition of width $\pw$ of its primal graph as input and decides the problem in time $(B - \varepsilon)^\pw{} \cdot |I|^{\Oh(1)}$.
\end{theorem}

Lampis \cite{lampisPrimalPathwidthSETH2025} also provides some additional equivalent problem variants, such as a weighted variant, but we do not require these variants for our lower bounds.

Suppose that in problem $P_1$ every variable has a domain of size $B_1$, and we already have a lower bound ruling out
$(B_1 - \varepsilon)^\pw{} \cdot |I|^{\Oh(1)}$ time algorithms for problem $P_1$.
Let $P_2$ be a problem where every vertex has $B_2$ states in the natural dynamic programming algorithm operating on path decompositions, and our goal is to
rule out $(B_2 - \varepsilon)^\pw{} \cdot |I|^{\Oh(1)}$ time algorithms for problem $P_2$.
We can try to prove this lower bound by a reduction from problem $P_1$ to problem $P_2$.
If $B_1=B_2$, then we can hope for a fairly direct reduction: each vertex of problem $P_2$ represents a single variable of problem $P_1$, with each of the $B_2$ states representing one of the $B_1$ possible domain values.

Of course, if $B_1\neq B_2$, then we cannot hope for such a direct reduction. In such a case, one can use a grouping idea (already present in the work of Lokshtanov, Marx, and Saurabh~\cite{lokshtanovKnownAlgorithmsGraphs2018}). Instead of a one-to-one correspondence, we want to represent a group of $q$ variables of problem $P_1$ by a group of $g$ variables of problem $P_2$, where $q$ and $g$ are constants depending only on $\varepsilon>0$, $B_1$, and $B_2$.
This may require the construction of larger gadgets, as now each gadget needs to operate on some number of groups of $g$ variables.
Furthermore, the values of $q$ and $g$ needs to be carefully selected.
We want to represent each of the $B_1^q$ possible values of $q$ variables in problem $P_1$ by a ``codeword'': a vector of length $g$ consisting of possible states of $g$ variables in problem $P_2$.
For this to be possible, we need $B_2^g \ge B_1^q$, or in other words, $g$ needs to be an integer at least $q\cdot\log_{B_2} B_1$.
Moreover, $g$ should not be significantly larger than $q \cdot \log_{B_2} B_1$, because we want to argue  that a $(B_2-\varepsilon)^g$ factor in the running time is strictly less than $(B_1-\varepsilon')^q$ for some $\varepsilon'>0$.

\cref{thm:lampis}
can sometimes be used to avoid this grouping idea, as this result allows us to start the reduction from a problem where we can choose the domain size $B$ arbitrarily.
Unfortunately, while being able to choose the initial domain size can streamline proofs significantly, it does not always alleviate the need for grouping.
The problem is that it may not easily be possible to create gadgets that operate on $x$ vertices and that can handle all $B_2^x$ possible states of these vertices: some parity or weight issue might prevent us from being able to handle every possible combination of states.
However, if we are able to handle a large (say, constant) fraction of the $B_2^g$ combinations on $g$ vertices, then this is a negligible loss, and we can still use the grouping idea to represent $q$ variables of problem $P_1$ by codewords of length $g$ of states from problem $P_2$ with a value of $g$ that is slightly larger than $q\cdot\log_{B_2} B_1$.

The goal of the following lemma is to find the appropriate constants for the reduction.
For convenience, we will be using a slightly different setting compared to what is described above.
If $B$ is the domain size of the target problem, then we want to reduce from a problem $P_1$ whose domain size is $B^q$, in such a way that each variable of the source problem is represented by a group of $g$ variables in the target problem.
We want to find $g$ such that it is sufficiently large that it is possible to find $B^q$ codewords of length $g$, but it is sufficiently small that solving problem $P_2$ in time $(B-\varepsilon)^{g \cdot \pw}$ implies solving problem $P_1$ in time $(B^q-\varepsilon')^{\pw}$ for some $\varepsilon' > 0$.
Note that the factor of $g$ in the exponent stems from the fact that our reductions blow up the pathwidth by a factor of $g$, that is, given a path decomposition of width $\pw$ of the primal graph, they output a graph with a path decomposition of width roughly $g \cdot \pw$.
The lemma works even if not every possible group size $g$ is allowed (e.g., to handle parity issues) and when only a $f(g)$ fraction of the codewords can be used.
Note that these types of constants and calculations are used in the literature for the purpose of transforming the lower bound provided by the Strong Exponential Time Hypothesis to the correct lower bound for the considered problem \cite{fockeTightComplexityBoundsLowerBound, marxDegreesGapsTight2021Arxiv}.

\begin{lemma}
    \label{thm:choosing_alphabet_size}
    Let $B \geq 2$ be an integer, let $f: \nat \rightarrow \mathbb{Z}_{\geq 1}$ be a function with $f \in 2^{o(n)}$, and $\epsilon > 0$ be a real.
    Moreover, let $X \subseteq \nat$ be an infinite set.
    Then, there are positive integers $q$, $g$ with $g \in X$ and a real $\varepsilon' > 0$ with $B^q \leq \frac{|B^g|}{f(g)}$ and $(B - \epsilon)^g \leq B^q - \varepsilon'$.
\end{lemma}
\begin{proof}
    Proceed as follows.
    \begin{enumerate}
        \item Choose some $0 < \varepsilon'' < \varepsilon$.
        \item Set $\alpha = \log_{B-\varepsilon}(B-\varepsilon'')$. Because $\varepsilon'' < \varepsilon$ we have $B-\varepsilon < B- \varepsilon''$ and hence $\alpha > 1$.
        \item Choose integer $\blocksPerVertexGadget \geq 1$ large enough such that $\blocksPerVertexGadget \leq \alpha \lfloor \blocksPerVertexGadget - \log_{B}(f(g)) \rfloor$ and $g \in X$.
        \item Fix $q = \lfloor \blocksPerVertexGadget - \log_{B}(f(g)) \rfloor$. Observe that we have $\blocksPerVertexGadget \leq \alpha \cdot q$.
    \end{enumerate}

    The only step that might require some additional explanation is the third step.
    We show that we can indeed always choose $\blocksPerVertexGadget \geq 1$ large enough such that $\blocksPerVertexGadget \leq \alpha \lfloor \blocksPerVertexGadget - \log_{B}(f(g)) \rfloor$.
    In particular, we have that $\log \circ f$ is in $o(n)$, and thus, also function $h(x) = \alpha \cdot \log_{B}(f(x)) + \alpha$ is in $o(n)$, as $\alpha$ is just a constant.
    Set $\beta = \alpha - 1$, which is a positive real.
    Then, there exists a constant $x_0 > 0$ such that, for all $x > x_0$, we have $h(x) < \beta \cdot x$.
    Set $g$ to an integer larger than $x_0$ that is in the set $X$.

    Then, we trivially have
    $g = \alpha \cdot g - \beta \cdot g$, and we also have $h(g) < \beta \cdot g$.
    This implies that
    \begin{align*}
        g  \leq \alpha \cdot g - h(g)
         & = \alpha \cdot g - \alpha \log_B(f(g)) - \alpha     \\
         & = \alpha \cdot (g - \log_B(f(g)) - 1)               \\
         & \leq \alpha \cdot \lfloor g - \log_B(f(g)) \rfloor,
    \end{align*}
    which is what we wanted to show.

    By our choice of $q$ and $g$, we have $B^q = B^{\lfloor \blocksPerVertexGadget - \log_{B}(f(g)) \rfloor} \leq B^{g - \log_B(f(g))} = \frac{B^g}{f(g)}$.
    Additionally, we obtain
    $
        (B - \varepsilon)^{g} \leq (B - \varepsilon)^{\alpha \cdot q} = (B - \varepsilon'')^{q}.
    $
    Finally, since $\varepsilon'' > 0$ we have that $(B - \varepsilon'')^q < B^q$ and hence $(B - \varepsilon'')^{q} \leq B^q - \varepsilon'$
    for some $\varepsilon' > 0$.
    To conclude the proof, briefly observe that both $q$ and $g$ are positive integers.
\end{proof}

We remark that \cref{thm:choosing_alphabet_size} is not constructive, in the sense that it may not be possible to efficiently compute the sketched values. 
For example, it may be difficult to compute whether $g$ is in the set $X$ in step 3.
However, this is not problematic because one must not actually compute these values to make the lower bounds work.
Indeed, to show tight lower bounds, we will have a separate reduction that rules out improvements for each \emph{fixed constant} $\varepsilon$.
Since these reductions are relative to this fixed $\varepsilon$, they are allowed to simply hard-code the values $q$ and $g$, and their pure existence is sufficient.

 \label{sec:alphanet_size}
\subsection{The Intermediate Lower Bound for the Partial Problem}
\label{sec:intermediate_partial}

As mentioned above, we first reduce \CSP{} to an intermediate problem called \minParGenDomSetRel{}.
Problem instances of the intermediate problem contain an additional set of constraints, more precisely they contain lists of allowed selections for subsets of the vertex set, these additional constraints are referred to as \emph{relations} (a concept that has already proven to be useful in \cite{curticapeanTightConditionalLower2016,fockeTightComplexityBoundsLowerBound,greilhuberResidueDominationBoundedTreewidth2025,marxDegreesGapsTight2021,marxAntifactorFPTParameterized2025}).
They allow us to express the input instance, which is not a graph problem, as graph problem in which the constraints of the input instance essentially directly correspond to specific relations of the output instance.
Moreover, these relations are helpful in ensuring that the output instance behaves properly in other ways as well.
The idea is to then later get rid of these relations by replacing them with gadgets which model their behavior.
In this section, we cover this intermediate reduction.

For technical reasons, we actually utilize two separate intermediate problems with slightly different properties.
Let us now define these problems \minParGenDomSetRel{} and \parGenDomSetRel{} we reduce to.
Before we can give the formal problem definitions, we need to formally define the notion of a relation.
\begin{restatable}[Relations]{definition}{defRelations}
    \label{def:relations}
    Given a graph $G$, a \emph{relation} $R$ (of $G$) is a pair $(\scp{R},\acc{R})$, where $\scp{R} \subseteq V(G)$ is the scope of the relation, and $\acc{R} \subseteq 2^{\scp{R}}$ the set of accepted assignments.
    The arity of a relation $R$ is $|\scp{R}|$.
    A relation $R$ is superset-closed if for any set $r \in \acc{R}$, it holds that any superset $r' \subseteq \scp{R}$ of $r$ is also in $\acc{R}$.
\end{restatable}
Now, we are ready to state the problem definitions.

\problembox{\minParGenDomSetRel{}}{Graph $G$, integers $k, \ell$, set of superset-closed relations $\mathcal{R}$}{Is there a set $S \subseteq V(G)$ such that $|S| \leq k$ and at most $\ell$ vertices of $V(G)$ are violated by $S$ relative to $(\sigma,\rho)$, and for each $R \in \mathcal{R}$ we have $S \cap \scp{R} \in \acc{R}$?}

\problembox{\parGenDomSetRel{}}{Graph $G$, integer $\ell$, set of relations $\mathcal{R}$}{Is there a set $S \subseteq V(G)$ such that at most $\ell$ vertices of $V(G)$ are violated by $S$ relative to $(\sigma,\rho)$, and for each $R \in \mathcal{R}$ we have $S \cap \scp{R} \in \acc{R}$?}

We say that an instance of \minParGenDomSetRel{} or \parGenDomSetRel{} has arity $d$ (or at most $d$), when $|\scp{R}| \leq d$ for each $R \in \mathcal{R}$.

Observe that the problem \parGenDomSetRel{} is not a minimization problem at all.
The reason for this is that we can actually also easily show the lower bound for \parGenDomSetRel{}, which is sufficient for us in some cases.
However, in other cases, it is not easy to reduce \emph{from} \parGenDomSetRel{}, so, we instead utilize the problem \minParGenDomSetRel{}, which only uses superset-closed relations.

It turns out that, for our purpose, some additional conditions on the output instances created by our reduction are also convenient.
For this purpose, we define the notion of a good instance.

\defGoodInstances*{}

Now, we define the notion of a complement of a state.
This type of definition is essentially the definition also used in \cite[Definition 3.10]{fockeTightComplexityBoundsLowerBound} extended to our case, in which we also have overflow-states for finite sets.
Recall that $\tauLargestNonPartial = \min \tau$ if $\tau$ is simple cofinite, and $\max \tau$ if $\tau$ is finite.

\begin{definition}[Complement of States]
    \label{def:partial_state_complements}
    Let $\sigma, \rho$ be nonempty finite or simple cofinite sets.
    Let $\tau_c \in \allStatesPartial$ be an arbitrary state.
    We define the complement of state $\tau_c$ as
    \begin{equation*}
        \complState{\tau_c} = \begin{cases}
            \tau_{\tauLargestNonPartial - c} & \text{if $\tau_c$ is not an overflow state}, \\
            \tau_c                           & \text{if $\tau_c$ is an overflow state}.     \\
        \end{cases}
    \end{equation*}
\end{definition}

Observe that, for any state $\tau_c$, we have that $\complState{\complState{\tau_c}} = \tau_c$.
Furthermore, when $\tau$ is simple cofinite, we have that $\tau_x = \complState{\tau_y}$ implies $x +y = \min \tau$, and when $\tau$ is finite and $x \leq \max \tau$, we have that $\tau_x = \complState{\tau_y}$ implies $x + y = \max \tau$.
On the other hand, when $\tau_c$ is an overflow-state, also $\complState{\tau_c}$ is an overflow-state.

For the purpose of choosing the correct alphabet size, we need to also define weight notions for states of $\allStatesPartial$.
The main idea behind these weights also goes back to \cite{fockeTightComplexityBoundsLowerBound}, they and we use weights to ensure proper propagation of states in our construction.

\begin{definition}[Weight of State-Vectors]
    \label{def:partial_weights}
    Let $\tau_c \in \allStatesPartial$.
    The weight of $\tau_c$ is defined as
    \begin{align*}
        \weight{\tau_c} = \begin{cases}
                              c                  & \text{if $\tau_c$ is not an overflow state}, \\
                              \allLargestPartial & \text{if $\tau_c$ is an overflow state}.
                          \end{cases}
    \end{align*}

    For a vector $\vec{u} \in \allStatesPartial^n$, for any integer $n \geq 1$, let $P_\sigma$ be those positions of $\vec{u}$ which are $\sigma$-states, and $P_\rho$ those positions of $\vec{u}$ which are $\rho$-states.
    We define the weight of $\vec{u}$ as
    \begin{align*}
        \weight{\vec{u}} = \sum_{i \in \range{n}} \weight{\vec{u}\,[i]},
    \end{align*}
    the $\sigma$-weight of $\vec{u}$ as    \begin{align*}
        \sigWeight{\vec{u}} = \sum_{i \in P_\sigma} \weight{\vec{u}\,[i]},
    \end{align*}
    and the $\rho$-weight of $\vec{u}$ as
    \begin{align*}
        \rhoWeight{\vec{u}} = \sum_{i \in P_\rho} \weight{\vec{u}\,[i]}.
    \end{align*}
\end{definition}

Finally, we need to define the notion of a path decomposition of instances of \minParGenDomSetRel{} and \parGenDomSetRel{}.

\begin{definition}[Path Decompositions of the Problems with Relations]
    \label{def:partial_intermediate_path_decompositions}
    Let $(G,\ell,\mathcal{R})$ be an instance of \parGenDomSetRel{}.
    A path decomposition of the instance is a path decomposition of $G$, such that for each $R \in \mathcal{R}$ there exists a bag of the decomposition that contains all vertices of $\scp{R}$.
    Path decompositions for instances of \minParGenDomSetRel{} are defined analogously.
\end{definition}

\begin{figure}
    \centering
    \includegraphics[width = 0.99\linewidth]{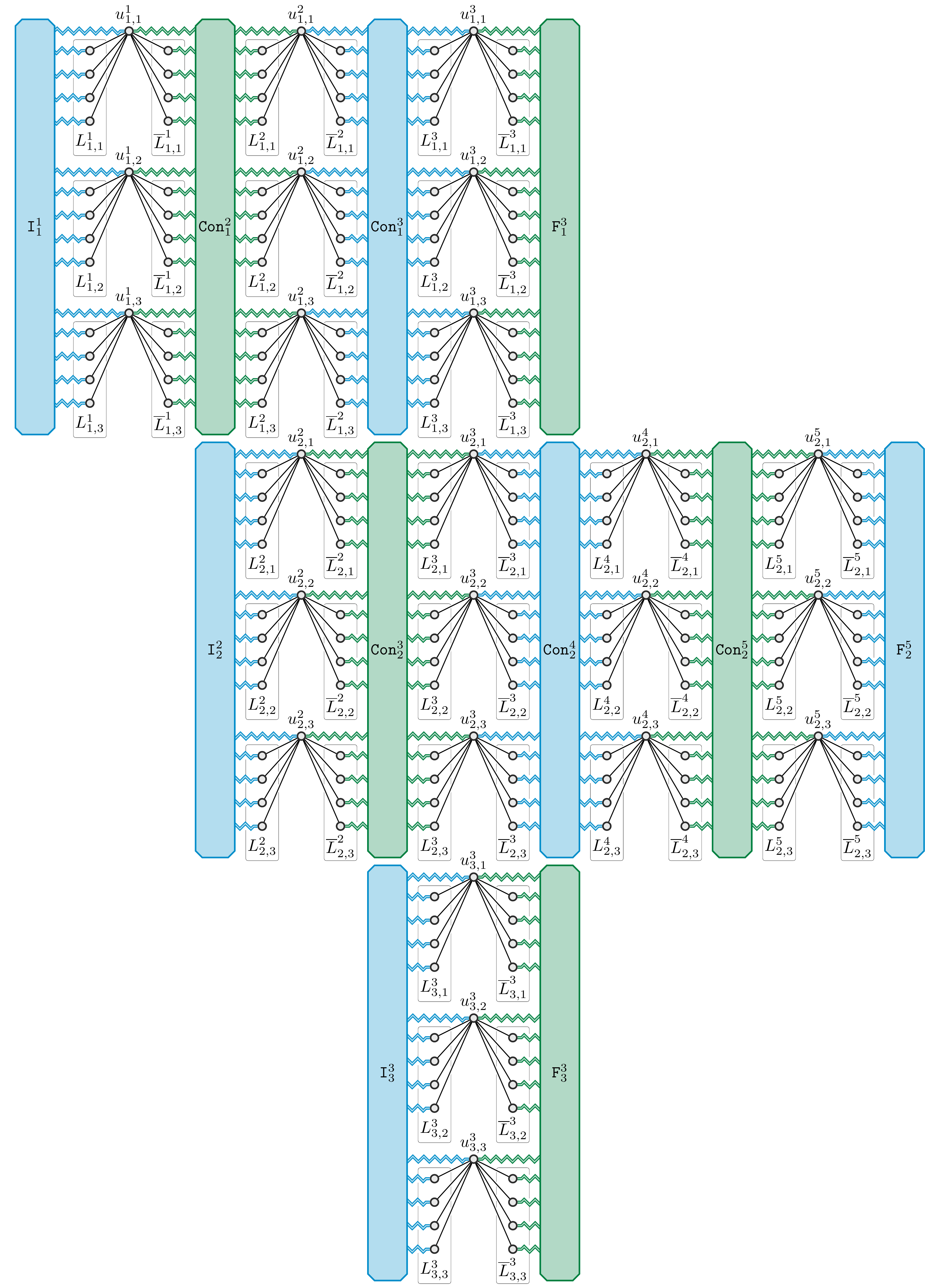}
    \caption{An illustration of the intermediate construction for an input instance on three variables $x_1,x_2,x_3$ with path decomposition of the primal graph $\{x_1\},\{x_1,x_2\},\{x_1,x_2,x_3\},\{x_2\},\{x_2\}$. Furthermore, we have that $g = 3$, and $\allLargestPartial = 4$. For the sake of visibility, the clique $K$ and the constraint relations are not shown, and the drawn relations of $\mathcal{R}$ are labeled with the name of their corresponding relation of $\mathcal{R}_0$. The relations are colored to make it easier to distinguish them.
    }
    \label{fig:intermediate_construction}
\end{figure}

Now, we have all ingredients we require for our first intermediate lower bound.

\thmMinimizationHighLvl*{}
\begin{proof}
    Although the theorem statement only explicitly mentions the case where $\rho$ is simple cofinite, most parts of this proof also work for the case when $\rho$ is finite.
    Actually the only part of the proof that does not work for finite $\rho$ is the fact that then, the output instance is not necessarily good.
    Since we want to reuse parts of this proof later,
    let $\sigma, \rho$ be nonempty sets which are finite or simple cofinite, and keep in mind that we only require the final proof of the theorem statement, in particular that the output instance is good, when $\rho$ is simple cofinite.

    \item\subparagraph*{The size of the alphabet.}

    Before we formally state the reduction, we need to figure out which alphabet size of the 2-\problem{CSP} problem is correct for our use case.
    For any $n$, let $W_n$ be the partition of $\allStatesPartial^n$, such that each string in the same block has the same number of $\sigma$-states, the same number of overflow $\sigma$-states, the same number of overflow $\rho$-states, the same $\sigma$-weight, and the same $\rho$-weight.

    We will use a simple argument to bound the size of $W_n$ for each $n$.
    Observe that the number of $\sigma$-states ranges from $0$ to $n$, similarly, the number of overflow $\sigma$-states and overflow $\rho$-states range from $0$ to $n$.
    Finally, the $\sigma$-weight and $\rho$-weight range from $0$ to $n \cdot \allLargestPartial$.
    Now, we obtain that $|W_n| \leq (n+1)^3 \cdot (n \cdot \allLargestPartial + 1)^2$, which is polynomial in $n$ (note that $\allLargestPartial$ is a constant).\footnote{A slightly more involved argument would give a better bound on the size of $|P_n|$, however, the stated bound is perfectly sufficient for our purposes.}
    Moreover, $|W_n| \geq 1$ for all $n \geq 0$ ($\allStatesPartial^0$ contains the empty string).

    Next, fix an arbitrary $\varepsilon > 0$.
    By \cref{thm:choosing_alphabet_size} (used with $B = |\allStatesPartial|$, $f(n) = |W_n|$, the chosen $\varepsilon$ and $X = \mathbb{Z}_{\geq 3}$),
    we get that there are positive integers $q$, $g \geq 3$ and a real $\varepsilon' > 0$ with $|\allStatesPartial|^q \leq \frac{|\allStatesPartial|^g}{|W_g|}$, and $(|\allStatesPartial| - \varepsilon)^g \leq (|\allStatesPartial|^q - \varepsilon')$.
    Let $\usedStates$ be a largest block of $W_g$, then $\usedStates$ has size at least $\frac{|\allStatesPartial|^g}{|W_g|} \geq |\allStatesPartial|^q$.
    So, we have $(|\allStatesPartial| - \varepsilon)^g \leq (|\allStatesPartial|^q - \varepsilon') \leq (|A| - \varepsilon')$.
    Moreover, since we have that $g \geq 3$, the size of $A$ is at least $3$ since for example, the set in which each string has $1$ state that is $\rho_0$ and $g-1$ states that are $\sigma_0$ is a subset of a partition of $W_g$ that contains at least three strings.

    Let $\selStateVertexPerGadget$ be the number of $\sigma$-states per vector of $A$, $\numOverflowSigma$ be the number of overflow $\sigma$-states per vector, $\numOverflowRho$ the number of overflow $\rho$-states per weight, $\selWeight$ the $\sigma$-weight per weight, and $\unselWeight$ the $\rho$-weight per weight.
    Note that when $\sigma$ is simple cofinite, we have that $\numOverflowSigma = 0$, similarly, when $\rho$ is simple cofinite, we have that $\numOverflowRho = 0$.
    Also define the total number of overflow-states per vector $\numOverflowAll = \numOverflowRho + \numOverflowSigma$, the weight per vector, $\weightConstant = \selWeight + \unselWeight$, and the number of $\rho$-states per vector $\unselStateVertexPerGadget = \blocksPerVertexGadget - \selStateVertexPerGadget$.

    We stress that the values $g,q$ and the set $A$ are not computed during the actual reduction.
    Rather, there is a separate reduction for each fixed constant $\varepsilon > 0$ which can hard-code these values.
    In particular, the pure existence of these values is perfectly sufficient, and we do not require \cref{thm:choosing_alphabet_size} to be constructive.

    \subparagraph*{The construction.}
    We reduce from 2-\problem{CSP} over the alphabet $\range{|A|}$.
    Because we can easily bijectively map integers from $\range{|A|}$ to elements of $A$, we will treat this as an instance over the alphabet $A$, and keep the fact that we use such a mapping implicit.

    Let the variables of the input instance $I$ be $\variable{1},\dots,\variable{\variableCount}$, and the constraints be $\constraint{1},\dots,\constraint{\constraintCount}$.
    Moreover, let the provided path decomposition of the primal graph be $\bag{1},\dots,\bag{\bagCount}$.

    Observe that we can, without loss of generality, assume that there exists an injective function $\constraintBagMapping{}$ which assigns each constraint (the index of) a bag such that all variables of the constraint are present in the bag.
    Indeed, if no such function exists we can simply copy bags of the decomposition until such a function can easily be computed.\footnote{This type of standard argument was already provided by Lampis \cite{lampisPrimalPathwidthSETH2025}.}

    We build the graph $G$ of the output instance as follows.
    \begin{itemize}
        \item For each $i \in \range{\bagCount}$ and each $\variable{j} \in \bag{i}$ and each $p \in \range{\blocksPerVertexGadget}$ we create \emph{state} vertex $\statevertex{i}{j}{p}$.
        \item For each $i \in \range{\bagCount}$ and each $\variable{j} \in \bag{i}$ and each $p \in \range{\blocksPerVertexGadget}$ we create $\allLargestPartial$ vertices, and denote the set of them as $\leftBlock{i}{j}{p}$. We connect each vertex of $\leftBlock{i}{j}{p}$ to vertex $\statevertex{i}{j}{p}$ with an edge.
        \item For each $i \in \range{\bagCount}$ and each $\variable{j} \in \bag{i}$ and each $p \in \range{\blocksPerVertexGadget}$ we create $\allLargestPartial$ vertices, and denote the set of them as $\rightBlock{i}{j}{p}$. We connect each vertex of $\rightBlock{i}{j}{p}$ to vertex $\statevertex{i}{j}{p}$ with an edge.
        \item We add a clique $K$ on $\allLargestPartial+1$ vertices to the graph.
        \item For each $i \in \range{t}$, each $x_j \in \bag{i}$, each $p \in \range{\blocksPerVertexGadget}$ and each $v \in \leftBlock{i}{j}{p} \cup \rightBlock{i}{j}{p}$, we add an edge between $v$ and each vertex of $K$. In other words, we connect each vertex that is neither in $K$ nor a state vertex to each vertex of $K$ with an edge.
    \end{itemize}

    This concludes the description of the output graph $G$.
    For convenience, we define some additional sets.
    For any bag $\bag{i}$ and $\variable{j} \in \bag{i}$, we define $U^i_j = \bigcup_{p \in \range{g}} \{\statevertex{i}{j}{p}\}$, $L^i_j = \bigcup_{p \in \range{g}} \leftBlock{i}{j}{p}$ and $\compl{L}^i_j = \bigcup_{p \in \range{g}} \rightBlock{i}{j}{p}$.
    We also set $\bag{0} = \bag{\bagCount + 1} = \emptyset$.

    Next, we define the state of arbitrary state vertex $\statevertex{i}{j}{p}$ relative to some vertex set $S \subseteq V(G)$ as follows.
    Let $\tau$ be $\sigma$ if $\statevertex{i}{j}{p}$ is selected by $S$, and $\rho$ otherwise.
    Define $c = \min(|S \cap \leftBlock{i}{j}{p}|,\tauLargestPartial)$, and $\state{\statevertex{i}{j}{p}} = \tau_c$.already
    Let $v_1,v_2,\dots,v_\ell$ be a list of $\ell$ state vertices.
    Then, we define $\state{v_1,v_2,\dots,v_\ell} = \state{v_1}\state{v_2}\dots\state{v_\ell}$.

    We similarly define the complementary states of state vertices, $\rightState{\statevertex{i}{j}{p}}$, the only difference in the definition is that we utilize $\rightBlock{i}{j}{p}$ instead of $\leftBlock{i}{j}{p}$.
    We also extend the definition of complementary states of state vertices to lists of complementary states of state vertices.

    \subparagraph*{The relations.}

    Now, we are ready to elaborate on the relations we use.
    We will first create a set of relations $\mathcal{R}_0$, which will not be the set of the relations we end up using, as each relation of the set will not necessarily be superset-closed.
    Later, we compute the set $\mathcal{R}$ of relations that we actually use from $\mathcal{R}_0$.
    So, let us now describe the relations of $\mathcal{R}_0$.

    For each $i \in \range{\bagCount}$ and each $\variable{j} \in \bag{i}\setminus{\bag{i-1}}$, we create an \emph{initial} relation $\initialRelation{i}{j}$.
    The relation scope of this relation is $U^i_j \cup L^i_j$.
    The relation accepts selection $S \cap \scp{\initialRelation{i}{j}}$ if and only if
    \begin{enumerate}
        \item $\state{\statevertex{i}{j}{1},\statevertex{i}{j}{2},\dots,\statevertex{i}{j}{\blocksPerVertexGadget}} \in A$, and
        \item for every $p \in \range{g}$, we have that $S$ selects exactly $\weight{\state{\statevertex{i}{j}{p}}}$ vertices of $L^i_j$.
    \end{enumerate}

    For each $i \in \range{\bagCount}$ and each $\variable{j} \in \bag{i} \cap \bag{i-1}$ we create a \emph{consistency} relation $\consistencyRelation{i}{j}$.
    The relation scope of this relation is $\scp{\consistencyRelation{i}{j}} = U^i_j \cup U^{i-1}_j \cup L^i_j \cup \compl{L}^{i-1}_j$.
    The relation accepts a selection $S \cap \scp{\consistencyRelation{i}{j}}$ if and only if all the following items hold:
    \begin{enumerate}
        \item $\state{\statevertex{i}{j}{1},\statevertex{i}{j}{2},\dots,\statevertex{i}{j}{\blocksPerVertexGadget}} \in A$.
        \item For all $p \in \range{\blocksPerVertexGadget}$ we have that $\state{\statevertex{i}{j}{p}} = \complState{\rightState{\statevertex{i-1}{j}{p}}}$.
        \item For each $p \in \range{\blocksPerVertexGadget}$ exactly $\weight{\rightState{\statevertex{i-1}{j}{p}}}$ vertices of $\rightBlock{i-1}{j}{p}$ are selected.
        \item For each $p \in \range{\blocksPerVertexGadget}$ exactly $\weight{\state{\statevertex{i}{j}{p}}}$ vertices of $\leftBlock{i}{j}{p}$ are selected.
    \end{enumerate}

    Next, for each bag $\bag{i}$ and variable $\variable{j} \in \bag{i}$ with $\variable{j} \notin \bag{i+1}$, we add a final relation $\finalRelation{i}{j}$ with relation scope $U^i_j \cup \compl{L}^i_j$.
    This relation ensures that the selection properly corresponds to the complement of some selection of $A$.
    That is, a selection of $\scp{\finalRelation{i}{j}}$ is accepted if and only if
    \begin{enumerate}
        \item $\complState{\rightState{\statevertex{i}{j}{1}}}\dots\complState{\rightState{\statevertex{i}{j}{\blocksPerVertexGadget}}} \in \usedStates$, and
        \item for each $p \in \range{\blocksPerVertexGadget}$ exactly $\weight{\rightState{\statevertex{i}{j}{p}}}$ vertices of $\rightBlock{i}{j}{p}$ are selected.
    \end{enumerate}

    Now, we add a relation to each vertex of $K$ which forces the vertex to be selected.
    We do not need to give these relations explicit names.

    We proceed by adding the relations we use to model the constraints of the input instance.
    For each $\ell \in \range{\constraintCount}$ we create a \emph{constraint} relation $\constraintRelation{\ell}$ (for constraint $\constraint{\ell}$).
    Set $i = \constraintBagMapping{\constraint{\ell}}$ (recall that $\gamma$ is an injective function that assigns each constraint the index of a bag that contains both variables of the constraint).
    Let $(x_j,x_{j'})$ be the variables appearing in constraint $\constraint{\ell}$.
    The scope of the relation is then $\scp{\constraintRelation{\ell}} = L^i_{j} \cup U^i_j \cup L^i_{j'} \cup U^i_{j'}$.
    The constraint accepts a selection of $\scp{\constraintRelation{\ell}}$ if and only it is the case that
    \[(\state{\statevertex{i}{j}{1},\dots,\statevertex{i}{j}{g}},\state{\statevertex{i}{j'}{1},\dots,\statevertex{i}{j'}{g}}) \in \acc{\constraint{\ell}}.\]

    Let $\mathcal{R}_0$ be the set of all relations described so far.
    Compute the superset-closure $\mathcal{R}$ of $\mathcal{R}_0$, that is, $\mathcal{R}$ is created by iterating overall $R_0 \in \mathcal{R}$, adding a relation $R$ with $\scp{R} = \scp{R_0}$ and $\acc{R} = \{X' \mid X \in \acc{R_0},\scp{R_0} \supseteq X' \supseteq X\}$.

    \subparagraph*{The solution size and violation count.}
    Next, we define our number $\ell$ of allowed violations.
    Before giving the exact values, we provide some intuition for the value we require.
    Recall that the relations $\mathcal{R}$ (and $\mathcal{R}_0$) force all vertices of the clique $K$ to be selected.
    If $\sigma$ is simple cofinite, each vertex of $K$ will be satisfied, and each selected vertex that is not a state vertex will be satisfied due to being adjacent to the vertices of $K$.
    If $\sigma$ is finite, each selected non-state vertex will be violated since the selected neighbors of $K$ will give them too many selected neighbors.
    If $\rho$ is simple cofinite, then each unselected non-state vertex will be satisfied due to the vertices of $K$, but if $\rho$ is finite, then each unselected non-state vertex will be violated due to the vertices of $K$.
    Finally, the number of satisfied and violated state vertices per bag and variable in the bag is completely determined by $\numOverflowSigma, \numOverflowRho$ and $\blocksPerVertexGadget$.
    Then, the violation number $\ell$ is simply the smallest number of vertices which have to be violated if all relations are fulfilled, which is a property that we will of course prove in more detail in the later parts of the proof.

    Now, we proceed to the precise definition of the value $\ell$.
    Recall that each vector of $A$ has the same number $\numOverflowSigma$ of overflow-$\sigma$ states, the same number $\numOverflowRho$ of overflow-$\rho$ states, and that $\numOverflowAll = \numOverflowSigma + \numOverflowRho$.
    Moreover, each vector of $A$ has the same number $\selStateVertexPerGadget$ of $\sigma$ states, and hence also the same number $\unselStateVertexPerGadget$ of $\rho$ states.
    We first want to conveniently calculate how many states of each vector of $A$ (and thus how many state vertices per bag and variable) are not overflow states.
    We set $h_\sigma = \selStateVertexPerGadget - \numOverflowSigma$ and $h_\rho = \unselStateVertexPerGadget -  \numOverflowRho$
    as the number of satisfied selected and unselected state vertices per variable and bag, respectively.

    Next, we want some further variables that describe how many selected non-state vertices we have per variable and bag, and how many unselected non-state vertices we have per variable and bag.
    Regarding the number of selected non-state vertices, we will later show that a state vertex with an overflow-state has exactly $2 \cdot \allLargestPartial$ selected neighbors.
    On the other hand, if a state vertex has state $\tau_c$ that is not an overflow-state, then it will have either $\min \tau$, or $\max \tau$ selected neighbors, depending on whether $\tau$ is finite or simple cofinite.
    Recall that $\sigLargestNonPartial = \min \sigma$ if $\sigma$ is simple cofinite and $\max \sigma$ otherwise, and that $\rhoLargestNonPartial = \min \rho$ if $\rho$ is simple cofinite and $\max \rho$ otherwise.
    We can now see that per satisfied selected state vertex, we have exactly $\sigLargestNonPartial$ selected non-state vertices, and per unselected satisfied state vertex, we have exactly $\rhoLargestNonPartial$ of selected non-state vertices.
    So, per bag and variable in the bag, we will have $\#_\sigma = 2 \cdot \allLargestPartial \cdot \numOverflowAll + h_\sigma \cdot \sigLargestNonPartial + h_\rho \cdot \rhoLargestNonPartial$ selected non-state vertices, which stem from the violated state vertices, as well as the satisfied state vertices.

    Since the state vertices have exactly $2 \cdot \allLargestPartial$ neighbors each, we can also calculate how many unselected non-state vertices there are.
    We have $\#_\rho = \blocksPerVertexGadget \cdot 2 \cdot \allLargestNonPartial - \#_\sigma$ unselected non-state vertices per bag and variable in the bag.
    The way we set $\ell$ now depends on $\sigma$ and $\rho$, as the violation status of the non-state vertices depends on whether $\sigma$ and $\rho$ are finite or not.

    \begin{description}
        \item[$\sigma$ and $\rho$ finite:]
              We set
              \[\ell = |K| + \sum_{i \in \range{\bagCount}}\sum_{\variable{j} \in \bag{i}} \big( \numOverflowAll + \#_\sigma + \#_\rho \big).\] In other words, the number of allowed violations is simply all vertices of the graph, apart from $h_\sigma + h_\rho$ per bag and variable in each bag.
        \item[$\sigma$ is finite and $\rho$ is simple cofinite:]
              In this case, we set
              \[\ell = |K| + \sum_{i \in \range{\bagCount}}\sum_{\variable{j} \in \bag{i}} \big( \numOverflowAll + \#_\sigma \big).\]
        \item[$\sigma$ is simple cofinite and $\rho$ is finite:]
              We set \[\ell = \sum_{i \in \range{\bagCount}}\sum_{\variable{j} \in \bag{i}} \big( \numOverflowAll + \#_\rho \big) .\]
        \item[$\sigma$ and $\rho$ are simple cofinite:] We set $\ell = \sum_{i \in \range{\bagCount}}\sum_{\variable{j} \in \bag{i}} \numOverflowAll = 0$.
    \end{description}

    Set the number of vertices which are allowed to be selected to
    \begin{align*}
        k = |K| + \sum_{i \in \range{\bagCount}}\sum_{\variable{j} \in \bag{i}} \big( \selStateVertexPerGadget + \#_\sigma \big).
    \end{align*}

    The output instance is then $I' = (G,k,\ell,\mathcal{R})$, consult \cref{fig:intermediate_construction} for an illustration (where the names of the relations are the same as the corresponding relations of $\mathcal{R}_0$).

    \subparagraph*{Properties of the construction.}

    We now show that this construction has the desired properties.

    \begin{claim}
        \label{claim:minimization_high_lvl_forward}
        If the input instance $I$ is a yes-instance, then the instances $(G,k,\ell,\mathcal{R}_0)$ and $(G,k,\ell,\mathcal{R})$ are yes-instances.
    \end{claim}
    \begin{claimproof}
        Assume that $I$ is a yes-instance.
        Then there exists some assignment $\delta$ satisfying all constraints.
        We show next that this assignment can be transformed into a selection that selects at most $k$ vertices and violates at most $\ell$ vertices while fulfilling all relations of $\mathcal{R}_0$.
        As any selection fulfilling all relations of $\mathcal{R}_0$ also fulfills all relations of $\mathcal{R}$, this shows that both $(G,k,\ell,\mathcal{R}_0)$ and $(G,k,\ell,\mathcal{R})$ are yes-instances.

        We begin by selecting each vertex of $K$, which also fulfills all relations upon the vertices of $K$.
        Then, for each bag $i \in \range{t}$ and each $x_j \in \bag{i}$ and each $p \in \range{g}$, we select vertex $\statevertex{i}{j}{p}$ if and only if $\tau_c = \delta(j)[p]$ is a $\sigma$-state.
        Moreover, if $\tau_c$ is an overflow-state, then we select all $\allLargestPartial$ vertices of $\leftBlock{i}{j}{p}$ and all $\allLargestPartial$ vertices of $\rightBlock{i}{j}{p}$.
        If $\tau_c$ is not an overflow-state, then we select $c$ vertices from $\leftBlock{i}{j}{p}$ and $\tauLargestNonPartial - c$ vertices from $\rightBlock{i}{j}{p}$.
        Observe that this selection guarantees that $\state{\statevertex{i}{j}{p}} = \tau_c$, and hence, $S$ fulfills all constraint-relations.
        Moreover, our selection ensures that $\rightState{\statevertex{i}{j}{p}} = \complState{\state{\statevertex{i}{j}{p}}}$.
        So, we also have that $\tau_c = \complState{\complState{\tau_c}} = \complState{\rightState{\statevertex{i}{j}{p}}}$.

        Then, it is straightforward to check that the selection fulfills all initial relations, all consistency relations, and all final relations.
        So, all that remains is arguing that $S$ is not too large, and that $S$ violates at most $\ell$ vertices of the graph.
        It is not difficult to observe that $|S| = k$, which meets the size bound.

        Regarding the number of violated vertices, fix some $i \in \range{t}$ and $x_j \in \bag{i}$.
        Observe that exactly $\numOverflowAll$ of the state vertices  $\statevertex{i}{j}{1},\dots,\statevertex{i}{j}{\blocksPerVertexGadget}$ are violated because they have too many selected neighbors.
        All other state vertices have a number of selected neighbors which is either the maximum, or minimum value of their respective set, and hence they are satisfied.

        The vertices of the clique $K$ are always selected.
        If $\sigma$ is finite, this means that each vertex of $K$ has more than $\max \sigma$ selected neighbors, violating them all.
        On the other hand if $\sigma$ is simple cofinite, then this guarantees that they have at least $\min \sigma$ selected neighbors, making them satisfied.

        Each vertex $v$ that is not a state vertex is adjacent to all vertices of $K \setminus \{v\}$, this means that $v$ also has at least $\allLargestPartial$ selected neighbors.
        If $\rho$ is finite, this means that unselected vertices that are not state vertices are violated, if $\rho$ is simple cofinite, these are satisfied.
        Similarly, if $\sigma$ is finite, selected vertices that are not state vertices are violated, and if $\sigma$ is simple cofinite, they are satisfied.
        So, we have that $S$ violates exactly $\ell$ vertices of $V(G)$.
    \end{claimproof}

    Now, towards proving the second direction of correctness, we show some key properties of sets $S$ that fulfill all relations of $\mathcal{R}$.
    \begin{claim}
        \label{claim:partial_minimization_relations_fulfilled_size_property}
        Any set $S \subseteq V(G)$ that fulfills all relations of $\mathcal{R}$ must have size at least $k$.
        Moreover, any set $S \subseteq V(G)$ that fulfills all relations of $\mathcal{R}$ and has size at most $k$ fulfills all relations of $\mathcal{R}_0$.
    \end{claim}
    \begin{claimproof}
        Let $S \subseteq V(G)$ be a set that fulfills all relations of $\mathcal{R}_0$.
        Fix an arbitrary $i \in \range{t}$ and $\variable{j} \in \bag{i}$.
        Then, there is an initial relation or consistency relation that enforces $\state{\statevertex{i}{j}{1},\dots,\statevertex{i}{j}{\blocksPerVertexGadget}} \in \usedStates$.
        Furthermore, the relation ensures that the number of selected vertices of $\leftBlock{i}{j}{p}$ is exactly $\weight{\state{\statevertex{i}{j}{p}}}$.
        In particular, the relation ensures that if $\state{\statevertex{i}{j}{p}}$ is an overflow-state, then all $\allLargestPartial$ vertices of $\leftBlock{i}{j}{p}$ are selected.
        This means that at least $\selStateVertexPerGadget$ of the state vertices of $U^i_j$ are selected, and at least $\weightConstant$ vertices of $L^i_j$ are selected.
        Moreover, a final relation or consistency relation guarantees that $\vec{u} = \complState{\rightState{\statevertex{i}{j}{1}}}\dots\complState{\rightState{\statevertex{i}{j}{\blocksPerVertexGadget}}} \in \usedStates$.

        We now show that this means that at least $\compl{\weightConstant} = 2 \cdot \numOverflowAll \cdot \allLargestPartial + h_\sigma \cdot \sigLargestNonPartial + h_\rho \cdot \rhoLargestNonPartial - \weightConstant$ vertices of $\compl{L}^i_j$ must be selected.
        Let $P_\sigma$ be those positions of $\vec{u}$ which are $\sigma$-states that are not overflow-states, and $P_\rho$ be those positions of $\vec{u}$ which are $\rho$-states that are not overflow-states.
        Also, let $P_{>}$ be those positions of $\vec{u}$ which are overflow-states.
        It is the case that $\complState{\tau_c}$
        is an overflow-state if and only if $\tau_c$ is an overflow-state.
        Moreover, $\vec{u}$ has exactly $\numOverflowAll$ overflow-states as it is in $A$.
        We then immediately obtain that \[\sum_{p \in P_>} \weight{\vec{u}\,[p]} = \sum_{p \in P_>} \allLargestPartial = \numOverflowAll \cdot \allLargestPartial = \sum_{p \in P_>}  \weight{\rightState{\statevertex{i}{j}{p}}}.\]

        We also know that $\weight{\vec{u}} = \weightConstant$, so,
        \[\sum_{p \in P_\sigma \cup P_\rho} \weight{\vec{u}\,[p]} = \weightConstant - \numOverflowAll \cdot \allLargestPartial.\]
        Moreover, we have that \[\sum_{p \in P_\sigma} \weight{{\rightState{\statevertex{i}{j}{p}}}} = \sum_{p \in P_\sigma} \sigLargestNonPartial - \weight{\vec{u}\,[p]}, \text{ and } \sum_{p \in P_\rho} \weight{{\rightState{\statevertex{i}{j}{p}}}} = \sum_{p \in P_\rho} \rhoLargestNonPartial - \weight{\vec{u}\,[p]}. \]

        Combining these equations, we get \[ \sum_{p \in P_\sigma \cup P_\rho} \weight{{\rightState{\statevertex{i}{j}{p}}}} = h_\sigma \cdot \sigLargestNonPartial + h_\rho \cdot \rhoLargestNonPartial - \weightConstant + \numOverflowAll \cdot \allLargestPartial.
        \]
        So, overall, \[\sum_{p \in \range{\blocksPerVertexGadget}} \weight{{\rightState{\statevertex{i}{j}{p}}}} = 2 \cdot \numOverflowAll \cdot \allLargestPartial + h_\sigma \cdot \sigLargestNonPartial + h_\rho \cdot \rhoLargestNonPartial - \weightConstant = \compl{\weightConstant}. \]

        We again have that, if $\rightState{\statevertex{i}{j}{p}}$ is an overflow-state, then a final relation or consistency relation ensures that all $\allLargestPartial$ vertices of $\rightBlock{i}{j}{p}$ are selected.
        Now, this means that the relations $\mathcal{R}_0$ enforce that at least $\selStateVertexPerGadget$ vertices of $U^i_j$ are selected.
        They also enforce that at least $\weightConstant$ vertices of $L^i_j$ are selected, and finally, using the elaborations above, that at least $\compl{\weightConstant}$ vertices of $\compl{L}^i_j$ are selected.
        Furthermore, $\weightConstant + \compl{\weightConstant} = 2 \cdot \numOverflowAll \cdot \allLargestPartial + h_\sigma \cdot \sigLargestNonPartial + h_\rho \cdot \rhoLargestNonPartial = \#_\sigma$.
        Since $\mathcal{R}$ is simply the superset-closure of $\mathcal{R}_0$, the same holds true for any set that respects all relations of $\mathcal{R}$.
        Moreover, the relations among $K$ are the same under $\mathcal{R}$ and $\mathcal{R}_0$.
        So indeed, already the relations $\mathcal{R}$ ensure that at least $k$ vertices of $V(G)$ must be selected.

        Finally, we show that if we have a set $S$ of size at most $k$ that fulfills all relations of $\mathcal{R}$, it must also fulfill all relations of $\mathcal{R}_0$.
        Towards a contradiction, assume that for some $R_0 \in \mathcal{R}_0$, it is not the case that $S \cap \scp{R_0} \in \acc{R_0}$.
        We must have $S \cap \scp{R} \in \acc{R}$, where $R$ is the superset-closure of $R_0$.
        Consider the case that $R$ is a consistency relation.
        Then, its relation scope is $U^i_j \cup U^{i-1}_j \cup L^i_j \cup \compl{L}^{i-1}_j$ for some $i,j$.
        If $S \cap \scp{R_0} \notin \acc{R_0}$, then, we must have that $S \cap \scp{R_0}$ is a proper superset of some accepted assignment in $\acc{R_0}$.
        However, every accepting assignment of $\acc{R_0}$ ensures that we have exactly $\selStateVertexPerGadget$ selected vertices in $U^i_j$ and $U^{i-1}_j$, as well as $\weightConstant$ selected vertices of $L^i_j$, and $\compl{\weightConstant}$ selected vertices of $\compl{L}^{i-1}_j$.
        Hence, if $S$ actually selects a proper superset of some assignment of $\acc{R_0}$, it would have to select more vertices than that from one of these sets.
        But, then the size bound $k$ could not be met.
        The arguments for initial-relations, constraint-relations, and final-relations are similar, while the relations on $K$ are the same in $\mathcal{R}$ and $\mathcal{R}_0$.

        Hence, $S$ simply cannot afford to select a proper superset of any accepted assignment of any relation in $\mathcal{R}_0$, and hence, it fulfills all relations of $\mathcal{R}_0$.
    \end{claimproof}

    Next, we show that if some selection $S$ violates at most $\ell$ vertices and fulfills all relations of $\mathcal{R}_0$, then the input instance is a yes-instance.
    Observe that no assumption about the solution size is made here.

    \begin{claim}
        \label{claim:partial_minimization_high_lvl_backwards_assuming_relations}
        If there exists a set $S$ violating at most $\ell$ vertices that also fulfills all relations of $\mathcal{R}_0$, then
        the input instance $I$ is a yes-instance.
    \end{claim}
    \begin{claimproof}
        Let $S  \subseteq V(G)$ be a set that violates at most $\ell$ vertices, and fulfills all relations of $\mathcal{R}_0$, and thus also all relations of $\mathcal{R}$.
        By \cref{claim:partial_minimization_relations_fulfilled_size_property}, the size of $S$ is then at least $k$.
        Moreover, partly using properties established in the proof of \cref{claim:partial_minimization_relations_fulfilled_size_property}, it is not difficult to see that the relations of $\mathcal{R}_0$ ensure that each vertex of $K$ is selected, that for any $i \in \range{\bagCount}$ and $\variable{j} \in \bag{i}$ exactly $\#_\sigma = 2 \cdot \numOverflowAll \cdot \allLargestPartial + h_\sigma \cdot \sigLargestNonPartial + h_\rho \cdot \rhoLargestNonPartial$ vertices of $L^i_j \cup \compl{L}^i_j$ are selected, and that exactly $\selStateVertexPerGadget$ vertices of $U^i_j$ are selected.
        This means that exactly $k$ vertices of $G$ are selected by $S$.

        Regarding the number of violated vertices, since $S$ must select each vertex of $K$, it is the case that at most $\ell' = \sum_{i \in \range{\bagCount}} \sum_{\variable{j} \in \bag{i}} \numOverflowAll{}$ state vertices can be violated.
        Indeed, the violation status of any vertex that is not a state vertex is completely determined, because such vertices are either guaranteed to be violated, or guaranteed to be satisfied due to the selected neighbors in $K$.
        Moreover, the relations tightly control how many vertices are selected.
        Then, we obtain that at most $\ell'$ further vertices, in particular at most $\ell'$ state vertices, can be violated.

        Now, observe that the initial relations and consistency relations guarantee that the weight of $\state{\statevertex{i}{j}{1},\dots,\statevertex{i}{j}{g}}$ is $\weightConstant$ for every bag $\bag{i}$ and $\variable{j} \in \bag{i}$, and that $\numOverflowAll$ of the states in $\state{\statevertex{i}{j}{1},\dots,\statevertex{i}{j}{g}}$ are overflow states.
        Moreover, the consistency relations and final relations guarantee that $\numOverflowAll$ of the states in $\rightState{\statevertex{i}{j}{1},\dots,\statevertex{i}{j}{g}}$ are overflow states.

        For any state vertex $\statevertex{i}{j}{p}$ it is the case that if at least one of $\state{\statevertex{i}{j}{p}}$ and $\rightState{\statevertex{i}{j}{p}}$ is an overflow state, then $\statevertex{i}{j}{p}$ is violated.
        Hence, the relations guarantee that at least $\numOverflowAll$ vertices of $U^i_j$ are violated for any bag $\bag{i}$ and variable $\variable{j} \in \bag{i}$.
        Also, for any state vertex $\statevertex{i}{j}{p}$, we must have that $\state{\statevertex{i}{j}{p}}$ is an overflow state if and only if $\rightState{\statevertex{i}{j}{p}}$ is an overflow state, as otherwise, $S$ would violate more than $\numOverflowAll$ vertices of $U^i_j$, and then $S$ would have to violate more than $\ell'$ state vertices overall.
        In particular, this also means that whenever $\state{\statevertex{i}{j}{p}}$ or $\rightState{\statevertex{i}{j}{p}}$ is not an overflow state, then, the vertex $\statevertex{i}{j}{p}$ cannot be violated, as the violations due to the overflow states alone cause $\ell'$ violated state vertices.

        Now, fix some bag $\bag{i}$ and $x_j \in X_i$, where we have $x_j \in X_{i-1}$.
        It is of vital importance to prove that $\state{\statevertex{i}{j}{p}} = \state{\statevertex{i-1}{j}{p}}$ for all $p \in \range{g}$.
        If this is given, we can immediately extract a satisfying assignment for $I$, as the constraint relations then ensure that the states correspond to a satisfying assignment.

        The consistency relation $\consistencyRelation{i}{j}$ guarantees that vertex $\statevertex{i}{j}{p}$ is selected if and only if $\statevertex{i-1}{j}{p}$ is selected.
        So, to argue that their states are equivalent, it suffices to focus on the number of selected vertices in $\leftBlock{i}{j}{p}$ and $\leftBlock{i-1}{j}{p}$.
        First, we consider the case that $\statevertex{i-1}{j}{p}$ is selected, i.e., $\state{\statevertex{i-1}{j}{p}} = \sigma_c$.
        \begin{description}
            \item[$\sigma_c$ is an overflow state:] In this case, we immediately get that $\rightState{\statevertex{i-1}{j}{p}}$ is also an overflow-state, and then the consistency relation $\consistencyRelation{i}{j}$ ensures that also $\state{\statevertex{i}{j}{p}}$ is an overflow state.
                  Naturally, this also implies that $\weight{\state{\statevertex{i-1}{j}{p}}} = \weight{\state{\statevertex{i}{j}{p}}}$.

            \item[$\sigma_c$ is not an overflow state and $\sigma$ is finite:] This immediately implies that $\statevertex{i-1}{j}{p}$ has at most $\max \sigma$ selected vertices, so we must have $|S \cap L^{i-1}_j| + |S \cap \compl{L}^{i-1}_j| \leq \max \sigma$.
                  The relation $\consistencyRelation{i}{j}$ however ensures that $|S \cap \compl{L}^{i-1}_j| + |S \cap L^i_j| = \max \sigma$, which yields $\weight{\state{\statevertex{i-1}{j}{p}}} = |S \cap L^{i-1}_j| \leq |S \cap L^{i}_j| = \weight{\state{\statevertex{i}{j}{p}}}$.

            \item[$\sigma$ is simple cofinite:] In this case the situations flips, we know that $|S \cap L^{i-1}_j| + |S \cap \compl{L}^{i-1}_j| \geq \min \sigma$.
                  Moreover, the relation $\consistencyRelation{i}{j}$ guarantees that $\state{\statevertex{i}{j}{p}} = \complState{\rightState{\statevertex{i-1}{j}{p}}}$.
                  In particular, using the fact that $\consistencyRelation{i}{j}$ also ensures that at most $\min \sigma$ vertices of $\compl{L}^{i-1}_j$ are selected, this yields $|S \cap \compl{L}^{i-1}_j| + |S \cap L^i_j| = \min \sigma$,
                  resulting in $\weight{\state{\statevertex{i-1}{j}{p}}} = |S \cap L^{i-1}_j| \geq |S \cap L^{i}_j| = \weight{\state{\statevertex{i}{j}{p}}}$.
        \end{description}

        Let $P_\sigma$ be the set of positions of $\state{\statevertex{i}{j}{1},\dots,\statevertex{i}{j}{\blocksPerVertexGadget}}$ which are $\sigma$ states, and observe that these are also exactly those positions of $\state{\statevertex{i-1}{j}{1},\dots,\statevertex{i-1}{j}{\blocksPerVertexGadget}}$ that are $\sigma$ states.

        \begin{description}
            \item[$\sigma$ is finite:] Using our observations from above, we have that $\weight{\state{\statevertex{i-1}{j}{p'}}}  \leq \weight{\state{\statevertex{i}{j}{p'}}}$ for all $p' \in P_\sigma$.
                  However, the $\sigma$-weight of $\state{\statevertex{i-1}{j}{1},\dots,\statevertex{i-1}{j}{\blocksPerVertexGadget}}$ and of $\state{\statevertex{i}{j}{1},\dots,\statevertex{i}{j}{\blocksPerVertexGadget}}$ must both be $\selWeight$.
                  That is, we must have $\selWeight = \sum_{p' \in P_\sigma} \weight{\state{\statevertex{i-1}{j}{p'}}} = \sum_{p' \in P_\sigma} \weight{\state{\statevertex{i}{j}{p'}}}$.
                  Hence, we must actually have $\weight{\state{\statevertex{i-1}{j}{p}}} = \weight{\state{\statevertex{i}{j}{p}}}$, which directly yields that $\state{\statevertex{i-1}{j}{p}} = \state{\statevertex{i}{j}{p}}$.

            \item[$\sigma$ is simple cofinite:] We have that $\weight{\state{\statevertex{i-1}{j}{p'}}}  \geq \weight{\state{\statevertex{i}{j}{p'}}}$ for all $p' \in P_\sigma$.
                  However, the $\sigma$-weight of $\state{\statevertex{i-1}{j}{1},\dots,\statevertex{i-1}{j}{\blocksPerVertexGadget}}$ and of $\state{\statevertex{i}{j}{1},\dots,\statevertex{i}{j}{\blocksPerVertexGadget}}$ must both be $\selWeight$.
                  That is, we must have $\selWeight = \sum_{p' \in P_\sigma} \weight{\state{\statevertex{i-1}{j}{p'}}} = \sum_{p' \in P_\sigma} \weight{\state{\statevertex{i}{j}{p'}}}$.
                  Hence, we must actually have $\weight{\state{\statevertex{i-1}{j}{p}}} = \weight{\state{\statevertex{i}{j}{p}}}$, which directly yields that $\state{\statevertex{i-1}{j}{p}} = \state{\statevertex{i}{j}{p}}$.
        \end{description}

        The arguments for the case when $\statevertex{i}{j}{p}$ is unselected are analogous, keeping in mind that each vector of $A$ also has the same $\rho$-weight $\unselWeight$.
        So, the states of all state vertices of the same variables are the same.
        Then, the constraint relations guarantee that these correspond to a satisfying assignment of $I$.
    \end{claimproof}
    Now, we can easily show the backwards direction of correctness.

    \begin{claim}
        \label{claim:partial_minimization_high_lvl_backwards}
        If the output instance $I'$ is a yes-instance, then the input instance $I$ is a yes-instance.
    \end{claim}
    \begin{claimproof}
        Assume that $I'$ is a yes-instance, and let $S \subseteq V(G)$ be a selection of size at most $k$ that violates at most $\ell$ vertices of $V(G)$ and fulfills all relations of $\mathcal{R}$.
        Then, by \cref{claim:partial_minimization_relations_fulfilled_size_property} it must be the case that $S$ actually also fulfills all relations of $\mathcal{R}_0$.
        So, by \cref{claim:partial_minimization_high_lvl_backwards_assuming_relations}, the input instance $I$ is a yes-instance.
    \end{claimproof}

    Next, we show that whenever $\rho$ is simple cofinite, the instance is good (see \cref{def:partial_good_instance} for the definition of good instances).
    \begin{claim}
        \label{claim:good_instance}
        If $\rho$ is simple cofinite, then $I'$ is a good instance.
    \end{claim}
    \begin{claimproof}
        Let $S \subseteq V(G)$ be a set that fulfills all relations of $\mathcal{R}$.
        Then, \cref{claim:partial_minimization_relations_fulfilled_size_property} already establishes that $|S| \geq k$.
        If $\sigma$ and $\rho$ are both simple cofinite, then this is all that is necessary for the output instance to be good, as we also have $\ell = 0$ in that case.

        Otherwise, $\sigma$ is finite and $\rho$ is simple cofinite.
        Recall that each non-state vertex is adjacent to at least $\max \sigma$ vertices of the clique $K$.
        The relations of $\mathcal{R}$ then ensure that each vertex of $K$ is selected, and thus, every selected non-state vertex is violated due to having more than $\max \sigma$ selected neighbors.
        Furthermore, the relations guarantee that we have at least $\#_\sigma$ selected non-state vertices per bag and variable in the bag.
        So, we have at least $|K| + \sum_{i \in \range{t}} \sum_{\variable{j} \in \bag{i}} \#_\sigma$ selected non-state vertices that have more than $\max \sigma$ selected neighbors.

        The relations $\mathcal{R}$ also guarantee that $\numOverflowAll = \numOverflowSigma$ of the state vertex of each bag and variable in the bag have a state that is an overflow-state.
        Concretely, this means that these state vertices are themselves selected, and also have more than $\max \sigma$ selected neighbors.
        This means that at least $\sum_{i \in \range{t}} \sum_{\variable{j} \in \bag{i}} \numOverflowAll$ selected state vertices have more than $\max \sigma$ selected neighbors.
        Overall, we obtain that at least $\ell = |K| + \sum_{i \in \range{t}} \sum_{\variable{j} \in \bag{i}} (\numOverflowAll + \#_\sigma)$ vertices are selected with more than $\max \sigma$ selected neighbors.
    \end{claimproof}
    Note that we generally cannot prove that the output instance is good if $\rho$ is finite.
    For example, if $\rho$ is finite and $\sigma$ is simple cofinite, then we have $\ell \not = 0$, so there cannot be a set $D$ of size at least $\ell$ such that each vertex of $D$ is selected with more than $\max \sigma$ selected neighbors.

    Next, we show some further important properties of the reduction regarding its computation time and pathwidth.
    \begin{claim}
        The pathwidth of the output instance is $\pw \cdot \blocksPerVertexGadget + \Oh(1)$, where $\pw$ is the width of the path decomposition of the primal graph provided together with the input.
        Moreover, a path decomposition of the output instance of width $\pw \cdot \blocksPerVertexGadget + \Oh(1)$ can be computed in polynomial-time.
    \end{claim}
    \begin{claimproof}
        Recall that a path decomposition of the output instance is a path decomposition of the graph $G$, such that additionally, for each relation there exists a bag containing all vertices of the relation scope.
        Throughout the proof, we treat any set whose indices are out of bounds as the empty set.
        Also, recall that we set $\bag{0} = \bag{t+1} = \emptyset$ for convenience.

        Begin by copying the path decomposition provided with the input, denote the copy as $\outputBag{1},\dots,\outputBag{\bagCount}$.
        We begin by modifying the path decomposition by introducing two empty bags $\outputBag{0}$ and $\outputBag{\bagCount+1}$.
        In a first step, for any bag $\outputBag{i}$ and $\variable{j} \in \outputBag{i}$, we replace $\variable{j}$ with the state vertices $U^i_j$.

        Proceed as follows for any $i \in \range[1]{\bagCount+1}$.
        Order the variables of $\bag{i-1} \cup \bag{i}$ as \[\variable{\lambda_1^i},\dots,\variable{\lambda_r^i},\variable{\lambda_{r+1}^i},\dots,\variable{\lambda_q^i},\variable{\lambda_{q+1}^i},\dots,\variable{\lambda_z^i},\] such that for any $j \in \range{r}$, $\variable{\lambda_j^i}$ is in $\bag{i-1}$ but not in $\bag{i}$, for any $j \in \range[r+1]{q}$, we have that $\variable{\lambda_j^i}$ is in $\bag{i-1} \cap \bag{i}$, and for any $j \in \range[q+1]{z}$, we have that $\variable{\lambda_j^i}$ is in $\bag{i} \setminus \bag{i-1}$.
        Then, for any pair of adjacent bags $\outputBag{i-1}$ and $\outputBag{i}$, and any $j \in \range{z}$, we introduce initially empty bag $\outputBag{i}^j$.
        In the path decomposition, we put these new bags between $\outputBag{i-1}$ and $\outputBag{i}$ while keeping them in ascending order.
        That is, the bags $\outputBag{i-1},\outputBag{i}^1,\dots,\outputBag{i}^z,\outputBag{i}$ will form a subsequence of the path decomposition.
        The idea of these additional bags is that they should allow us to slowly transition from $\outputBag{i-1}$ to $\outputBag{i}$.
        So, we need to explain which vertices we put into them next.
        \begin{itemize}
            \item We first handle the vertices of variables in $\bag{i-1}$ which are not in $\bag{i}$.
                  For all $j \in \range{r}$, we add the vertices $\outputBag{i-1} \cup \compl{L}^{i-1}_{\lambda_j^i}$ to $\outputBag{i}^{j}$.
            \item Now, we want to handle the vertices of the variables in $\bag{i-1} \cap \bag{i}$.
                  For all $j \in \range[r+1]{q}$, we add the vertices $\bigcup_{j' \in \range[r+1]{j-1}} U^{i}_{\lambda_{j'}^i} \cup \bigcup_{j' \in \range[j+1]{q}} U^{i-1}_{\lambda_{j'}^{i}} \cup
                      U^i_{\lambda_j^i} \cup U^{i-1}_{\lambda_j^i} \cup \compl{L}^{i-1}_{\lambda_j^i} \cup L^i_{\lambda_j^i}$ to $\outputBag{i}^{j}$.
            \item We handle those vertices of the variables which are part of $\bag{i}$ but not part of $\bag{i-1}$.
                  For all $j \in \range[q+1]{z}$, we add the vertices $\outputBag{i} \cup L^i_{\lambda_j^i}$ to $\outputBag{i}^j$.
            \item If there is a constraint $C_\ell$ such that $\constraintBagMapping{C_\ell} = i$, then, we add all vertices of $\scp{\constraintRelation{\ell}}$ to the bag $\outputBag{i}^j$, for any $j \in \range{z}$, and also to the bag $\outputBag{i}$.
            \item  Finally, we add all vertices of $K$ to each bag.
        \end{itemize}
        This concludes the description of the path decomposition, we will give a brief example before proving that it is correct.
        \begin{example}
            For the input instance considered in \cref{fig:intermediate_construction} where the path decomposition of the primal graph is $\{x_1\},\{x_1,x_2\},\{x_1,x_2,x_3\},\{x_2\},\{x_2\}$, we would first add two empty bags at the start and end, and then replace all variables with their respective state-vertices, yielding the sequence of bags $\emptyset, U^1_1,U^2_1 \cup U^2_2, U^3_1 \cup U^3_2 \cup U^3_3, U^4_2, U^5_2,\emptyset$.
            Then, we would add the additional bags and vertices we require to transition between bags.
            Let us illustrate this step once for the transition between the bag $\outputBag{2}$ containing $U^2_1 \cup U^2_2$ and the bag $\outputBag{3}$ containing $U^3_1 \cup U^3_2 \cup U^3_3$.
            First, we order the variables of the bags $X_2 = \{x_1,x_2\}$ and $X_3 = \{x_1, x_2,x_3\}$.
            Since variables $x_1,x_2$ appear in both bags, and variable $x_3$ only appears in $X_3$, one potential order of these variables is $x_1,x_2,x_3$.
            The respective values $r$, $q$ and $z$ of this order fulfill $r = 0$, $q = 2$ and $z = 3$.
            We do not add any bags for variables which are in $\bag{3}$ but not in $\bag{2}$, since there are no such vertices.
            Then, we handle the vertices corresponding to the variables which are in $X_2 \cap X_3$.
            Here, we go through all $j \in \range[r+1]{q} = \{1,2\}$.
            First, we fix $j = 1$.
            We add a bag $\outputBag{3}^1$ containing $U^2_2 \cup U^3_1 \cup U^2_1 \cup \compl{L}^2_1 \cup L^3_1$, and then for $j = 2$ a bag $\outputBag{3}^2$ containing $U^3_1 \cup U^3_2 \cup U^2_2 \cup \compl{L}^2_2 \cup L^3_2$.
            Now, since $x_3$ is $\bag{3}$, but not in $\bag{2}$, we add another bag $\outputBag{3}^3$ that contains $U^3_1 \cup U^3_2 \cup U^3_3 \cup L^3_3$.
            If there is a constraint $C_\ell$ such that $\gamma(C_\ell) = 3$, then we add $\scp{\constraintRelation{\ell}}$ to bags $\outputBag{3}^1,\outputBag{3}^2,\outputBag{3}^3$ and $\outputBag{3}$.
            Finally, the vertices of the clique $K$ are added to all bags.
            Overall, the interval between the two considered bags $\outputBag{2}$ and $\outputBag{3}$ consists of the bags $\outputBag{2},\outputBag{3}^1,\outputBag{3}^2,\outputBag{3}^3,\outputBag{3}$.
        \end{example}

        Now, we formally show that the sketched sequence of bags is actually a valid path decomposition.
        First, we show that any vertex appears in some bag, and that for any initial-relation, final-relation and consistency-relation there exists some bag containing the relation scope.
        For a vertex of $K$ and all relations on $K$ this is easy, as all vertices of $K$ are in all bags.
        Each other vertex is in the scope of an initial-relation, a consistency-relation, or a final-relation.
        Observe that for any $\variable{j}$, if $\variable{j}$ is in $\bag{i} \setminus \bag{i-1}$, then, $\outputBag{i}^{j'}$ contains $U^i_j \cup L^i_j = \scp{\initialRelation{i}{j}}$, where $j'$ is so that $j = \lambda_{j'}^i$.
        Similarly, if $\variable{j}$ is both in $\bag{i}$ and $\bag{i-1}$, there exists bag $\outputBag{i}^{j'}$ that contains $U^i_j \cup L^i_j \cup U^{i-1}_j \cup \compl{L}^{i-1}_j = \scp{\consistencyRelation{i}{j}}$,
        where $j'$ is so that $j = \lambda_{j'}^i$.
        Finally, if $\variable{j}$ is in $\bag{i-1}$ but not in $\bag{i}$, then let $j'$ such that $j = \lambda^{i}_{j'}$.
        Then bag $\outputBag{i}^{j'}$ contains $U^{i-1}_j \cup \compl{L}^{i-1}_j = \scp{\finalRelation{i-1}{j}}$.
        This means that we indeed cover all initial-, consistency-, and final-relations, and thus also all vertices.
        Regarding the constraint relation $\constraintRelation{\ell}$, for any constraint $C_\ell$, observe that $\scp{\constraintRelation{\ell}}$ is contained in all bags between and including bag $\outputBag{\constraintBagMapping{C_\ell}}^1$ and $\outputBag{\constraintBagMapping{C_\ell}}$.

        It is also not difficult to see that each edge of $G$ appears in some bag.
        In particular, all edges of $K$ and going to $K$ are covered because each vertex of $V(G)$ is in some bag, and all vertices of $K$ are in all bags.
        Otherwise, an edge that is within a set $U^i_j$, $L^i_j$, or $\compl{L}^i_j$ is covered as these sets appear in some bag.
        The only other types of edges are between $U^i_j$ and $L^i_j$, or between $U^i_j$ and $\compl{L}^i_j$ (observe that there are no edges between $L^i_j$ and $\compl{L}^i_j$).
        The sets $U^i_j \cup L^i_j$ are always subject to an initial-relation or consistency-relation, and the sets $U^i_j \cup \compl{L}^i_j$ are always subject to a consistency-relation or a final-relation.
        Hence, we have already argued that these are covered.

        Finally, it remains to argue that for any vertex $v$, all nodes of bags containing $v$ form a connected subpath of the path decomposition.
        For the vertices of $K$, this is clear.
        For the vertices of the sets $L^i_j$ or $\compl{L}^i_j$, it is also the case that either they appear in exactly one bag, or they are subject to some constraint relation, in which case they appear in exactly the bags in which the scope of the constraint relation appears, which form a connected subpath.

        So, all that remains is arguing about vertices in a set $U^i_j$, which is the set of state-vertices of a variable $\variable{j} \in \bag{i}$.
        Clearly $U^i_j$ is part of $\outputBag{i}$.
        We now show that there is some integer $D_\ell$ such that $U^i_j$ is in all bags between and including $\outputBag{i}^{D_\ell}$ and $\outputBag{i}$, but in no bag before $\outputBag{i}^{D_\ell}$.
        
        Consider that $U^i_j$ is subject to a constraint relation, in that case $U^i_j$ is explicitly added to all bags between and including $\outputBag{i}^1$ and $\outputBag{i}$.
        Moreover, $U^i_j$ is not part of bag $\outputBag{i-1}$.
        Therefore, in this case $D_\ell = 1$.

        Otherwise, if we have that $\variable{j}$ is not in $\bag{i-1}$, then there is some $q$ (namely, the lowest integer $q$ such that $\variable{\lambda_{q+1}^i}$ is in $\bag{i} \setminus \bag{i-1}$) such that $U^i_j$ is part of the subsequence $\outputBag{i}^{q+1},\dots,\outputBag{i}$, but $U^i_j$ is not part of any bag before $\outputBag{i}^{q+1}$.
        In this case $D_\ell  = q + 1$

        Otherwise, if we have that $\variable{j}$ is in $\bag{i-1}$, then $U^i_j$ is in all bags of the subsequence $\outputBag{i}^{j'},\dots,\outputBag{i}$, but not part of any bag before $\outputBag{i}^{j'}$, where $j'$ is such that $j = \lambda_{j'}^i$.
        In this case $D_\ell = j'$.

        Now, we show that there is some integer $D_r$ such that $U^i_j$ is in all bags between and including $\outputBag{i}$ and bag $\outputBag{i+1}^{D_r}$.

        If we have that $\variable{j}$ is not in $\bag{i+1}$, then there is some integer $r$ (namely, the highest integer $r$ such that $\variable{\lambda_{r}^{i+1}}$ is not in $\bag{i+1}$), such that $U^i_j$ is contained in the subsequence of bags $\outputBag{i},\dots,\outputBag{i+1}^{r}$, but not in any bag after $\outputBag{i+1}^r$.
        In this case $D_r = r$.

        If we have that $\variable{j}$ is in $\bag{i}$ and $\bag{i+1}$, then let $j'$ be such that $j = \lambda^{i+1}_{j'}$, and observe that $U^i_j$ is in the subsequence $\outputBag{i},\dots,\outputBag{i+1}^{j'}$, but $U^i_j$ is not contained in any bag after $\outputBag{i+1}^{j'}$.
        In this case $D_r = j'$.
        
        Overall, we have that $U^i_j$ is exactly in the bags $\outputBag{i}^{D_\ell},\dots,\outputBag{i},\dots,\outputBag{i+1}^{D_r}$.
        It is now not difficult to see that our approach results in a valid path decomposition of width $\pw \cdot \blocksPerVertexGadget + \Oh(1)$, because our initial step of replacing the variables with the state vertices blows up the pathwidth by a factor of $\blocksPerVertexGadget$, and all the other steps ensure that each bag contains only $\pw \cdot \blocksPerVertexGadget + \Oh(1)$ state vertices, as well as a constant number of further vertices, where we utilize the fact that the sizes of $L^i_j$ and $\compl{L}^i_j$ for any $i,j$ as well as the size of $K$ are constant.
        It is also clear that we can compute the sketched decomposition in polynomial-time.
    \end{claimproof}

    \begin{claim}
        The arity of the output instance is at most $d$ for some constant $d$ depending only on $\varepsilon$ and $\sigma$ and $\rho$.
    \end{claim}
    \begin{claimproof}
        The relations upon $K$ have arity one.
        The other relations have a relation scope containing a constant number of sets of the form $U^i_j$, $L^i_j$, $\compl{L}^i_j$ (for some values $i,j$).
        Then, the claim follows from the fact that $\blocksPerVertexGadget$ is a constant only depending on $\varepsilon, \sigma$, and $\rho$, and that each block $\leftBlock{i}{j}{p}$ and $\rightBlock{i}{j}{p}$ also has constant size depending only on $\sigma$ and $\rho$, yielding that also the size of $U^i_j$, $L^i_j$ and $\compl{L}^i_j$ is constant for any $i,j$.
    \end{claimproof}

    Recalling that $\varepsilon$, $g$, and $q$ are constants, it is also not difficult to see that the reduction runs in polynomial-time.
    \begin{claim}
        The reduction runs in polynomial-time.
    \end{claim}

    \subparagraph*{Finishing the proof.}

    Assume that there is an $\varepsilon > 0$, and an algorithm
    that can solve instances $\tilde I$ of arity at most $d$, which additionally must be good when $\rho$ is simple cofinite, of \minParGenDomSetRel{} in time $(|\allStatesPartial| - \varepsilon)^\pw \cdot |\tilde I|^{\Oh(1)}$ when provided with a path decomposition of width $\pw$ of $\tilde I$.

    We take an instance $I$ of 2-\problem{CSP} over the alphabet $A$ as input, and utilize the reduction of this proof for the constant $\varepsilon$.
    This produces an equivalent instance $I'$ with pathwidth $\pw \cdot g + \Oh(1)$, where $\pw$ is the pathwidth of the decomposition of the primal graph of $I$ that is provided together with the input.
    We then apply our hypothetical fast algorithm on instance $I'$.
    This whole procedure decides instance $I$ in time
    \begin{align*}
        |I|^{\Oh(1)} + (|\allStatesPartial| - \varepsilon)^{\pw \cdot g + \Oh(1)} \cdot |I'|^{\Oh(1)} & \leq (|\allStatesPartial| - \varepsilon)^{\pw \cdot g} \cdot |I|^{\Oh(1)} \\
                                                                                                      & \leq (|A| - \varepsilon')^\pw \cdot |I|^{\Oh(1)},
    \end{align*}
    for constant $\varepsilon' > 0$.
    As we have that $|A| \geq 3$, this would contradict the \pwseth{} by \cref{thm:lampis}.
\end{proof}

Next, we show that the same type of statement also holds for the problem \parGenDomSetRel{}, and even when $\rho$ is finite.
\thmDecisionHighLvl*{}
\begin{proof}
    The used construction is identical to the one of the proof of \cref{thm:minimization_high_lvl}, with the exception that we actually use the relations of $\mathcal{R}_0$ for our output instance, and that we do not require an upper bound on the solution size.

    The correctness for the forward direction then directly follows from the proof of \cref{claim:minimization_high_lvl_forward}, which did not use the property that $\rho$ is simple cofinite.
    For the backwards direction, \cref{claim:partial_minimization_high_lvl_backwards_assuming_relations} proves correctness, which also does not exploit the fact that $\rho$ is simple cofinite.

    The bound on the pathwidth, the polynomial-time computability of the reduction, and the arity bound also work identically.
    Then, the concluding proof of the lemma directly follows in the same manner as well.
\end{proof}

 \label{sec:intermediate_lb_partial}

\section{Replacing Relations with Gadgets}
\label{sec:realizing_relations}
To complete our lower bound, we need to get rid of the set of relations that the problem instances resulting form the intermediate lower bound can have.
All reductions in this section have the property that they do not change the input value $\ell$ of allowed violations.
Hence, although we craft them for the partial problem, they will automatically also work for the nonpartial problem, which is the same as the partial problem with $\ell = 0$.
In that sense, the reductions we provide here are quite powerful: they can handle violated vertices, but do not rely on violating vertices to function.
Our reduction chain is different depending on whether $\rho$ is finite or not.
We first cover the case when $\rho$ is finite in \cref{sec:replacing_relations_rho_finite}, and later handle the case when $\rho$ is simple cofinite in \cref{sec:replacing_relations_rho_simple_cofinite}.

\subsection{Set \texorpdfstring{$\mathbf{\rho}$}{Rho} is Finite}
\label{sec:replacing_relations_rho_finite}

In this case, \cref{thm:minimization_high_lvl} cannot be used, as although most proofs of that lemma also work for finite $\rho$, the output instance of the provided construction would not necessarily be a good instance in that case, which makes things difficult.
Luckily, it will turn out that we can actually realize arbitrary relations well in this setting.
This means that it is more convenient to use the lower bound from \cref{thm:decision_high_lvl}, as we do not have to worry about the solution size of the input problem.

Our strategy will be to first utilize a result by Focke et al. \cite{fockeTightComplexityBoundsLowerBound}, which allows us to replace the relations of the input instance by $\hwRelation{}$-relations.
A relation $R$ is a $\hwRelation{}$-relation if and only if $\acc{R} = \{X \mid X \subseteq \scp{R}, |X| = 1\}$, that is, a relation that accepts a selection from its scope if and only if it has size exactly one.

\thmReplacingArbitraryHWOne*{}
\begin{proof}
    Focke et al. \cite[Corollary 8.8]{fockeTightComplexityBoundsLowerBound} show that, for the nonpartial problem and all constants $d$, there is a pathwidth-preserving reduction from the problem with arbitrary relations of arity at most $d$ to the problem using only $\hwRelation{}$ relations, such that the arity of the output instance is at most $2^d + 1$.

    We can actually use the same reduction they use, with the only difference being that our input instance contains the input integer $\ell$ of allowed violations, and the output instance contains the same integer $\ell$ of allowed violations.

    We now briefly explain why we can use their reduction even for the partial problem without going into too much depth to avoid lengthy explanations which are identical to those already given in the literature.
    In case more details are desired, we invite the reader to inspect the proof given in \cite{fockeTightComplexityBoundsLowerBound}.
    Within the reduction, Focke et al. modify the input instance only by adding additional relations, which are $\hwRelation$-relations, and copies of the useful provider $H$ of \cref{thm:sigma_s_rho_s_provider_gadget}.
    Note that the assumption $\rho \neq \{0\}$ is used to justify the existence of the useful provider.
    Each copy of $H$ is not connected to any vertex outside $H$ with an edge, but subject to relations.

    Crucially, for any copy of $H$ only the portal $u$ of $H$ can be part of a relation that also has vertices outside $H$ in its scope.
    Moreover, the graph $H$ has a selection that selects $u$ and satisfies all vertices of $H$, and a selection that does not select $u$ and satisfies all vertices of $H$.

    For the forward direction, one can then easily extend a solution $S$ of the input instance violating at most $\ell$ vertices to a solution $S'$ of the output instance violating at most $\ell$ vertices by using the same approach Focke et al. use.
    The selections $S$ and $S'$ are identical on the vertex set of the input graph.
    In particular, in any copy of $H$, one always has a selection making all vertices of $H$ satisfied, so $S'$ will only violate exactly those vertices of the output instance which $S$ already violated.

    For the backwards direction, it can be shown that any set $S'$ fulfilling all relations of the output instance must fulfill all relations of the input instance when restricted to the vertex set of the input instance.
    Moreover, when $S'$ violates at most $\ell$ vertices of the output instance, then $S'$ restricted to the vertices of the input graph can also violate at most $\ell$ vertices there, as no vertex of a copy of $H$ is adjacent to any vertex of the input graph.
\end{proof}

So, we may assume that the input instances of \parGenDomSetRel{} only use $\hwRelation{}$-relations.
We must thus now only show that we can simulate these relations with appropriate graphs.
We begin by introducing a very useful gadget for our purposes.
\begin{definition}[Tremendous Gadget]
    \label{def:tremendous_gadget}
    Let $\sigma$, $\rho$ be sets of non-negative integers.
    A tremendous gadget is a triple $(H,u, \portalSelectionGadgetConstant)$,
    where $H$ is a graph, $\portalSelectionGadgetConstant$ a non-negative integer, and $u \in V(H)$ is called the portal, such that any set $S \subseteq V(H)$
    \begin{enumerate}
        \item selects $u$ and at least $\portalSelectionGadgetConstant$ vertices of $H$, or
        \item violates a vertex of $V(H) \setminus \{u\}$, or
        \item violates $u$ relative to both $(\sigma, \rho)$ and $(\sigma - 1, \rho - 1)$.
    \end{enumerate}
    Moreover, at least one set $S$ of size $\portalSelectionGadgetConstant$ that violates no vertex of $H$ exists.
\end{definition}

Now, we show that these type of gadgets actually exist when we need them.

\begin{figure}
    \centering
    \begin{subfigure}{0.45\linewidth}
        \centering
        \includegraphics{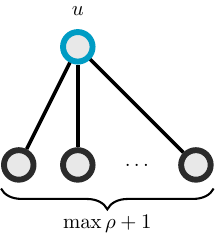}
        \subcaption{The tremendous gadget when $\min \sigma = 0$ and $\min \rho = 1$.}
    \end{subfigure}
    \hfill
    \begin{subfigure}{0.45\linewidth}
        \centering
        \includegraphics{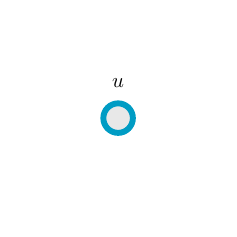}
        \subcaption{The tremendous gadget when $\min \sigma = 0$ and $\min \rho = 2$.}
    \end{subfigure}

    \begin{subfigure}{0.45\linewidth}
        \centering
        \includegraphics{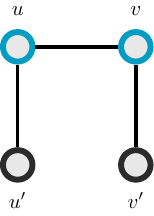}
        \subcaption{The tremendous gadget when $\min \sigma = \min \rho = 1$.}
    \end{subfigure}
    \hfill
    \begin{subfigure}{0.45\linewidth}
        \centering
        \includegraphics{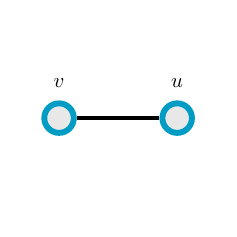}
        \subcaption{The tremendous gadget when $\min \sigma = 1$ and $\min \rho = 2$.}
    \end{subfigure}

    \begin{subfigure}{0.45\linewidth}
        \centering
        \includegraphics{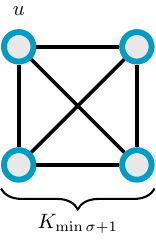}
        \subcaption{The tremendous gadget when $\min \sigma \geq 2$.}
    \end{subfigure}
    \caption{The tremendous gadgets of \cref{thm:tremendous_gadgets_exist} and a selection of size $\portalSelectionGadgetConstant$ that satisfies all vertices of the gadget. Selected vertices are drawn with a blue border.}
    \label{fig:tremendous_gadget}
\end{figure}

\begin{lemma}
    \label{thm:tremendous_gadgets_exist}
    Let $\sigma$ be a finite or simple cofinite set and $\rho$ a finite set with $0 \notin \rho$.
    Then, there is a tremendous gadget.
\end{lemma}
\begin{proof}
    Recall that a tremendous gadget is a triple $(H,u,\portalSelectionGadgetConstant)$, where $u \in V(H)$, and $\portalSelectionGadgetConstant$ is a integer.
    We consider different cases depending on $\sigma$ and $\rho$, consult \cref{fig:tremendous_gadget} for illustrations of the gadget in each case.

    \begin{description}
        \item[$\min \sigma = 0$ and $\min \rho = 1$:]               In this case $H$ consists of a vertex $u$ with $\max \rho + 1$ pendants, and $\portalSelectionGadgetConstant = 1$.
              Assume that $S \subseteq V(H)$ violates no vertex of $V(H) \setminus \{u\}$.
              If $u$ is not selected, then, each pendant of $u$ must be selected.
              But then $u$ has $\max \rho + 1$ selected neighbors, which violates $u$, both relative to $\rho$ and $\rho - 1$.
              So, otherwise, $u$ must be selected, and, we must thus select at least $\portalSelectionGadgetConstant$ vertices of $H$.
              Finally, if we consider the set $S = \{u\}$, this set has size exactly $\portalSelectionGadgetConstant$, and satisfies every vertex of $H$.

        \item[$\min \sigma = 0$ and $\min \rho \geq 2$:] The graph $H$ simply consists of a single vertex $u$, and $\portalSelectionGadgetConstant  = 1$. Observe that if $u$ is selected, it is satisfied. Moreover, if $u$ is not selected, then it is violated relative to both $\rho$ and $\rho - 1$, as $0$ is neither in $\rho$ nor in $\rho - 1$.

        \item [$\min \sigma = 1$ and $\min \rho = 1$:]
              Now consider the case that $\min \sigma = \min \rho = 1$.
              Then, $H$ consists of a path on four vertices $u',u,v,v'$, and $\portalSelectionGadgetConstant  = 2$.
              Assume that $S \subseteq V(H)$ violates no vertex of $V(H) \setminus U$.
              Since vertices $u'$ and $v'$ have degree one, it must be the case that both $u$ and $v$ are selected.
              Hence, $S$ must select $u$ and at least $\portalSelectionGadgetConstant$ vertices of $H$.
              Moreover, all vertices are satisfied when $u$ and $v$ are selected.
        \item [$\min \sigma = 1$ and $\min \rho \geq 2$:]
              Then, $H$ consists of vertices $v$ and $u$ connected by an edge, and $\portalSelectionGadgetConstant = 2$.
              If $v$ is not violated, then both $u$ and $v$ must be selected.
              So, either we select both $u$ and $v$, or we violate a vertex of $H \setminus \{u\}$.
              Moreover, all vertices of $H$ are satisfied when $u$ and $v$ are selected.

        \item [$\min \sigma \geq 2$:]
              In this case $H$ is a clique on $\min \sigma + 1$ vertices, $u$ is an arbitrary vertex of the clique, and $\portalSelectionGadgetConstant = \min \sigma + 1$.\footnote{Note that using cliques of size $\min \sigma + 1$ for gadgets is a staple for these types of problems, and was, for example, also used in \cite{fockeTightComplexityBoundsLowerBound,greilhuberResidueDominationBoundedTreewidth2025}.}

              Consider an arbitrary selection $S$ which violates no vertex of $H \setminus \{u\}$, and also does not violate $u$ relative to both $(\sigma, \rho)$ and $(\sigma - 1, \rho - 1)$.
              Observe that if any vertex $v \not = u$ of the clique is selected, then all vertices of the clique must be selected to satisfy $v$.
              So, either $S = \{u\}$, or $S = \emptyset$, or $S$ must select every vertex of the clique.

              If we have that $S = \emptyset$, then, this would clearly violate a vertex  of $V(H) \setminus \{u\}$, as $0 \notin \rho$.

              If we have that $S = \{u\}$, then, $u$ would be violated relative to both $\sigma$ and $\sigma - 1$ because $\min \sigma \geq 2$, but $u$ has no selected neighbors in the graph $H$.

              Hence, it must be the case that $S = V(H)$, which means that $S$ selects $u$ and $\portalSelectionGadgetConstant$ vertices of $H$.
              Finally, if we select all vertices of the clique $H$ all of them are satisfied, as each vertex is then selected with $\min \sigma$ selected neighbors.
    \end{description}
\end{proof}

We now utilize the tremendous gadget to build another type of gadget, which we call the robust realization gadget.

\defRobustRealization*{}

We now show that these types of gadgets exist for arbitrarily large penalties.
The construction is reminiscent of gadgets used by Focke et al. \cite{fockeTightComplexityBoundsLowerBound}, and can be seen as an extension of the gadgets they use to simulate the $\hwRelation$-relations, which now also take violations and solution sizes into account.

\thmRobustRelationGadgetsExist*{}
\begin{proof}
    The graph $G$ consists of a vertex $v$, that is adjacent to the portal of $\min \rho - 1$ fresh (graphs of the) copies of the tremendous gadget from \cref{thm:tremendous_gadgets_exist}.
    Set $t = \penaltySym \cdot \portalSelectionGadgetConstant + 2\penaltySym + 1$, and $\beta = \portalSelectionGadgetConstant$, where $\portalSelectionGadgetConstant$ is the constant of the tremendous gadget.
    Then, $v$ is also adjacent to $t$ vertices $r_{1},\dots,r_{t}$, each of them adjacent to the portals of $\max \rho$ fresh copies of the tremendous gadget.

    Similarly, $G$ contains a vertex $w$, adjacent to $\max \rho - 1$ portals of fresh copies of the tremendous gadget, and $w$ is furthermore adjacent to $t$ vertices $z_{1},\dots,z_{t}$, each of them adjacent to $\max \rho$ portals of fresh copies of the tremendous gadget.

    Vertices $w$ and $v$ are adjacent to the fresh vertices $U = \{u_1,\dots,u_d\}$, the vertices $U$ themselves form an independent set.
    Let $\#_\uptextsf{t}$ be the number of copies of the tremendous gadget.
    The value $\gamma = \#_\uptextsf{t} \cdot \portalSelectionGadgetConstant$, which is in $\Oh(\delta)$ as $\portalSelectionGadgetConstant$ is a constant.

    We now show that this pair $(G,U)$ has the desired properties.
    Let $H$ be a copy of a tremendous gadget with portal $p$, and $S \subseteq V(G)$.
    From \cref{def:tremendous_gadget}, we get certain guarantees regarding $S \cap V(H)$, we show that these guarantees extend to the whole set $S$ in graph $G$.
    If $S \cap V(H)$ violates a vertex of $V(H) \setminus \{p\}$, then, $S$ violates a vertex of $V(H) \setminus \{p\}$, because no vertex of $V(H) \setminus \{p\}$ has a neighbor outside $H$.
    Next, consider the case that $S \cap V(H)$ violates $p$ relative to both $(\sigma, \rho)$ and $(\sigma - 1, \rho - 1)$.
    In graph $G$, vertex $p$ has only a single neighbor outside $H$, call this neighbor $p'$.
    If $S$ does not select $p'$, then the number of selected neighbors of $p$ is the same in $H$ under $S \cap V(H)$ and $G$ under $S$.
    Hence, $p$ is violated by $S$ relative to $(\sigma, \rho)$ in $G$, because $p$ is violated by $S \cap V(H)$ relative to $(\sigma, \rho)$ in $H$.
    Otherwise, $S$ selects $p'$; then the number of selected neighbors of $p$ is exactly one larger in $G$ under $S$ than in $H$ under $S \cap V(H)$.
    Since $p$ is violated by $S \cap V(H)$ relative to $(\sigma - 1, \rho - 1)$ in $H$, this implies that $p$ is violated by $S$ relative to $(\sigma, \rho)$ in $G$.
    Hence, by \cref{def:tremendous_gadget}, any selection $S \subseteq V(G)$ that violates no vertex of $H$ contains at least $\portalSelectionGadgetConstant$ vertices of $H$, and moreover selects $p$.

    Furthermore, if a set $S \subseteq V(G)$ violates at most $\ell$ vertices of $V(G) \setminus U$ for some $\ell$, then it is clear that at least $\#_\uptextsf{t} - \ell$
    copies of the tremendous gadget must have all of their vertices non-violated.
    In particular this means that $S \setminus U$ must have size at least $(\#_\uptextsf{t} - \ell) \cdot \portalSelectionGadgetConstant = \gamma - \ell \cdot \portalSelectionGadgetConstant = \gamma - \ell \cdot \beta$.
    This covers the second property of \cref{def:robust_realization}.

    Next, we show that for any $u_i \in U$, there exists a set $S$ which violates no vertex of $V(G) \setminus U$ that fulfills $|S \setminus U| = \gamma$, selects no vertex of $U$ apart from $u_i$, and selects no vertex of $N(U)$.
    Indeed, we must simply select vertex $u_i$, and the solution of size $\portalSelectionGadgetConstant$ in each tremendous gadget that selects the portal vertex.
    Then, each vertex of a tremendous gadget is satisfied, because they receive no additional selected neighbors, and the selection within the gadget satisfies each vertex.
    Vertex $r_i$ for any $i \in \range{t}$ is satisfied because it has exactly $\max \rho$ selected neighbors and is unselected, for the same reason vertex $z_i$ is also satisfied.
    Finally, the vertex $v$ has exactly $\min \rho$ selected neighbors, and $w$ exactly $\max \rho$ selected neighbors, which again satisfies these vertices.
    This covers the first property of \cref{def:robust_realization}.

    Next, consider an arbitrary selection $S$ such that $|S \setminus U| \leq \gamma$ and $S$ violates no vertex of $V(G) \setminus U$.
    We show that then, we must have $|S \cap U| = 1$.
    All portal vertices of tremendous gadgets must be selected by $S$.
    Moreover, we select $\gamma$ vertices in the tremendous gadgets alone, hence, no vertex of $V(G) \setminus U$ which is not part of a tremendous gadget is selected.
    In particular, this means that unselected vertex $v$ has at most $\min \rho - 1$ selected neighbors from the set $V(G) \setminus U$, so it is violated unless $|S \cap U| > 0$.
    Similarly, unselected vertex $w$ has at least $\max \rho - 1$ selected neighbors from the set $V(G) \setminus U$, so $w$ is violated if $|S \cap U| > 1$.
    Hence, we must have $|S \cap U| = 1$.
    This shows the third property of \cref{def:robust_realization}.

    Finally, consider an arbitrary selection $S$ such that $|S \setminus U| \leq \gamma + \penaltySym$, and $S$ violates at most $\penaltySym$ vertices of $V(G) \setminus U$.
    Observe that this means that in all apart from $\penaltySym$ copies of the tremendous gadget, $S$ must select at least $\portalSelectionGadgetConstant$ vertices.
    This means that $S$ must select at least $\gamma - \penaltySym \cdot \portalSelectionGadgetConstant$ vertices from the copies of the tremendous gadgets alone.
    Recall that the value $t = \penaltySym \cdot \portalSelectionGadgetConstant + 2\penaltySym + 1$.
    In particular, since $S \setminus U$ has size at most $\gamma + \penaltySym$, $S$ can select at most $\penaltySym \cdot \portalSelectionGadgetConstant + \penaltySym = t - (\delta + 1)$ vertices which are not part of a tremendous gadget.
    So, there exists an index set $I \subseteq \range{t}$ such that $|I| \geq \penaltySym + 1$, and for each $i \in I$, vertex $r_i$ is not selected.
    Given that $S$ violates at most $\penaltySym$ vertices, there must even exist an $i$ such that vertex $r_i$ is not selected, none of the $\max \rho$ tremendous gadgets adjacent to $r_i$ contain a violated vertex, and $r_i$ is not violated.
    Hence, $r_i$ has $\max \rho$ selected neighbors, and vertex $v$ cannot be selected, as this would violate $r_i$.
    An analogous argument shows that also $w$ cannot be selected.
    Hence, we have $S \cap N(U) = \emptyset$.
    This proves the fourth property of \cref{def:robust_realization}.

    Regarding the size of the gadget, observe that $t$ is polynomial in $\penaltySym$, and that each tremendous gadget has constant size.
    Then, it is easy to confirm that the construction can be computed in time polynomial in $\penaltySym + d$.

    Finally, let us argue about the pathwidth of $G$.
    For each tremendous gadget, we create a bag containing all vertices of the gadget.
    We add the vertices $U$, $v$, $w$, as well as all vertices of the $\min \rho - 1$ tremendous gadgets attached to $v$, as well as all vertices of the $\max \rho - 1$ tremendous gadgets attached to $w$ to all bags.
    Observe that so far, each bag only contains $\Oh{(d)}$ number of vertices, because $\max \rho, \min \rho$ and the size of each tremendous gadget are constant.

    All that is left is to cover the vertices $r_1,\dots,r_t$ and $z_1,\dots,z_t$ and their incident edges.
    For this purpose, we order the bags created thus far such that, for any $i \in \range{t}$, the nodes of the bags containing the vertices of the tremendous gadgets attached to $r_i$ and $z_i$ form a connected subpath.
    Since each tremendous gadget is attached to exactly one vertex, this can be done easily.
    Now, for any $i$, we simply add the vertex $r_i$ to all the bags containing a tremendous gadget attached to $r_i$, and we similarly proceed for $z_i$.
    This covers vertices $r_i$, $z_i$ and their incident edges for any $i \in \range{t}$ while ensuring that the bags containing $r_i$ and $z_i$ form a connected subpath.
    Moreover, for any $i,i' \in \range{t}$ with $i \not = i'$ and any bag, at most one of $r_i,r_{i'},z_i,z_{i'}$ can appear in the bag.
    Hence, the sketched path decomposition has width $\Oh{(d)}$, and it is easy to see that it can be computed in time polynomial in the size of the gadget.
\end{proof}

This gadget already suffices for the pathwidth preserving reduction from the problem \parGenDomSetRel{} using only $\hwRelation$-relations to \minParGenDomSet{}, which we state next.
\thmRhoFinitePartialRealizeRelationsReduction*{}
\begin{proof}
    Let $I = (G, \ell, \mathcal{R})$ be the input instance of \parGenDomSetRel{}.
    Let $\beta$ be the tradeoff constant from \cref{thm:robust_relations_exist}.
    We utilize the robust realization provided by \cref{thm:robust_relations_exist} for the penalty value $\delta = |V(G)| + \ell + \ell \cdot \beta + 1$ (and multiple arity sizes, depending on the used relations).
    Let $\gamma$ be the corresponding cost value from \cref{thm:robust_relations_exist}.

    For any relation $R \in \mathcal{R}$, we add $\delta$ copies of the robust realization gadget with arity $|\scp{R}|$ and penalty $\delta$.
    Denote these copies as $G_R^1,\dots,G_R^{\delta}$.
    Identify the set $U$ of each of these copies with $\scp{R}$.
    Let $G'$ be the resulting graph.
    Then, set $k$ to $|V(G)| + |\mathcal{R}| \cdot \delta \cdot \gamma$.
    The output instance is $I' = (G',k,\ell)$.

    \begin{claim}
        If $I$ is a yes-instance, then $I'$ is a yes-instance.
    \end{claim}
    \begin{claimproof}
        Let $S$ be a solution to the input instance.
        Then, recalling that $S$ selects exactly one vertex of the scope of each relation, we can easily extend $S$ to a solution $S'$ for the output instance.
        To do this, we simply utilize the solution that selects exactly one vertex of the relation scope and $\gamma$ vertices in the rest of the gadget for each copy of the robust realization gadget.
        This selection satisfies all vertices which are not part of $V(G)$.
        Moreover, since the set $U$ of the realization gadget is an independent set, and the selection within each realization gadget does not select neighbors of $V(G)$, the vertices of $V(G)$ have exactly the same number of selected neighbors in $G$ under $S$ and $G'$ under $S'$.
        Hence, $S'$ violates at most $\ell$ vertices.
        Regarding the size, observe that $S'$ has size at most $k$ as $|S| \leq |V(G)|$.
    \end{claimproof}

    \begin{claim}
        If $I'$ is a yes-instance, then $I$ is a yes-instance.
    \end{claim}
    \begin{claimproof}
        Let $S$ be a solution for the output instance, that is, a set of size at most $k$ that violates
        at most $\ell$ vertices.
        Set $S' = S \cap V(G)$, our goal is showing that $S'$ is a solution for the input instance.

        For an arbitrary $R \in \mathcal{R}$, and $i \in \range{\delta}$ let $\ell_R^i$ be the number of violated vertices of $V(G_R^i) \setminus \scp{R}$.
        It is clear that $\ell_R^i \leq \sum_{R' \in \mathcal{R}, i' \in \range{\delta}} \ell_{R'}^{i'} < \ell + 1 < \delta.$
        From the tradeoff value of $\beta$ we now obtain that the size of $S \cap (V(G_R^i) \setminus \scp{R})$ is at least $\gamma - \ell_R^i \cdot \beta$.
        Now, consider the case that $S$ selects at least $\gamma + \delta$ vertices of $V(G_R^i) \setminus \scp{R}$.
        We will show that then, the solution $S$ must be larger than $k$.

        For this purpose, let $P = \bigcup_{R' \in \mathcal{R}, i' \in \range{\delta}} V(G_{R'}^{i'}) \setminus \scp{R'}$ be the set of vertices of all copies of the robust realization gadget (excluding the relation scopes).
        Using the tradeoff value we must have
        \begin{align*}
            |S \cap (P \setminus V(G_R^i))| & \geq \sum_{R' \in \mathcal{R}, i' \in \range{\delta}, (R',i') \not = (R,i)} |S \cap V(G_{R'}^{i'}) \setminus \scp{R'}|                                 \\
                                            & \geq \sum_{R' \in \mathcal{R}, i' \in \range{\delta}, (R',i') \not = (R,i)} \gamma - \ell_{R'}^{i'} \cdot \beta                                        \\
                                            & \geq (|\mathcal{R}| \cdot \delta - 1) \cdot \gamma - \beta \cdot \sum_{R' \in \mathcal{R}, i' \in \range{\delta}, (R',i') \not = (R,i)} \ell_{R'}^{i'} \\
                                            & \geq (|\mathcal{R}| \cdot \delta - 1)\cdot \gamma - \ell \cdot \beta.
        \end{align*}
        From this, we directly obtain
        \begin{align*}
            |S \cap P| & \geq (|\mathcal{R}| \cdot \delta - 1) \cdot \gamma - \ell \cdot \beta + \gamma + \delta \\
                       & = |\mathcal{R}| \cdot \delta \cdot \gamma - \ell \cdot \beta + \delta                   \\
                       & = |\mathcal{R}| \cdot \delta \cdot \gamma + |V(G)| + \ell + 1 > k.
        \end{align*}
        So, $S$ is larger than $k$, a contradiction.

        Hence, we know that for any $R \in \mathcal{R}$ and $i \in \range{\delta}$, $S$ selects fewer than $\gamma + \delta$ vertices of $V(G_R^i) \setminus \scp{R}$, and violates fewer than $\delta$ vertices of $V(G_R^i) \setminus \scp{R}$.
        It then follows directly from the properties of a robust realization gadget, and the fact that no vertex of $V(G_R^i) \setminus \scp{R}$ has a neighbor outside $G_R^i$, that $S$ cannot select any vertex of $N(\scp{R}) \cap V(G_R^i)$.
        In particular, $S$ does not select any vertex of $N(V(G))$.
        Then, since $S$ selects no neighbors of $N(V(G))$, the vertex set $U$ of a realization gadget is an independent set, and $S$ violates at most $\ell$ vertices, it must be the case that $S'$ violates at most $\ell$ vertices of the input instance.

        So, all that is left is to confirm that $S'$ also fulfills all relations of the input instance.
        Assume that for some $R \in \mathcal{R}$ and all $i \in \range{\delta}$, $S$ violates a vertex of $V(G_R^i) \setminus \scp{R}$ or selects more than $\gamma$ vertices of $V(G_R^i) \setminus \scp{R}$.
        Without loss of generality, let $G_R^1,\dots,G_R^x$ be those realization gadgets for which $S$ violates no vertex outside $\scp{R}$.
        Moreover, let $P' = \bigcup_{i \in \range{x}} V(G_R^i) \setminus \scp{R}$.
        We clearly have that $\{P \setminus P', P'\}$ is a partition of $P$.
        So $|S \cap P| = |S \cap (P \setminus P')| + |S \cap P'|$.
        We can easily bound the quantities on the right side of this equation.

        Observe that the number of realization gadgets that make up the set $P \setminus P'$ is exactly $(|\mathcal{R}| -1) \cdot \delta + (\delta - x)$.
        Hence, once again using the tradeoff, we obtain
        \begin{align*}
            |S \cap (P \setminus P')| \geq (|\mathcal{R}| - 1) \cdot \delta \cdot \gamma  + (\delta - x) \gamma - \ell \cdot \beta.
        \end{align*}
        Furthermore, from each of the $x$ gadgets $G_R^1,\dots, G_R^x$ we know that $S$ selects at least $\gamma + 1$ vertices outside $\scp{r}$, so
        $
            |S \cap P'| \geq x \cdot (\gamma + 1).
        $
        Overall, we now have
        \begin{align*}
            |S \cap P| & \geq (|\mathcal{R}| - 1) \cdot \delta \cdot \gamma  + (\delta - x) \gamma - \ell \cdot \beta + x \cdot (\gamma + 1) \\
                       & = |\mathcal{R}| \cdot \delta \cdot \gamma - \ell \cdot \beta + x.
        \end{align*}
        Now, we know that $x \geq \delta - \ell = |V(G)| + \ell \cdot \beta + 1$.
        So, in fact $|S| \geq |\mathcal{R}| \cdot \delta \cdot \gamma + |V(G)| + 1 > k$, a contradiction.

        Hence, we see that for any $R \in \mathcal{R}$ there exists at least one value $i$ such that no vertex of $V(G_R^i) \setminus \scp{R}$ is violated, and such that no more than $\gamma$ vertices of $V(G_R^i) \setminus \scp{R}$ are selected.
        Then, by the third property of the realization gadget, it must be the case that exactly one vertex of $\scp{R}$ is selected.
        This means that $S'$ also fulfills all relations of $\mathcal{R}$, and thus the input instance is a yes-instance.
    \end{claimproof}

    \begin{claim}
        The reduction is pathwidth-preserving.
    \end{claim}
    \begin{claimproof}
        It is not hard to see that the reduction runs in polynomial-time, as we only add some polynomial-number of polynomial-sized gadgets.

        Regarding the pathwidth, observe that the path decomposition provided with the input graph contains a bag containing $\scp{R}$ for any $R \in \mathcal{R}$.
        Moreover, the pathwidth of the realization gadgets we added to the graph is actually constant (as $d$ is a constant), and we can easily compute a path decomposition of these gadgets.
        Hence, we must only copy the bag containing all vertices of $\scp{R}$ sufficiently many times, and overlap these with the path decomposition for the realization gadgets of $R$.

        This way, we only increase the pathwidth by the width of the decomposition of the realization gadgets, which is constant.
    \end{claimproof}
\end{proof}

\subsection{Set \texorpdfstring{$\mathbf{\rho}$}{Rho} is Simple Cofinite}
\label{sec:replacing_relations_rho_simple_cofinite}

Now, we tackle the case when the set $\rho$ is simple cofinite.
The following prevents us from using the same approach that worked when $\rho$ was finite: Any unselected vertex that is satisfied when it has at least one more selected neighbor is also satisfied when it receives more than one additional selected neighbors.
This means that we simply \emph{cannot} realize the $\hwRelation{}$-relation via the types of gadgets we used when $\rho$ was finite.

Our way out of this dilemma is to reduce from \minParGenDomSetRel{}, for which \cref{thm:minimization_high_lvl} provided the lower bound that utilizes relations that are superset-closed.
We show that superset-closed relations can be modeled using only $\hwGeqOneRelation{}$-relations, and these \emph{can} be replaced by gadgets in a subsequent step (recall that a $\hwGeqOneRelation$-relation accepts a selection from its scope if and only if it is nonempty).

\thmReplaceSupersetByHwGeqOne*{}
\begin{proof}
    Let the input instance be $I = (G,k,\ell,\mathcal{R})$.
    If there is a relation $R \in \mathcal{R}$, such that $\acc{R} = \emptyset$, then output a constant-sized no-instance.
    This is justified by the fact that then it is simply impossible to fulfill the relations.

    Otherwise,
    for each $R \in \mathcal{R}$, let $\complAcc{R}$ be the set of forbidden assignments, i.e., $\complAcc{R} = 2^{\scp{R}} \setminus \acc{R}$.
    For each set $r \in \complAcc{R}$, we add a $\hwGeqOneRelation{}$-relation with scope $\scp{R} \setminus r$ to the graph.
    Observe that $\scp{R} \setminus r$ is never the empty set, as this would imply that $r = \scp{R}$, and thus $\acc{R} = \emptyset$, taking into account that all relations are superset-closed.
    Let $\mathcal{R}'$ be the resulting set of relations.
    The output instance is then $I' = (G,k,\ell,\mathcal{R}')$.

    It is clear that this transformation does not increase the pathwidth as the scope of any relation of $\mathcal{R}'$ is a subset of the scope of some relation of $\mathcal{R}$.
    In particular, the path decomposition provided with the input instance is also a valid path decomposition for the output instance, naturally of the same width.
    Also observe that because the arity is a constant, the size of each relation of $\mathcal{R}$ and $\mathcal{R}'$ is also constant.
    Therefore, the reduction runs in polynomial-time.

    \begin{claim}
        If $I$ is a yes-instance, then $I'$ is a yes-instance.
    \end{claim}
    \begin{claimproof}
        Let $S$ be a solution to the input instance.
        Then, we show that each relation of the output instance is fulfilled.
        Indeed, as $S$ fulfills all relations $R$ of $\mathcal{R}$, we have that for each $R \in \mathcal{R}$ the set $S \cap \scp{R} \in \acc{R}$.
        Then, let $\complAcc{R}$ be the set of forbidden assignments of $R$, and consider an arbitrary $r \in \complAcc{R}$.

        Towards a contradiction, assume that $S$ selects no vertex of $\scp{R} \setminus r$.
        Then, $S \cap \scp{R}$ would necessarily be a subset of $r$.
        But, since $S \cap \scp{R}$ fulfilled relation $R$, this implies that any superset of $S \cap \scp{R}$ is in $\acc{R}$.
        Then, we must have that $r \in \acc{R}$, which clearly contradicts that $r \in \complAcc{R}$.

        Hence, $S$ fulfills all relations of $\mathcal{R}'$.
        Moreover, it is easy to observe that $S$ violates at most $\ell$ vertices and has size at most $k$.
    \end{claimproof}

    \begin{claim}
        If $I'$ is a yes-instance, then $I$ is a yes-instance.
    \end{claim}
    \begin{claimproof}
        Assume that $S$ is a solution to the output instance.
        We argue that the same set is a solution to the input instance.
        Clearly, it violates at most $\ell$ vertices and selects at most $k$ vertices.

        All that remains is arguing that $S$ fulfills each relation of $\mathcal{R}$.
        Towards a contradiction, assume there is some relation $R \in \mathcal{R}$ which is violated by $S$.
        Then, $S \cap \scp{R} \in \complAcc{R}$.
        However, there is a relation in $\mathcal{R}'$ that ensures at least one vertex of $\scp{R} \setminus (S \cap \scp{R})$ is selected by $S$, which is clearly impossible.
    \end{claimproof}

    \begin{claim}
        If $I$ is a good instance, then $I'$ is a good instance.
    \end{claim}
    \begin{claimproof}
        We have already established in this proof that a set $S$ fulfills all relations of $\mathcal{R}$ if and only if it fulfills all relations of $\mathcal{R}'$.
        Hence, when instance $I$ is good, so is instance $I'$, as the set of subsets of $V(G)$ that fulfill all relations remains exactly the same.
    \end{claimproof}

    By combining all these claims, the lemma is proven.
\end{proof}

This means that we can assume that each relation in our instance is a $\hwGeqOneRelation{}$-relation, and all that remains is showing that these can be simulated with appropriate gadgets.
We first define the notion of a \emph{fragile realization}, which is the notion of realization gadget we require to do this.

\defFragileRealization*{}

We now show that these fragile realization gadgets actually exist.\footnote{The described gadget uses cliques of size $\min \sigma + 1$, which are frequently useful for these types of problems, and shares similarities with gadgets used in \cite{fockeTightComplexityBoundsLowerBound,greilhuberResidueDominationBoundedTreewidth2025}.}

\begin{lemma}
    \label{thm:fragile_realization_gadgets_exist}
    Let $\sigma$ and $\rho$ be nonempty sets where $\rho$ is simple cofinite with $0 \notin \rho$.
    Then, there is a  constant $\gamma$, such that for any $d \geq 1$, there is a pair $(G,U)$ with $|U| = d$ fragilely realizing the $\hwGeqOneRelation{}$-relation of arity $d$ with cost $\gamma$.
\end{lemma}
\begin{proof}
    We begin by adding the vertices $U = \{u_1,\dots,u_d\}$ to the graph, as well as $\min \rho - 1$ cliques $K^{1},\dots,K^{\min \rho - 1}$, each on $\min \sigma + 1$ vertices.
    We make $v$ adjacent to exactly one vertex of each of these cliques, and to each vertex of $U$.
    This concludes the description of the graph $G$.

    Set $\gamma = (\min \rho - 1) \cdot (\min \sigma + 1)$, and quickly observe that $U$ is indeed an independent set.

    First, we argue that for each clique, it must be the case that all vertices of the clique are selected if no vertex of them is violated.
    Consider an arbitrary clique $K^i$.
    Observe that the mere existence of the clique implies $\min \rho > 1$.
    If the clique has size exactly one, then, $\min \sigma = 0$, and the vertex of the clique must be selected to not be violated since it has degree one, but $0,1 \notin \rho$.
    Otherwise, there exists some vertex $w$ in the clique that is not adjacent to $v$.
    If this vertex is selected, then, $w$ requires that all other vertices of the clique are also selected.
    Otherwise, if $w$ is not selected, then $w$ requires at least two selected neighbors.
    But even in this case, there would be a selected neighbor of $w$ which is not $v$,
    forcing the whole clique to be selected.

    We have shown that, unless a vertex of $V(G) \setminus U$ is violated, any selection $S$ must have size at least $\gamma$, as all $\min \sigma + 1$ vertices of each clique must be selected by $S$.
    Now, consider that $S$ selects no vertex of $U$.
    Then, $v$ is violated if it is not selected, because $v$ would then have only $\min \rho - 1$ selected neighbors.
    So, if $S$ violates no vertex of $V(G) \setminus U$, and $S$ selects no vertex of $U$, then the set $S$ must also select vertex $v$, and thus $|S \setminus U| > \gamma$.

    Similarly, because $N(U) = \{v\}$, it must be the case that when $S$ selects a vertex of $N(U)$, either $|S \setminus U| > \gamma$, or $S$ violates a vertex of $V(G) \setminus U$.
    Finally, observe that if we select all vertices of the cliques, and at least one vertex of $U$, then each vertex of $V(G) \setminus U$ is satisfied, and we select exactly $\gamma$ vertices of $V(G) \setminus U$.
\end{proof}

With these gadgets, we are already ready to state the final reduction we require to handle the case when $\rho$ is simple cofinite.

\thmPartialRhoCofiniteReplaceRelations*{}
\begin{proof}
    Let $I = (G,k,\ell,\mathcal{R})$ be the input instance, and $\gamma$ be the cost constant from \cref{thm:fragile_realization_gadgets_exist}.
    Moreover, set $t = \ell + \ell \cdot \gamma + k + 1$, and $k' = k + \gamma \cdot |\mathcal{R}| \cdot t$.
    The number $t$ represents how many copies of the fragile realization gadget we require per relation in $\mathcal{R}$.

    For each $r \in \mathcal{R}$, consider the fragile realization gadget $(H,U)$ of \cref{thm:fragile_realization_gadgets_exist} which realizes the relation with scope size $|U| = |\scp{r}|$.
    The number of copies of $(H,U)$ we add to the graph is $t$, we denote them by $H_R^1,\dots,H_R^t$.
    We also identify $U$ of each copy with $\scp{r}$.
    Let the resulting graph be $G'$.
    Then, the output instance is $I' = (G',k',\ell)$.

    \begin{claim}
        If $I$ is a yes-instance, then $I'$ is a yes-instance.
    \end{claim}
    \begin{claimproof}
        This is the straightforward direction of the proof.
        Let $S$ be a solution to instance $I$.

        Then, for any $R \in \mathcal{R}$ and $i \in \range{t}$, we know that $S$ selects at least one vertex of $\scp{R}$.
        Hence, we can extend $S$ to a solution of $H_R^i$ which violates no vertex of $V(H_R^i) \setminus \scp{R}$, selects exactly $\gamma$ vertices of $V(H_R^i) \setminus \scp{R}$, selects exactly the vertices $S \cap \scp{R}$ of $\scp{R}$ and selects no vertex of $N_{H_R^i}(\scp{R})$.

        It is clear that the resulting selection $S'$ has size at most $k'$.
        Moreover, this selection $S'$ does not change the violation status of any vertex that is in $V(G)$, and, no vertex of $V(G') \setminus V(G)$ is violated.
        Hence, $S'$ violates at most $\ell$ vertices.
    \end{claimproof}

    \begin{claim}
        If $I'$ is a yes-instance, then $I$ is a yes-instance.
    \end{claim}
    \begin{claimproof}
        Let $S$ be a solution for the output instance.
        First, we argue that then, $S \cap \scp{r} \in \acc{r}$ for all $r \in \mathcal{R}$.
        This will in turn automatically mean that $S$ must violate at least $\ell$ vertices of $G$, and that $S$ must select at least $k$ vertices of $V(G)$, because the input instance was good.
        Then, the solution $S$ cannot afford to violate any vertices of $V(G') \setminus V(G)$, and cannot afford to select more than $\gamma$ vertices of each attached gadget.
        This will imply that the solution behaves like we want.

        Observe that at most $\ell$ copies of the fragile realization gadget can have a violated vertex outside the scope of the relation the gadget is for.
        Then, consider that in the whole graph, this means that we must select at least $t \cdot |\mathcal{R}| \cdot \gamma - (\ell \cdot \gamma)$ vertices just from $V(G') \setminus V(G)$.
        The difference between this value and $k'$ is only $(\ell \cdot \gamma) + k$.

        In particular, since $S$ has size at most $k'$, this means that there are at most $(\ell \cdot \gamma) + k$ gadgets $(H,U)$ (which are copies of the fragile realization gadget) in which $S$ can select more than $\gamma$ vertices of $V(H) \setminus U$.
        So, there are at most $\ell + (\ell \cdot \gamma) + k < t$ copies $(H,U)$ of the fragile realization gadget in which a violation occurs or more than $\gamma$ vertices of $V(H) \setminus U$ are selected.

        Now, fix an arbitrary $R \in \mathcal{R}$.
        This means there exists at least one $i \in \range{t}$ such that at most $\gamma$ vertices of $V(H_R^i) \setminus \scp{R}$ are selected, and no vertex of $V(H_R^i) \setminus \scp{R}$ is violated.
        But then, by the properties of \cref{def:fragile_realization}, we obtain that $S$ must select at least one vertex of $\scp{R}$.

        Now, consider the set $S \cap V(G)$, and observe that this set fulfills each relation of $\mathcal{R}$.
        By the fact that the input instance was good, this means that $S \cap V(G)$ has size at least $k$, and violates at least $\ell$ vertices such that adding further selected vertices to these $\ell$ vertices does not cause them to be no longer violated.

        Hence, we see that $S$ selects $k$ vertices of $V(G)$, and violates $\ell$ vertices of $V(G)$.
        But then it must be the case that actually not a single vertex of $V(G') \setminus V(G)$ is violated.
        Moreover, for any added copy $(H,U)$ of a realization gadget it must then be the case that $S$ selects at least $\gamma$ vertices of $V(H) \setminus U$.
        But, then by the chosen value of $k'$, $S$ must select exactly $\gamma$ vertices of $V(H) \setminus U$, and from this, we obtain that $S$ cannot select a vertex of $N_H(U)$, as any solution for $(H,U)$ that satisfies all vertices of $V(H) \setminus U$ and selects a vertex of $N_H(U)$ has size strictly larger than $\gamma$.

        This means that actually, $S \cap V(G)$ not only fulfills all relations of $\mathcal{R}$, but also, any vertex in $V(G)$ has no selected neighbor in $V(G') \setminus V(G)$.
        Thus, we know that $S \cap V(G)$ violates at most $\ell$ vertices of $V(G)$, and furthermore, $S \cap V(G)$ must have size at most $k$.
        So, the input instance $I$ is a yes-instance.
    \end{claimproof}

    \begin{claim}
        The reduction is pathwidth-preserving.
    \end{claim}
    \begin{claimproof}
        This follows almost immediately from the fact that each added fragile realization gadget from \cref{thm:fragile_realization_gadgets_exist} actually has constant size, because the size of these gadgets only depends on $\sigma, \rho$ as well as $d$, which are all constants.
        Then, it is easy to see that the reduction runs in polynomial-time.

        The path decomposition provided with the input graph contains a bag containing $\scp{R}$ for any $R \in \mathcal{R}$.
        We can simply copy this bag sufficiently often, such that for each fragile realization gadget of $R$, we add all vertices to one of the copied bag in a way that never increases the pathwidth by more than a constant amount.
    \end{claimproof}
\end{proof}

\subsection{Finishing the Lower Bound Proof}
\label{sec:replacing_relations_final_reductions}

Now, we must only put together all the pieces we have so carefully crafted throughout the last two sections.

\mtheoremLowerBoundPartial*
\begin{proof}
    Observe that $|\allStatesPartial| = \sigLargestPartial + \rhoLargestPartial + 2$.
    We utilize two different reduction chains depending on whether $\rho$ is simple cofinite or not.

    \begin{description}
        \item[$\rho$ is simple cofinite:] In this case, we utilize the intermediate lower bound of \cref{thm:minimization_high_lvl} which works for good instances of bounded arity.
              Then, \cref{thm:partial_replace_superset_by_hw_geq_1} is used to replace all relations with $\hwGeqOneRelation$-relations via a pathwidth-preserving reduction.
              Finally, \cref{thm:partial_rho_cofinite_replace_relations} provides a reduction from the resulting instance to \minParGenDomSet{}.
        \item[$\rho$ is finite:] In this case, we utilize the intermediate lower bound of \cref{thm:decision_high_lvl}.
              Then, \cref{thm:replacing_arbitrary_with_hw_1} provides a reduction to the problem in which all relations are $\hwRelation$-relations.
              Finally, \cref{thm:rho_finite_partial_realize_relations_reduction} provides a reduction to \minParGenDomSet{}.
    \end{description}

    All used reductions are pathwidth-preserving, and also the arity of each used instance of an intermediate problem is bounded by a constant.
    Hence, the lower bound from the intermediate problems transfer to \minParGenDomSet{}.
\end{proof}

\section{The Nonpartial Problem}
\label{sec:nonpartial}
Let us now proceed to the tight lower bound for the nonpartial problem.
The main goal of this section is proving \cref{thm:main_theorem_lower_bound_nonpartial}.
The used reduction chain is similar to the one used for the partial problem.

\subsection{The Intermediate Lower Bound for the Nonpartial Problem}

Two different intermediate problems are used depending on whether $\rho$ is finite or simple cofinite.
Recalling the definition of relations in \cref{def:relations}, these problems are formally defined next.

\problembox{\minGenDomSetRel{}}{Graph $G$, integer $k$, set of superset-closed relations $\mathcal{R}$}{Is there a set $S \subseteq V(G)$ such that $|S| \leq k$ and no vertices of $V(G)$ are violated by $S$ relative to $(\sigma,\rho)$, and for each $R \in \mathcal{R}$ we have $S \cap \scp{R} \in \acc{R}$?}

\problembox{\genDomSetRel{}}{Graph $G$, set of relations $\mathcal{R}$}{Is there a set $S \subseteq V(G)$ such that no vertices of $V(G)$ are violated by $S$ relative to $(\sigma,\rho)$, and for each $R \in \mathcal{R}$ we have $S \cap \scp{R} \in \acc{R}$?}

Observe that an instance of \minGenDomSetRel{} can be viewed as an instance of \minParGenDomSetRel{} with $\ell = 0$.
Similarly, an instance of \genDomSetRel{} is an instance of \parGenDomSetRel{} with $\ell = 0$.

For the nonpartial problem we utilize the fact that lower bounds for the problem \genDomSetRel{} already exist in the literature.
For us, these lower bounds are useful to handle the cases when $\rho$ is finite and $\sigma$ is simple cofinite.
The following lower bound was first proven under SETH by Focke et al. \cite[Lemmas 2.1 and 2.2]{fockeTightComplexityBoundsLowerBound}.
The same lower bound can also be achieved under the \pwseth{} by adapting the proofs accordingly \cite{schepperFasterHigherEasier}.

\begin{lemma}[{Follows from \cite[Main Lemmas 10.3.7 and 10.4.3]{schepperFasterHigherEasier}}]
    \label{thm:intermediate_lower_bound_sigma_cofinite_rho_finite}
    Let $\sigma$ be a simple cofinite set, and $\rho$ be a nonempty finite set.
    Then, unless the \pwseth{} is false, the problem \genDomSetRel{} cannot be solved in time $(\sigLargestNonPartial + \rhoLargestNonPartial + 2 - \varepsilon)^{\pw} \cdot |V(G)|^{\Oh(1)}$ for any $\varepsilon > 0$, even when a path decomposition of width $\pw$ is provided with the input.
\end{lemma}

Although the lower bound of \cref{thm:intermediate_lower_bound_sigma_cofinite_rho_finite} would also hold for the case when $\rho$ is simple cofinite, we cannot utilize it in that scenario.
The problem is that we cannot replace arbitrary relations by gadgets when $\rho$ is finite.
Here, we instead go into a direction that is similar to our lower bound for the partial problem.
That is, we utilize superset-closed relations for the intermediate lower bound, and later we show that those can actually be replaced by gadgets.

The overall strategy of this intermediate reduction is the same as for the partial variant in \cref{sec:intermediate_partial}, but we need to make some important changes.
That we no longer allow violations makes the reduction simpler on the one hand, since we no longer need to deal with violated vertices and no longer need to propagate overflow-states along the construction.
But, it makes the reduction more difficult on the other hand, since we now need to ensure that all vertices can be satisfied if the input instance is a yes-instance.
In the partial variant, we could control the violation status of all vertices which are not state-vertices by making them adjacent to a clique of constant size, whose vertices were forced to be selected.
This same approach no longer works since we are not allowed to violate any vertices.
Before continuing, we need to define the notion of states for the nonpartial problem.

\begin{definition}[State Sets for the Nonpartial Problem]
    \label{def:states_non_partial}
    Let $\sigma, \rho$ be nonempty finite or simple cofinite sets.
    We define the set of $\rho$-states for the nonpartial problem as $\rhoStatesNonPartial = \{\rho_0,\dots,\rho_{\rhoLargestNonPartial}\}$,
    and the set of $\sigma$-states for the nonpartial problem as $\sigStatesNonPartial = \{\sigma_0,\dots,\sigma_{\sigLargestNonPartial}\}$.
    We denote the set of all states for the nonpartial problem as $\allStatesNonPartial = \sigStatesNonPartial \cup \rhoStatesNonPartial$.
\end{definition}

Focke et al. \cite{fockeTightComplexityBoundsLowerBound} solve the problem of ensuring that vertices are not violated by building complex graph families, called $A$-\emph{managers} (where $A \subseteq \allStatesNonPartial$).
Intuitively, there is a manager for each integer $n$, which has a set of $n$ distinct portal vertices.
The manager also has some notion of a left side, and a right side.
The state of a portal vertex in a solution is determined by whether the portal vertex is selected and by the number of selected vertices it has in the left side of the manager.
The $A$-manager for integer $n$ has the property that for any string $\vec{s} \in A^n$ there is a $(\sigma,\rho)$-set of the graph in which the portals have this state vector as their state.
The portal vertices of managers almost form a cut of the manager, and hence they can be used as the state-vertices in the reduction while keeping the pathwidth small.
Ideally, we would directly reuse their manager constructions for $A = \allStatesNonPartial$.
Unfortunately, they give no explicit guarantees regarding the size of these solutions.
This is a problem for our use case, as our idea behind using superset-closed relations only works if we can tightly control how many vertices are selected in a solution.
Note that it may be the case that their construction, potentially with some small adjustments or restrictions, has the properties we need.
However, proving this fact seems to be significantly more tedious and lengthy compared to just giving our own constructions.
It should be noted that the challenges we face for our construction are similar to the ones Focke et al. \cite{fockeTightComplexityBoundsLowerBound} face, so also similarities between the constructions would not be surprising.
Since the graphs we create are simpler than those of Focke et al. in some sense, we call them \emph{simple managers}.

\emph{Simple managers}  are inspired by the concept of managers in \cite[Definition 4.7]{fockeTightComplexityBoundsLowerBound}.
However, our managers are simpler in the sense that they are objects of constant size, and not infinite families of graphs.
We want to use these simple managers in our lower bound construction such that each simple manager should represent one variable of the input instance in one bag.
For this purpose, we need to ensure that a simple manager has a small cut $U$ (which will later be our state vertices), and such that the vertices in this cut $U$ can have (almost) any possible state in a solution, where the state of a vertex is given by the number of selected vertices in the ``left'' side of the gadget and by whether the vertex is selected.
We formalize these properties together with some additional guarantees on the solution size in the next definition.

\begin{definition}[Simple Manager]
    \label{def:simple_manager}
    Let $\sigma, \rho$ be finite or simple cofinite sets, and $A \subseteq \allStatesNonPartial^g$ for some $g \geq 1$.
    A simple $A$-manager is a structure $(G,U = \{v_1,\dots,v_g\},L,\compl{L},\simpleManagerLeftSolSize,\simpleManagerRightSolSize,\simpleManagerPortSolSize)$, where $G$ is a graph, $U,L,\compl{L}$ are pairwise disjoint subset of $V(G)$, and $\simpleManagerLeftSolSize,\simpleManagerRightSolSize,\simpleManagerPortSolSize$ are integers.
    The structure has all the following properties:
    \begin{enumerate}
        \item The set $V(G) \setminus U$ can be partitioned into the set $L$, the left side of the manager and the set $\compl{L}$, the right side of the manager, such that no vertex of $L$ has a neighbor in $\compl{L}$.
        \item For each $v \in U$ we have $N(v) \subseteq L \cup \compl{L}$. In other words, the set $U$ is independent.
        \item For any $\vec{s} \in A$, there is a $(\sigma,\rho)$-set $S_{\vec{s}}$ of $G$.
              Fix an arbitrary $i \in \range{g}$, and let $\vec s [i] = \tau_c$.
              The set $S_{\vec s}$ has the following properties:
              \begin{itemize}
                  \item Vertex $v_i$ is selected by $S_{\vec{s}}$ if and only if $\tau$ is the symbol $\sigma$.
                  \item Vertex $v_i$ has exactly $c$ selected neighbors in $L$.
                  \item If $v_i$ is selected, then $v_i$ has exactly $\sigLargestNonPartial - c$ selected neighbors in $\compl{L}$.
                  \item If $v_i$ is not selected, then $v_i$ has exactly $\rhoLargestNonPartial - c$ selected neighbors in $\compl{L}$.
                  \item Exactly $\simpleManagerPortSolSize$ vertices of $U$, $\simpleManagerLeftSolSize$ vertices of $L$, and $\simpleManagerRightSolSize$ vertices of $\compl{L}$ are selected.
              \end{itemize} \end{enumerate}
\end{definition}

Let us remark that the managers of Focke et al. \cite{fockeTightComplexityBoundsLowerBound} are more flexible in the sense that they can build $\allStatesNonPartial$-managers\footnote{Unless $(\sigma,\rho)$ is $m$-structured for some $m \geq 2$.}.
It is clear from \cref{def:simple_manager} that we cannot build simple $A$-managers for arbitrary $A \subseteq \allStatesNonPartial^g$, as for example, condition $3$ implies that each vector of $A$ has the same number of $\sigma$-states.
However, our construction will still be sufficient to craft the lower bound.

Before going into the details of how we construct the simple managers, we introduce a gadget of Focke et al. \cite[Lemma 5.1]{fockeTightComplexityBoundsLowerBound}.
\begin{lemma}
    \label{thm:sigma_s_rho_s_provider_gadget}
    Let $\sigma$ and $\rho$ be arbitrary nonempty subsets of the natural numbers where $\rho \not = \{0\}$.
    Then, for any $s \in \sigma$ and $r \in \rho \setminus \{0\}$, there is a graph $G$ and a vertex $u \in V(G)$, called the portal, such that there is a $(\sigma,\rho)$-set of $G$ of size $(s +1) \cdot r$ that selects $u$ and exactly $s$ neighbors of $u$, and a $(\sigma,\rho)$-set of $G$ of size $(s+1) \cdot r$ that does not select $u$ and that selects exactly $r$ neighbors of $u$.
\end{lemma}
\begin{proof}
    A quick inspection of the proof of \cite[Lemma 5.1]{fockeTightComplexityBoundsLowerBound} confirms these properties for the gadget created there.
\end{proof}

Since we use the gadget of \cref{thm:sigma_s_rho_s_provider_gadget} several times in our construction, we will simply refer to it as the \emph{useful provider}.
Let us note that
because $\allStatesNonPartial \subseteq \allStatesPartial$, the weight definitions in \cref{def:partial_weights} also work for the nonpartial case in the natural way.
Similarly, the definition of state complements in \cref{def:partial_state_complements} also applies for the nonpartial state set, notice here that no state of $\allStatesNonPartial$ is an overflow state.
We begin by providing a simple manager for the case that $\sigma \not = \{0\}$, and that $\rho$ is a simple cofinite set.

\begin{figure}
    \centering
    \includegraphics[width = \linewidth]{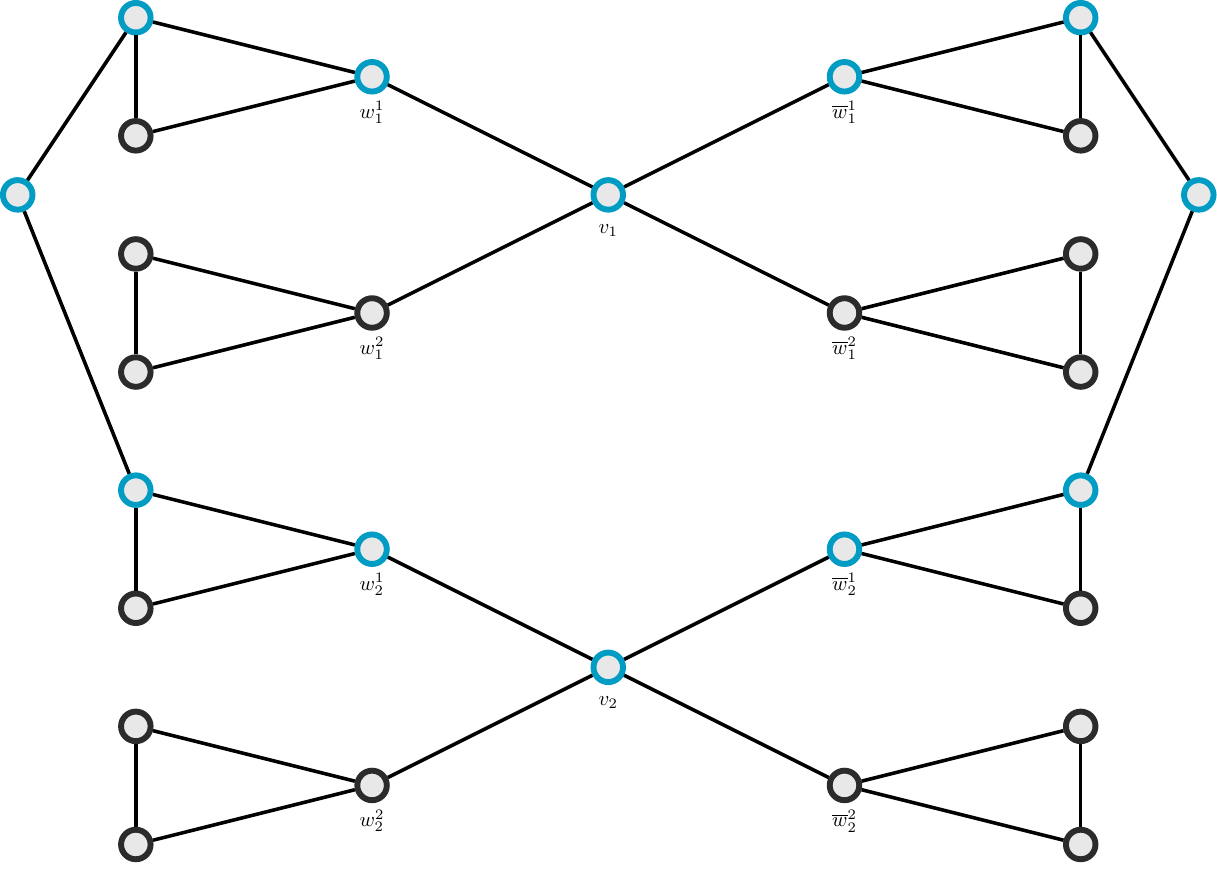}
    \caption{An example of the simple manager from \cref{thm:simple_manager_sigma_not_zero} for the case $A = \{\sigma_1\sigma_1\}$, $\sigma = \{2\}$, $\parityConstant = 2$, $\rho = \mathbb{Z}_{\geq 1}$ and a selection satisfying all vertices. Only selected vertices of $Q$ and $\compl{Q}$ are drawn, and copies of the useful provider are omitted. The selections within the omitted useful providers guarantee that all unselected vertices are satisfied.}
    \label{fig:simple_manger_sigma_no_zero}
\end{figure}

\begin{lemma}
    \label{thm:simple_manager_sigma_not_zero}
    Let $\sigma$ be a finite or simple cofinite set such that $\sigma \not = \{0\}$, and $\rho$ a simple cofinite set.
    Furthermore, let $A \subseteq \allStatesNonPartial^g$ (where $g \geq 1$) be a set of state-vectors, such that every $\vec s \in A$ has the same weight $\weightConstant$, the same number of $\sigma$-states $\selStateVertexPerGadget$, and the same $\sigma$-weight $\selWeight$, and there is some $\parityConstant \in \sigma$ with $\parityConstant > 0$ such that $\selStateVertexPerGadget \equiv_\parityConstant 0$ and $\selWeight \equiv_\parityConstant 0$.
    Then, there is a simple $A$-manager.
\end{lemma}
\begin{proof}
    Since each vector of $A$ has the same weight $\weightConstant$ and the same $\sigma$-weight $\selWeight$, each vector of $A$ also has the same $\rho$-weight $\unselWeight$.
    Each vector of $A$ has $\selStateVertexPerGadget$ $\sigma$-states, so also $\unselStateVertexPerGadget = g - \selStateVertexPerGadget$ $\rho$-states.

    We introduce the vertices $U = \{v_1,\dots,v_g\}$ to the graph.
    Then, for each $i \in \range{g}$, we introduce $\allLargestNonPartial$ vertices $w_i^1,\dots,w_i^{\allLargestNonPartial}$, and $\allLargestNonPartial$ vertices $\compl{w}_i^1,\dots,\compl{w}_i^{\allLargestNonPartial}$.
    For any $i \in \range{g}$ and $j \in \range{\allLargestNonPartial}$, we make vertices $w_i^j$ and $\compl{w}_i^j$ adjacent to vertex $v_i$.
    Furthermore, for each $i \in \range{g}$ and each $j \in \range{\allLargestNonPartial}$ we introduce a clique $K_i^j$ on $\parityConstant$ vertices, similarly introduce a clique $\compl{K}_i^j$ on $\parityConstant$ vertices.
    Each vertex of $K_i^j$ is made adjacent to vertex $w_i^j$, and each vertex of $\compl{K}_i^j$ is made adjacent to vertex $\compl{w}_i^j$.

    Now, let $Z = \bigcup_{i \in \range{g}, j \in \range{\allLargestNonPartial}} V\left(K_i^j\right)$ and $\compl{Z} = \bigcup_{i \in \range{g}, j \in \range{\allLargestNonPartial}} V\left(\compl{K}_i^j\right)$.
    Then, for each subset $z$ of $Z$ of size exactly $\parityConstant$, introduce an additional vertex $q_z$.
    Make vertex $q_z$ adjacent to each vertex of $z$.
    Denote the set of all of these vertices $q_z$ introduced so far as $Q$.
    Similarly, for each subset $z$ of $\compl{Z}$ of size exactly $\parityConstant$, introduce an additional vertex $q_z$, and make this vertex adjacent to each vertex of $z$.
    Denote the set of all of these vertices $q_z$ we introduced for $\compl{Z}$ as $\compl{Q}$.

    Let $G'$ be the graph created so far.
    Proceed by adding $\min \rho$ copies of the useful provider from \cref{thm:sigma_s_rho_s_provider_gadget} (for the constants $\sigLargestNonPartial \geq 0$ and $\rhoLargestNonPartial + 1 > 0$) to the graph for every vertex $v \in V(G') \setminus U$.
    Make the portal vertex of each of these $\min \rho$ gadgets adjacent to vertex $v$.
    For convenience, define $C = (\sigLargestNonPartial + 1) \cdot (\rhoLargestNonPartial + 1)$, which corresponds to the size of the solutions which are guaranteed to exist in each useful provider.
    Denote the vertices of all of these $\min \rho$ copies of the useful provider attached to $v$ as the vertex set $P_v$.
    This concludes the description of $G$.
    Consult \cref{fig:simple_manger_sigma_no_zero} for an illustration.

    Define $L' = Q \cup \bigcup_{i \in \range{g}, j \in \range{\allLargestNonPartial}} \big( \{w_i^j\} \cup V(K_i^j) \big)$, and $L = L' \cup \bigcup_{v \in L'} P_v$.
    Similarly, define $\compl{L'} = \compl{Q} \cup \bigcup_{i \in \range{g}, j \in \range{\allLargestNonPartial}} \big( \{\compl{w}_i^j\} \cup V(\compl{K}_i^j) \big)$, and $\compl{L} = \compl{L'} \cup \bigcup_{v \in \compl{L'}} P_v$.
    We utilize $L$ as the left side of the simple manager, and $\compl{L}$ as the right side of the simple manager.
    It is not difficult to see that $L,\compl{L}$ is a partition of $V(G) \setminus U$.
    Hence, $V(G) \setminus U$ fulfills condition 1 of the simple manager definition.
    Moreover, each vertex of $U$ only has neighbors in $L' \cup \compl{L'}$, so $U$ fulfills condition 2 of the definition of simple managers.

    So, all that remains is showing that an appropriate set $S_{\vec s}$ exists for any $\vec s \in A$.
    For this purpose, define
    \[M_{L'}= \unselWeight{} \cdot (\parityConstant + 1) + \selWeight{} \cdot \parityConstant + \frac{\selWeight \cdot (\parityConstant-1)}{\parityConstant},\] which turns out to be the number of selected vertices of the set $L'$ in our chosen solutions.
    First, observe that since $\selWeight \equiv_\parityConstant 0$, the value $M_{L'}$ is an integer.
    The value $M_{L'}$ comes together as follows.
    For each selected vertex $w_i^j$, we will select almost all vertices of the clique of size $\parityConstant$ that is attached to $w_i^j$.
    When $v_i$ is not selected, we will select the whole clique of size $\parityConstant$, when $v_i$ is selected, we will select the whole clique except for one vertex.
    Hence, in exactly $\unselWeight{}$ cliques we will select $\parityConstant$ vertices, and in $\selWeight{}$ we will select $\parityConstant - 1$ vertices.
    Furthermore, we will select $\unselWeight + \selWeight$ vertices of the form $w_i^j$.
    If the whole clique of size $\parityConstant$ is selected, and the neighbor $w_i^j$ of each clique vertex is selected, this ensures that each clique vertex and vertex $w_i^j$ are satisfied since those vertices then have $\parityConstant \in \sigma$ selected neighbors.
    However, if only $\parityConstant - 1$ vertices of the clique are selected, this is not the case.
    While the vertex $w_i^j$ neighboring each clique vertex then has $\parityConstant$ selected neighbors ($\parityConstant - 1$ from the clique and selected neighbor $v_i$), each selected clique vertex has only $\parityConstant - 1$ selected neighbors: $\parityConstant - 2$ from within the clique, and the selected neighbor $w_i^j$.
    Since we might not have $\parityConstant - 1 \in \sigma$, each of the $\parityConstant - 1$ selected vertices of the clique requires an additional selected vertex.
    So, we will select $\frac{\selWeight \cdot (\parityConstant-1)}{\parityConstant}$ further vertices $q_z$ to satisfy all selected clique vertices.

    Now, let us define $\compl{\selWeight} = \selStateVertexPerGadget \cdot \sigLargestNonPartial - \selWeight$ which corresponds to the $\sigma$-weight of the complement of vectors in $A$.
    Note that, since $\selStateVertexPerGadget$ and $\selWeight$ are divisible by $\parityConstant$, so is $\compl{\selWeight}$.
    The $\rho$-weight the complements of vectors in $A$ is $\compl{\unselWeight} = \unselStateVertexPerGadget \cdot \rhoLargestNonPartial - \unselWeight$.
    Next, define \[M_{\compl{L'}} = \compl{\unselWeight} \cdot (\parityConstant+1) + \compl{\selWeight} \cdot \parityConstant + \frac{\compl{\selWeight} \cdot (\parityConstant-1)}{\parityConstant}.\]
    The intuition behind this value is the same as for $M_{L'}$, only that it is for the complement of states of $A$.

    Next, fix an arbitrary $\vec s \in A$.
    We now build solution $S_{\vec s}$ of condition 3 of the simple manager definition.
    We select vertex $v_i$ if and only if $\vec{s}[i]$ is a $\sigma$-state.
    Moreover, if $\vec{s}[i] = \tau_c$, we select exactly $c$ neighbors of $v_i$ in $L$, which are vertices from the set $w_i^1,\dots,w_i^{\allLargestNonPartial}$, and $s_\tau - c$ vertices from the set $\compl{w}_i^1,\dots,\compl{w}_i^{\allLargestNonPartial}$.
    Then, if vertex $w_i^j$ is selected and $v_i$ is not selected, we select all vertices of the clique $K_i^j$.
    On the other hand, if $w_i^j$ is selected and $v_i$ is also selected, we select all except for one vertex of $K_i^j$.
    Similarly, if $\compl{w}_i^j$ is selected and $v_i$ is not selected, we select all vertices of $\compl{K}_i^j$.
    If both $\compl{w}_i^j$ and $v_i$ are selected, we select all vertices of $\compl{K}_i^j$ apart from one.

    Call the selection so far $S_{\vec s}''$.
    Observe that each vertex of $U$ is satisfied by $S_{\vec s}''$.
    When $w_i^j$ is selected and $v_i$ is not selected, also the vertex $w_i^j$ and each vertex of $K_i^j$ are satisfied, because these latter vertices then have exactly $\parityConstant \in \sigma$ selected neighbors each.
    On the other hand, when both $w_i^j$ and $v_i$ are selected, this is not the case.
    Then, vertex $w_i^j$ is satisfied, but each selected vertex of $K_i^j$ only has $\parityConstant - 1$ selected neighbors.
    Let $L_{\uptextsf{v}}$ denote all the selected vertices which are part of a clique $K_i^j$ that only have $\parityConstant - 1$ selected neighbors in set $S_{\vec s}''$.
    Observe that we have exactly $\parityConstant - 1$ of these vertices for every selected $w_i^j$ with selected $v_i$.
    Given that the $\sigma$-weight of $\vec s$ is $\selWeight$, this means that $|L_{\uptextsf{v}}| = (\parityConstant - 1) \cdot \selWeight$.
    In particular, $|L_{\uptextsf{v}}|$ is divisible by $\parityConstant$ because $\selWeight$ is divisible by $\parityConstant$ by assumption.
    Similarly, let $\compl{L_{\uptextsf{v}}}$ be the selected vertices which are part of a clique $\compl{K}_i^j$ and only have $\parityConstant - 1$ selected neighbors in the set $S_{\vec s}''$.
    Observe that we again have $\parityConstant - 1$ of these vertices for every vertex $\compl{w}_i^j$ with selected $v_i$.
    In particular, this means that $\left|\compl{L_{\uptextsf{v}}}\right| = (\parityConstant - 1) \cdot \compl{\selWeight}$.
    We have already established that $\compl{\selWeight}$ is divisible by $\parityConstant$.
    So, also $\left|\compl{L_{\uptextsf{v}}}\right|$ is divisible by $\parityConstant$.

    We now extend $S_{\vec s}''$ to a selection $S_{\vec s}'$.
    Partition $L_{\uptextsf{v}}$ into $\frac{|L_{\uptextsf{v}}|}{\parityConstant}$ blocks of size exactly $\parityConstant$.\footnote{The empty set can be partitioned into zero blocks of arbitrary size.}
    Then, iterate through all blocks $z$ of this partition, and add the vertex $q_z$ to the set $S_{\vec s}'$.
    Similarly, partition $\compl{L_{\uptextsf{v}}}$ into blocks of size exactly $\parityConstant$, iterate through each block $z$ of this partition, and add the vertex $q_z$ to the set $S_{\vec s}'$.

    Now, observe that in $S_{\vec s}'$, each vertex of $L_{\uptextsf{v}}$ and $\compl{L_{\uptextsf{v}}}$ received exactly one additional neighbor compared to $S_{\vec s}''$, so each of these vertices is now satisfied.
    Moreover, each selected vertex $q_z$ (for any set $z$) is satisfied because all of its $\parityConstant \in \sigma$ neighbors are selected, and it is selected itself.

    In particular, every selected vertex of $L' \cup \compl{L'}$ is satisfied, and every vertex of $U$ is satisfied.
    So, finally, we want to extend solution $S_{\vec s}'$ to solution $S_{\vec s}$ which also makes the remaining vertices satisfied, that is, the unselected vertices and vertices of the useful providers.
    For this purpose, consider an arbitrary unselected vertex $v$ of $L' \cup \compl{L'}$.
    Then, this vertex $v$ is adjacent to $\min \rho$ vertices of the set $P_v$.
    Each useful provider that makes up $P_v$ has some $(\sigma,\rho)$-set of size $C$ selecting the neighbor of $v$ in the gadget.
    We choose this solution in the gadget.
    Then, $v$ has at least $\min \rho$ selected neighbors, and is thus satisfied.

    Finally, we need to show that, even when vertex $v \in L' \cup \compl{L'}$ is selected, we can satisfy all vertices of $P_v$.
    This is also not difficult.
    In each gadget that makes up $P_v$, we choose the solution of size $C$ that does not select the neighbor of $v$.
    Considering that $\rho$ is simple cofinite this selection satisfies all vertices of $P_v$.
    Regarding the number of selected vertices, it is not difficult to confirm that we select exactly $\simpleManagerPortSolSize = \selStateVertexPerGadget$ vertices of $U$, $M_{L'}$ vertices of $L'$, and $M_{\compl{L'}}$ vertices of $\compl{L'}$.
    Furthermore, we select exactly $C$ vertices of each useful provider, and there are $\min \rho$ useful providers attached to each vertex of $L'$ and $\compl{L'}$.
    Hence, we select $\simpleManagerLeftSolSize = M_{L'} + |L'| \cdot \min \rho \cdot C$ vertices of $L$, and $\simpleManagerRightSolSize = M_{\compl{L'}} + |\compl{L'}| \cdot \min \rho \cdot C$ vertices of $\compl{L}$.
\end{proof}

Since the simple manager above does not work when $\sigma = \{0\}$ (observe that in this case we certainly cannot find a suitable value for $\parityConstant$), we need another construction to handle this case, which we present next.

\begin{lemma}
    \label{thm:simple_manager_sigma_zero}
    Let $\sigma = \{0\}$ and $\rho$ be a simple cofinite set.
    Moreover, let $A \subseteq \allStatesNonPartial^g$ (where $g \geq 1$), such that each $s \in A$ has the same number of selected vertices $\selStateVertexPerGadget$ and the same weight $\weightConstant$.
    Then, there is a simple $A$-manager.
\end{lemma}
\begin{proof}
    This construction is similar to the one for the case when $\sigma \not = \{0\}$.
    We begin by introducing the vertices $U = \{v_1,\dots,v_g\}$ to the graph.
    Then, for each $i \in \range{g}$, introduce $\allLargestNonPartial$ vertices $w_i^1,\dots,w_i^{\allLargestNonPartial}$, and $\allLargestNonPartial$ vertices $\compl{w}_i^1,\dots,\compl{w}_i^{\allLargestNonPartial}$.

    Proceed by adding $\min \rho$ copies of the useful provider from \cref{thm:sigma_s_rho_s_provider_gadget} (for the constants $0 \in \sigma$ and $\min \rho + 1$) to the graph for every vertex $v$ that is not in $U$, and make the portal of each gadget adjacent to vertex $v$.
    Denote the vertices of all of these $\min \rho$ copies attached to $v$ as the set $P_v$.
    This concludes the description of $G$.

    Define $L' = \bigcup_{i \in \range{g}, j \in \range{\allLargestNonPartial}} \{w_i^j\}$, and $L = L' \cup \bigcup_{v \in L} P_v$.
    Similarly, define $\compl{L'} = \bigcup_{i \in \range{g}, j \in \range{\allLargestNonPartial}} \{\compl{w}_i^j\}$, and $\compl{L} = \compl{L'} \cup \bigcup_{v \in \compl{L'}} P_v$.
    It is not difficult to see that $L,\compl{L}$ is a partition of $V(G) \setminus U$, we will use $L$ as the left side, and $\compl{L}$ as the right side of the simple manager.
    Furthermore, it is clear that each vertex of $U$ only has neighbors in $L' \cup \compl{L'}$.

    So, all that remains is showing that an appropriate set $S_{\vec s}$ exists for any ${\vec s} \in A$.
    Fix an arbitrary ${\vec s} \in A$.
    Then, we select vertex $v_i$ if and only if ${\vec s}[i]$ is a $\sigma$-state.
    Moreover, if ${\vec s}[i] = \rho_c$, we select exactly $c$ neighbors of $v_i$ in $L$, which are vertices from the set $w_i^1,\dots,w_i^{\allLargestNonPartial}$, and $\rhoLargestNonPartial - c$ vertices from the set $\compl{w}_i^1,\dots,\compl{w}_i^{\allLargestNonPartial}$.

    Recalling that $\sigma = \{0\}$ observe that under selection $S_{\vec s}$ all vertices of $U$ and all selected vertices of $L' \cup \compl{L'}$ are satisfied.

    To satisfy the remaining vertices, for any unselected $v$ which is not part of $U$ and each useful provider attached to $v$ we select $C = \min \rho + 1$ vertices of the useful provider, such that $v$ has $\min \rho$ selected neighbors in the set $P_v$.
    For each selected vertex $v$ that is not part of $U$, we instead select $C$ vertices of each useful provider that is attached to $v$, such that $v$ has no neighbors in $P_v$.

    Regarding the size of this solution, we clearly select exactly $\simpleManagerPortSolSize = \selStateVertexPerGadget$ of $U$.
    Furthermore, we select $\weightConstant$ vertices of $L'$.
    Then, in each of the $\min \rho$ useful provider attached to each vertex of $L'$ we select $C$ vertices.
    Hence, we select $\simpleManagerLeftSolSize = \weightConstant + |L'| \cdot \min \rho \cdot C$ vertices of $L$.
    Similarly, we select $M_{\compl{L'}} = \selStateVertexPerGadget \cdot \sigLargestNonPartial + (g - \selStateVertexPerGadget) \cdot \rhoLargestNonPartial - \weightConstant$ vertices of $\compl{L'}$.
    Thus, we select $\simpleManagerRightSolSize = M_{\compl{L'}} + |\compl{L'}| \cdot \min \rho \cdot C$ vertices of $\compl{L}$.
\end{proof}

As a last step before giving the reduction, we extend the notions of good instances and path decomposition to the nonpartial problem.
Recall the notions for the partial problem defined in \cref{def:partial_good_instance,def:partial_intermediate_path_decompositions}.
An instance $(G,k,\mathcal{R})$ of \minGenDomSetRel{} is good if and only if the instance $(G,k,0,\mathcal{R})$ of \minParGenDomSetRel{} is good.
This is exactly the case if each subset of $V(G)$ that fulfills all relations of $\mathcal{R}$ has size at least $k$.
A path decomposition of instance $(G,k,\mathcal{R})$ of \minGenDomSetRel{} is the same as a path decomposition of the instance $(G,k,0,\mathcal{R})$ of \minParGenDomSetRel{}.
That is, it is a path decomposition of graph $G$ such that additionally for each $R \in \mathcal{R}$ there is a bag containing $\scp{R}$.
Now, we are finally ready to proof our intermediate lower bound.

\begin{lemma}
    \label{thm:minimization_non_partial_high_lvl}
    Let $\sigma$ be a nonempty finite or simple cofinite set and $\rho$ be a simple cofinite set.
    Then, for any $\varepsilon > 0$ there exists a constant $d$ such that no algorithm can solve the problem \minGenDomSetRel{} on good instances $I$ with arity at most $d$ in time $(|\allStatesNonPartial| - \varepsilon)^{\tw} \cdot |I|^{\Oh(1)}$, even when a path decomposition of width $\pw$ is provided with the input, unless the \pwseth{} is false.
\end{lemma}
\begin{proof}
    We follow the same general proof template as for the partial variant, in principle this construction is not majorly different from the one for the partial problem.
    The main difference is that we need to utilize the simple managers to ensure that all vertices can be satisfied.

    \subparagraph*{Choosing the alphabet size.}
    We first need to figure out which alphabet size we need to use for the 2-\problem{CSP} problem.
    Concretely, we need to ensure that the conditions of \cref{thm:simple_manager_sigma_not_zero} or \cref{thm:simple_manager_sigma_zero} are fulfilled, as otherwise we cannot use those simple managers.

    Fix an arbitrary $\varepsilon > 0$, we will have a reduction for this fixed $\varepsilon$ that will work against an improvement of $\varepsilon$ in the base of the running time.
    Let $P_n$ be the partition of $\allStatesNonPartial^n$ such that each string of each block has the same number of $\sigma$-states, the same $\sigma$-weight, and the same weight.
    Since each string of $\allStatesPartial^n$ has a number of $\sigma$-states that is between $0$ and $n$, a $\sigma$-weight between $0$ and $n \cdot \sigLargestNonPartial$, and a weight that is between $0$ and $n \cdot \allLargestNonPartial$, we find that $|P_n| \leq (n + 1) \cdot (n \cdot \sigLargestNonPartial + 1) \cdot (n \cdot \allLargestNonPartial + 1)$.
    Because $\sigLargestNonPartial$ and $\allLargestNonPartial$ are constants this value is polynomial in $n$.
    Furthermore, set $X = \mathbb{Z}_{\geq 3}$.
    We now use \cref{thm:choosing_alphabet_size} for the value $B = |\allStatesNonPartial|$, function $f(x) = |P_x|$, value $\varepsilon$ and infinite set $X$.
    The lemma yields that there are positive integers $q'$, and $g' \in X$, and a real $\varepsilon'' > 0$ such that $|\allStatesNonPartial|^{q'} \leq \frac{|\allStatesNonPartial|^{g'}}{|P_{g'}|}$ and $(|\allStatesNonPartial| - \varepsilon)^{g'} \leq (|\allStatesNonPartial|^{q'} - \varepsilon'')$.
    We have $g' \geq 3$ because $g' \in X$.
    Note that the statement of \cref{thm:choosing_alphabet_size} is purely existential, but this is not a problem since we must only prove the existence of a reduction algorithm for our fixed $\varepsilon$, and these values can be hard-coded into the reduction.
    Now, let $A'$ be the largest block of $P_{g'}$, then $|\allStatesNonPartial|^{q'} \leq |A'|$.
    Let the $\sigma$-weight of the strings of $A'$ be $\selWeight'$, the weight of strings of $A'$ be $\weightConstant'$, and the number of $\sigma$-states of strings of $A'$ be $\selStateVertexPerGadget'$.
    Note that we have $|A'| \geq 3$ as for example the set containing state vectors that have the state $\sigma_0$ a total of $g' - 1$ times and the state $\rho_0$ one time is a subset of a block of $P_{g'}$ that already has size $g' \geq 3$.

    If $\sigma = \{0\}$, we can use $A'$ as our set for \cref{thm:simple_manager_sigma_zero}.
    So, we set $q = q', \varepsilon' = \varepsilon'', g = g'$ and $A = A'$, moreover we set $\selStateVertexPerGadget = \selStateVertexPerGadget'$, $\selWeight = \selWeight'$, and $\weightConstant = \weightConstant'$.

    Otherwise, if $\sigma \not = \{0\}$, then $A'$ is not yet sufficient for our purposes.
    We need a set of state-vectors, such that the number of $\sigma$-states and the $\sigma$-weight are both divisible by $\parityConstant$ for some $\parityConstant \in \sigma \setminus \{0\}$.
    For this purpose fix an arbitrary $\parityConstant \in \sigma \setminus \{0\}$.

    To achieve our goal we set $A = (A')^{\parityConstant}$.
    Although technically each element of $A$ is a vector of dimension $\parityConstant$ of vectors from $\allStatesNonPartial^{g'}$, we treat elements of $A$ as vectors of dimension $\parityConstant \cdot g'$ with elements from $\allStatesNonPartial$.
    Observe that since states of $A'$ have the same number of $\sigma$-states and the same $\sigma$-weight, we get that the number of $\sigma$-states and the $\sigma$-weight of each vector of $A$ are divisible by $\parityConstant$.
    The number of $\sigma$-states of each vector of $A$ is $\selStateVertexPerGadget = \selStateVertexPerGadget' \cdot \parityConstant$, the $\sigma$-weight is $\selWeight = \selWeight' \cdot \parityConstant$, and the total weight is $\weightConstant = \weightConstant' \cdot \parityConstant$.
    Also set $q = q' \cdot \parityConstant$ and $\blocksPerVertexGadget = \blocksPerVertexGadget' \cdot \parityConstant$.

    Then, from $|\allStatesNonPartial|^{q'} \leq |A'|$ we obtain $|\allStatesNonPartial|^q = |\allStatesNonPartial|^{q' \cdot \parityConstant} \leq |A'|^{\parityConstant} = |A|$.
    Furthermore, from $(|\allStatesNonPartial| - \varepsilon)^{g'} \leq (|\allStatesNonPartial|^{q'} - \varepsilon'')$, we obtain that $(|\allStatesNonPartial| - \varepsilon)^{\blocksPerVertexGadget} = (|\allStatesNonPartial| - \varepsilon)^{g' \cdot \parityConstant} \leq (|\allStatesNonPartial|^{q'} - \varepsilon'')^\parityConstant$.
    We have  $(|\allStatesNonPartial|^{q'} - \varepsilon'')^\parityConstant < |\allStatesNonPartial|^{q' \cdot \parityConstant} = |\allStatesNonPartial|^q$, so we also have $(|\allStatesNonPartial| - \varepsilon)^{g} \leq |\allStatesNonPartial|^{q}- \varepsilon' \leq |A| - \varepsilon'$ for some $\varepsilon' > 0$.
    Finally, from $|A'| \geq 3$ and $\parityConstant \geq 1$ we also obtain $|A| \geq 3$.
    We need this condition to apply \cref{thm:lampis}.

    \subparagraph*{The construction.}

    We reduce from $2$-\problem{CSP} with an alphabet of size $|A|$.
    We provide the proof as if the alphabet was exactly the set $A$ for convenience.
    Let the input instance $I$ be on the variables $\variable{1},\dots,\variable{\variableCount}$ and the input constraints be $\constraint{1},\dots,\constraint{\constraintCount}$.
    Moreover, let the path decomposition of the primal-graph of the instance that is provided with the instance be $\bag{1},\dots,\bag{\bagCount}$.
    We also assume we have some injective function $\constraintBagMapping{}$, which maps each constraint to the index of the bag containing all variables of the constraint.
    Such a function can be assumed to exist, as we could otherwise copy bags of the path decomposition until it is easy to compute.

    We build the graph of the output instance as follows.
    \begin{itemize}
        \item For each $i \in \range{\bagCount}$ and each $\variable{j} \in \bag{i}$ and each $p \in \range{\blocksPerVertexGadget}$ we create \emph{state} vertex $\statevertex{i}{j}{p}$.
        \item For each $i \in \range{\bagCount}$ and each $\variable{j} \in \bag{i}$ we add the simple $A$-manager from \cref{thm:simple_manager_sigma_not_zero} or \cref{thm:simple_manager_sigma_zero} (depending on whether $\sigma = \{0\}$) to the graph.
              We identify the $g$ vertices $U$ of the manager with the vertices $\statevertex{i}{j}{1},\dots,\statevertex{i}{j}{\blocksPerVertexGadget}$.
              We denote the left side of the manager as $L^i_j$, and the right side as $\compl{L}^i_j$.
    \end{itemize}

    This concludes the description of the output graph $G$.
    For convenience, we define some additional sets.
    For each bag $\bag{i}$ and $\variable{j} \in \bag{i}$, we define $U^i_j = \bigcup_{p \in \range{g}} \{\statevertex{i}{j}{p}\}$.
    We also set $\bag{0} = \bag{\bagCount + 1} = \emptyset$.

    Let $S$ be a subset of the vertices of the graph.
    We define the state of arbitrary state vertex $\statevertex{i}{j}{p}$ (relative to this selection $S$).
    Let $\tau$ be the symbol $\sigma$ if $\statevertex{i}{j}{p}$ is selected by $S$, and the symbol $\rho$ otherwise.
    Define $h = \min(|S \cap N(\statevertex{i}{j}{p}) \cap L^i_j|, \tauLargestNonPartial)$.
    We define $\state{\statevertex{i}{j}{p}} = \tau_h$.
    Let $v_1,v_2,\dots,v_\ell$ be a list of $\ell$ state vertices.
    Then, we define $\state{v_1,v_2,\dots,v_\ell} = \state{v_1}\state{v_2}\dots\state{v_\ell}$.

    We similarly define the complementary states of state vertices, $\rightState{\statevertex{i}{j}{p}}$, the only difference in the definition is that we utilize $\compl{L}^i_j$ instead of $L^i_j$.
    We also extend the definition of complementary states of state vertices to lists of complementary states of state vertices.
    Next, we elaborate on the relations we use.

    \subparagraph*{The used relations.}

    Similarly to the proof for the partial version, we first describe a set of relations $\mathcal{R}_0$ which is not superset-closed.
    Later we define the relations  $\mathcal{R}$ of the output instance to be the superset-closure of $\mathcal{R}_0$.
    Recall the values $\simpleManagerLeftSolSize$ and $\simpleManagerRightSolSize$ from the definition of simple managers in \cref{def:simple_manager}.

    For each $i \in \range{\bagCount}$ and each $\variable{j} \in \bag{i}\setminus{\bag{i-1}}$, we create an \emph{initial} relation $\initialRelation{i}{j}$.
    The relation scope of this relation is $U^i_j \cup L^i_j$.
    The relation accepts $S \cap \scp{\initialRelation{i}{j}}$ if and only if
    \begin{enumerate}
        \item $\state{\statevertex{i}{j}{1},\statevertex{i}{j}{2},\dots,\statevertex{i}{j}{\blocksPerVertexGadget}} \in A$, and
        \item for all $p \in \range{\blocksPerVertexGadget}$ we have that exactly $\weight{\state{\statevertex{i}{j}{p}}}$ vertices of $N(\statevertex{i}{j}{p}) \cap L^i_j$ are selected, and
        \item
              at least $\simpleManagerLeftSolSize$ vertices of $L^i_j$ are selected.
    \end{enumerate}

    For each $i \in \range{\bagCount}$ and each $\variable{j} \in \bag{i} \cap \bag{i-1}$ we create a \emph{consistency} relation $\consistencyRelation{i}{j}$.
    The relation scope of this relation is $\scp{\consistencyRelation{i}{j}} = U^i_j \cup U^{i-1}_j \cup L^i_j \cup \compl{L}^{i-1}_j$.
    The relation accepts a selection $S \cap \scp{\consistencyRelation{i}{j}}$ if and only if the following properties all hold:
    \begin{enumerate}
        \item $\state{\statevertex{i}{j}{1},\statevertex{i}{j}{2},\dots,\statevertex{i}{j}{\blocksPerVertexGadget}} \in A$.
        \item For all $p \in \range{\blocksPerVertexGadget}$ we have that $\state{\statevertex{i}{j}{p}} = \complState{\rightState{\statevertex{i-1}{j}{p}}}$.
        \item For all $p \in \range{\blocksPerVertexGadget}$ we have that exactly $\weight{\state{\statevertex{i}{j}{p}}}$ vertices of $N(\statevertex{i}{j}{p}) \cap L^i_j$ are selected.
        \item For all $p \in \range{\blocksPerVertexGadget}$ we have that exactly $\weight{\rightState{\statevertex{i-1}{j}{p}}}$ vertices of $N(\statevertex{i-1}{j}{p}) \cap \compl{L}^{i-1}_j$ are selected.
        \item At least $\simpleManagerLeftSolSize$ vertices of $L^i_j$ are selected.
        \item At least $\simpleManagerRightSolSize$ vertices of $\compl{L}^{i-1}_j$ are selected.
    \end{enumerate}

    Next, for each bag $\bag{i}$ and variable $\variable{j} \in \bag{i}$ with $\variable{j} \notin \bag{i+1}$, we add a final relation $\finalRelation{i}{j}$.
    The scope of this relation is $U^i_j \cup \compl{L}^i_j$.
    A selection of $\scp{\finalRelation{i}{j}}$ is accepted if and only if
    \begin{enumerate}
        \item $\complState{\rightState{\statevertex{i}{j}{1}}}\dots\complState{\rightState{\statevertex{i}{j}{\blocksPerVertexGadget}}} \in \usedStates$, and
        \item for all $p \in \range{\blocksPerVertexGadget}$ exactly $\weight{\rightState{\statevertex{i}{j}{p}}}$ vertices of $N(\statevertex{i}{j}{p}) \cap \compl{L}^i_j$ are selected, and
        \item at least $\simpleManagerRightSolSize$ vertices of $\compl{L}^i_j$ are selected.
    \end{enumerate}

    For each $\ell \in \range{\constraintCount}$ we create a \emph{constraint} relation $\constraintRelation{\ell}$ (for constraint $\constraint{\ell}$).
    Let $(x_j,x_{j'})$ be the variables appearing in constraint $\constraint{\ell}$.
    Recall that $\constraintBagMapping{\constraint{\ell}}$ is the index of a bag that contains $x_j$ and $x_{j'}$.
    Set $i = \constraintBagMapping{\constraint{\ell}}$.
    The scope of the relation is $\scp{\constraintRelation{\ell}} = L^i_{j} \cup U^i_j \cup L^i_{j'} \cup U^i_{j'}$.
    The relation accepts a selection of $\scp{\constraintRelation{\ell}}$ if and only it is the case that
    \begin{enumerate}
        \item $(\state{\statevertex{i}{j}{1},\dots,\statevertex{i}{j}{g}},\state{\statevertex{i}{j'}{1},\dots,\statevertex{i}{j'}{g}}) \in \acc{\constraint{\ell}}$, and
        \item at least $\simpleManagerLeftSolSize$ vertices of $L^i_j$ are selected, and
        \item at least $\simpleManagerLeftSolSize$ vertices of $L^i_{j'}$ are selected.
    \end{enumerate}

    Let $\mathcal{R}_0$ be the set of all relations described so far.
    Compute the superset-closure $\mathcal{R}$ of $\mathcal{R}_0$.
    That is, $\mathcal{R}$ is created by iterating over all $R_0 \in \mathcal{R}$ and adding a relation $R$ with $\scp{R} = \scp{R_0}$ and $\acc{R} = \{X' \mid X \in \acc{R_0}, X \subseteq X' \subseteq \scp{R_0}\}$.

    \subparagraph*{The solution size.}

    Set $\simpleManagerPortSolSize = \selStateVertexPerGadget$.
    We set the solution size to
    \[k = \sum_{i \in \range{\bagCount}} \sum_{\variable{j} \in \bag{i}} (\simpleManagerLeftSolSize + \simpleManagerRightSolSize + \simpleManagerPortSolSize).\]
    This value corresponds to the minimum solution size in each simple manager multiplied by the number of simple managers in the graph.
    The output instance is $I' = (G,k,\mathcal{R})$.
    Next, we show that our construction has the desired properties.

    \subparagraph*{Properties of the construction.}

    We begin by proving the forward direction of the correctness of the reduction.
    \begin{claim}
        \label{claim:non_partial_high_lvl_correctness_forward}
        If the input instance $I$ is a yes-instance, then the output instance $I'$ is a yes-instance.
    \end{claim}
    \begin{claimproof}
        Before we start, recall that for any state $\tau_c \in \allStatesNonPartial$ we have $\complState{\complState{\tau_c}} = \tau_c$.
        Now, assume that $I$ is a yes-instance.
        Then, there exists a variable assignment $\delta$ which fulfills all constraints of $I$.
        We build a $(\sigma,\rho)$-set of $I'$, initially start with an empty set $S$.

        For each variable $x_j$, we have that $\vec s = \delta(x_j)$ is a vector of $A$.
        Now, consider a bag $X_i$ such that $x_j \in X_i$.
        In our construction, we introduced a simple $A$-manager for this bag and variable.
        The set $U$ of this simple manager is the set of state vertices $U^i_j = \{\statevertex{i}{j}{1},\dots,\statevertex{i}{j}{g}\}$.
        By the properties of the simple $A$-managers, there is a $(\sigma,\rho)$-set $S_{\vec s}$ of the manager graph.
        Fix an arbitrary $p \in \range{g}$ and let $\vec s[p] = \tau_c$.
        The set $S_{\vec s}$ has the following properties:
        \begin{itemize}
            \item Set $S_{\vec s}$ selects vertex $\statevertex{i}{j}{p}$ if and only if $\tau$ is the symbol $\sigma$.
            \item Vertex $\statevertex{i}{j}{p}$ has exactly $c$ neighbors in $L^i_j$.
            \item If $\statevertex{i}{j}{p}$ is selected, then it has exactly $\sigLargestNonPartial - c$ selected neighbors in $\compl{L}^i_j$.
            \item If $\statevertex{i}{j}{p}$ is not selected, then it has exactly $\rhoLargestNonPartial - c$ selected neighbors in $\compl{L}^i_j$.
        \end{itemize}
        Moreover, $S_{\vec s}$ selects exactly $\simpleManagerLeftSolSize$ vertices of $L^i_j$, $\simpleManagerRightSolSize$ of $\compl{L}^i_j$, and $\simpleManagerPortSolSize = \selStateVertexPerGadget$ vertices of $U^i_j$.
        Observe that, using this selection the state $\state{\statevertex{i}{j}{p}}$ is exactly $\tau_c$.
        This implies that $\state{\statevertex{i}{j}{1},\dots,{\statevertex{i}{j}{p}}} = \delta(x_j) \in A$.
        Similarly, we have that $\rightState{\statevertex{i}{j}{p}} = \complState{\tau_c}$.
        This yields $\rightState{\statevertex{i}{j}{p}} = \complState{\state{\statevertex{i}{j}{p}}}$ and $\state{\statevertex{i}{j}{p}} = \complState{{\rightState{\statevertex{i}{j}{p}}}}$.
        Furthermore, because $\weight{\tau_c} = c$, $S_{\vec s}$ selects exactly $\weight{\tau_c}$ neighbors of $\statevertex{i}{j}{p}$ in $L^i_j$.
        Also, note that $\weight{\rightState{\statevertex{i}{j}{p}}} = c - \sigLargestNonPartial$ when $\statevertex{i}{j}{p}$ is selected, and $c - \rhoLargestNonPartial$ when $\statevertex{i}{j}{p}$ is not selected.
        Hence, we have that $S_{\vec s}$ selects exactly $\weight{\rightState{\statevertex{i}{j}{p}}}$ neighbors of $\statevertex{i}{j}{p}$ in $\compl{L}^i_j$.
        We add this solution $S_{\vec s}$ to the set $S$.

        Since distinct simple managers of the output instance are not connected by edges, the resulting set $S$ is a $(\sigma,\rho)$-set of $G$.
        Because we select exactly $\simpleManagerLeftSolSize + \simpleManagerRightSolSize + \simpleManagerPortSolSize$ vertices in each simple manager, we additionally have $|S| = k$.

        It remains to verify that $S$ fulfills the relations of $\mathcal{R}$.
        We will prove that $S$ even fulfills the relations of $\mathcal{R}_0$.
        By the elaborations above about the solution $S_{\vec s}$, it is clear that the conditions of the initial relations are fulfilled.

        For the consistency relations, observe that for any $x_j \in X_i \cap X_{i-1}$ we select the same solution in the simple manager for $x_j$ in bag $X_i$, and the simple manager for $x_j$ in bag $X_{i-1}$.
        Fix an arbitrary $p \in \range{t}$.
        We then have $\state{\statevertex{i}{j}{p}} = \state{\statevertex{i-1}{j}{p}}$ for all $p \in \range{g}$.
        Similarly, we ensure $\rightState{\statevertex{i}{j}{p}} = \rightState{\statevertex{i-1}{j}{p}}$.
        In particular, we have $\state{\statevertex{i}{j}{p}} = \complState{\rightState{\statevertex{i}{j}{p}}} = \complState{\rightState{\statevertex{i-1}{j}{p}}}$, as required by the relation.
        The other properties of the consistency relations follow straightforward from our selection within each simple manager.

        For a final relation $\finalRelation{i}{j}$, the second and third property again follow immediately from the properties of the selection within the simple managers.
        The first property is implied by the fact that $\complState{\rightState{\statevertex{i}{j}{p}}} = \state{\statevertex{i}{j}{p}}$ for all $p \in \range{g}$ and our selection within the simple manager.

        Finally, let us consider constraint relation $\constraintRelation{\ell}$.
        Set $i = \constraintBagMapping{\constraint{\ell}}$, and let $x_j$ and $x_{j'}$ be the variables appearing in the constraint.
        The second and third property of the constraint relations are fulfilled since the selection in each simple manager fulfills these properties.
        The first property follows from $\state{\statevertex{i}{j}{1},\dots,\statevertex{i}{j}{p}} = \delta(x_j)$, $\state{\statevertex{i}{j'}{1},\dots,\statevertex{i}{j'}{p}} = \delta(x_{j'})$, and $(\delta(x_j),\delta(x_{j'})) \in \acc{\constraint{\ell}}$.
        So $(G,k,\mathcal{R}_0)$ is a yes-instance, which implies that also $I' = (G,k,\mathcal{R})$ is a yes-instance.
    \end{claimproof}

    Now, we show that any set that fulfills all relations of the output instance $I'$ has size at least $k$, which means that the output instance is good.

    \begin{claim}
        Any set $S$ that fulfills all relations of $\mathcal{R}$ has size at least $k$.
        That is, $I'$ is a good-instance.
    \end{claim}
    \begin{claimproof}
        By the relations acting on each simple manager, $S$ must select at least $\simpleManagerLeftSolSize$ vertices of the left side of each manager, at least $\simpleManagerRightSolSize$ vertices of the right side of each manager, and at least $\simpleManagerPortSolSize = \selStateVertexPerGadget$ state-vertices of each manager.
        Then, we immediately obtain that any set $S$ that fulfills all relations of $I'$ must select at least $\simpleManagerLeftSolSize + \simpleManagerRightSolSize + \simpleManagerPortSolSize$ vertices of each simple manager.
        But the solution size $k$ is exactly $\simpleManagerLeftSolSize + \simpleManagerRightSolSize + \simpleManagerPortSolSize$ times the number of simple managers of $I'$, so, any set $S$ that fulfills all relations of $I'$ must have size at least $k$.
    \end{claimproof}

    Next, we show that any sufficiently small set that respects all relations of $\mathcal{R}$ must also respect all relations of $\mathcal{R}_0$.
    This property is used to show the backwards direction of safety later.

    \begin{claim}
        \label{claim:non_partial_high_lvl_small_set_fulfills_relations}
        Any set $S \subseteq V(G)$ that fulfills all relations of $\mathcal{R}$ and has size at most $k$ fulfills all relations of $\mathcal{R}_0$.
    \end{claim}
    \begin{claimproof}
        Let $S$ be a set of size at most $k$ that fulfills all relations of $\mathcal{R}$.
        We now argue that then, $S$ must actually fulfill all relations of $\mathcal{R}_0$.
        Consider an arbitrary simple manager of $G$.
        In the set $\mathcal{R}_0$, we have an initial-relation or consistency-relation that ensures at least $\simpleManagerLeftSolSize$ vertices of the left side of the simple manager are selected.
        Similarly, the relation ensures that at least $\simpleManagerPortSolSize$ state-vertices of the simple manager are selected, because each vector of $A$ has exactly $\simpleManagerPortSolSize = \selStateVertexPerGadget$ $\sigma$-states.
        A consistency-relation or a final-relation similarly ensures that at least $\simpleManagerRightSolSize$ vertices of the right side of the simple manager are selected.
        Since $\mathcal{R}$ is just the superset-closure of $\mathcal{R}_0$, the relations of $\mathcal{R}$ still guarantee these properties.
        Recall that $k$ is just $\simpleManagerLeftSolSize + \simpleManagerRightSolSize + \simpleManagerPortSolSize$ times the number of simple managers of $G$.
        Hence, the set $S$ must select exactly $\simpleManagerLeftSolSize$ vertices of the left side of each simple manager.
        Similarly, $S$ must select exactly $\simpleManagerRightSolSize$ vertices of the right side of each simple manager.
        Finally, $S$ must select exactly $\simpleManagerPortSolSize$ state-vertices of each simple manager.

        Now, consider an arbitrary relation $R$ of $\mathcal{R}$, and let $R_0$ be the relation corresponding to $R$ in $\mathcal{R}_0$.
        Since $S$ fulfills all relations of $\mathcal{R}$, the set $S$ fulfills relation $R$.
        Thus, the selection of $S$ in $\scp{R} = \scp{R_0}$ is a (not necessarily proper) superset of some selection in $\acc{R_0}$.
        If $R_0$ is an initial relation, then $\scp{R} = L^i_j \cup U^i_j$ for some $i,j$.
        Towards a contradiction, assume that $S \cap \scp{R}$ is a proper superset of some selection in $\acc{R_0}$.
        Observe that each accepted selection of $R_0$ selects at least $\simpleManagerLeftSolSize$ vertices of $L^i_j$, and at least $\simpleManagerPortSolSize$ vertices of $U^i_j$.
        Hence, $S$ must select more than $\simpleManagerLeftSolSize$ vertices of $L^i_j$, or more than $\simpleManagerPortSolSize$ vertices of $U^i_j$.
        In the last paragraph, we argued that this cannot be the case.
        Hence, $S$ cannot select a proper superset of some selection in $\acc{R_0}$.
        So, the selection of $S$ is in $\acc{R_0}$.

        The arguments for the remaining types of relations are analogous.
        If the relation scope contains set $L^i_j$, then the relation ensures at least $\simpleManagerLeftSolSize$ selected vertices of that set.
        If the relation scope contains set $\compl{L}^i_j$, it ensures at least $\simpleManagerRightSolSize$ selected vertices of that set, and if it contains some set $U^i_j$, at least $\simpleManagerPortSolSize$ selected vertices of that set.
        There is no room to select a proper superset of an accepted selection.
        So, $S$ fulfills all relations of $\mathcal{R}_0$.
    \end{claimproof}

    Now, we show that a $(\sigma,\rho)$-set that fulfills all relations of $\mathcal{R}_0$ yields an accepting variable-assignment for the input instance.

    \begin{claim}
        \label{claim:non_partial_high_lvl_backwards_1}
        If there is a $(\sigma,\rho)$-set set $S \subseteq V(G)$ of $G$ that fulfills all relations of $\mathcal{R}_0$, then the input instance is a yes-instance.
    \end{claim}
    \begin{claimproof}
        Let $S$ be a $(\sigma,\rho)$-set of $G$ that fulfills all relations of $\mathcal{R}_0$.
        We can prove that the state propagation works using essentially the same arguments as for the partial version when disregarding the possibility of violating vertices.
        This then also roughly amounts to those arguments used by Focke et al. \cite{fockeTightComplexityBoundsLowerBound} to ensure that the states propagate correctly.

        Consider an arbitrary variable $\variable{j}$ such that $\variable{j} \in \bag{i-1} \cap \bag{i}$.
        We want to show that $\state{\statevertex{i-1}{j}{p}} = \state{\statevertex{i}{j}{p}}$ for all $p \in \range{g}$, that is, the states of the state-vertices for the same variable are consistent.
        This ensures that those state-vertices representing the same variable in different bags have the same state.

        Fix an arbitrary $p \in \range{\blocksPerVertexGadget}$.
        The relation $\consistencyRelation{i}{j}$ ensures that $\state{\statevertex{i}{j}{p}} = \complState{\rightState{\statevertex{i-1}{j}{p}}}$.
        This guarantees that $\statevertex{i-1}{j}{p}$ is selected if and only if $\statevertex{i}{j}{p}$ is selected.
        It remains to argue that $\weight{\state{\statevertex{i-1}{j}{p}}} = \weight{\state{\statevertex{i}{j}{p}}}$.
        The consistency relation $\consistencyRelation{i}{j}$ also enforces that \[|S \cap L^i_j \cap N(\statevertex{i}{j}{p})| = \weight{\state{\statevertex{i}{j}{p}}}.\]
        Furthermore, it guarantees that $|S \cap \compl{L}^{i-1}_j \cap N(\statevertex{i-1}{j}{p})| = \weight{\rightState{\statevertex{i-1}{j}{p}}}$.
        We have that $\weight{\complState{\rightState{\statevertex{i-1}{j}{p}}}} = \tauLargestNonPartial - \weight{\rightState{\statevertex{i-1}{j}{p}}}$ by definition.
        From $\state{\statevertex{i}{j}{p}} = \complState{\rightState{\statevertex{i-1}{j}{p}}}$ we obtain that  $\weight{\state{\statevertex{i}{j}{p}}} = \weight{\complState{\rightState{\statevertex{i-1}{j}{p}}}}$.
        Thus, we obtain \[\weight{\state{\statevertex{i}{j}{p}}} = \tauLargestNonPartial - \weight{\rightState{\statevertex{i-1}{j}{p}}},\] yielding $\weight{\rightState{\statevertex{i-1}{j}{p}}} + \weight{\state{\statevertex{i}{j}{p}}} = \tauLargestNonPartial$.
        Observe that $\statevertex{i-1}{j}{p}$ has exactly $\weight{\state{\statevertex{i-1}{j}{p}}} + \weight{\rightState{\statevertex{i-1}{j}{p}}}$ selected neighbors.
        We consider two different cases depending on $\tau$.
        \begin{description}
            \item[$\tau$ is a finite set:] In this case, we know that $\weight{\state{\statevertex{i-1}{j}{p}}} + \weight{\rightState{\statevertex{i-1}{j}{p}}} \leq \max \tau = \tauLargestNonPartial$
                  because $S$ is a $(\sigma,\rho)$-set of the graph.
                  Additionally, we have that $\weight{\rightState{\statevertex{i-1}{j}{p}}} + \weight{\state{\statevertex{i}{j}{p}}} = \tauLargestNonPartial$.
                  This means that
                  \begin{equation}
                      \label{eq:intermediate_nonpartial_finite}
                      \weight{\state{\statevertex{i-1}{j}{p}}} \leq \weight{\state{\statevertex{i}{j}{p}}}.
                  \end{equation}
            \item[$\tau$ is a simple cofinite set:] We know that $\weight{\state{\statevertex{i-1}{j}{p}}} + \weight{\rightState{\statevertex{i-1}{j}{p}}} \geq \min \tau = \tauLargestNonPartial$.
                  From $\weight{\rightState{\statevertex{i-1}{j}{p}}} + \weight{\state{\statevertex{i}{j}{p}}} = \tauLargestNonPartial$ we then obtain
                  \begin{equation}
                      \label{eq:intermediate_nonpartial_simple_cofinite}
                      \weight{\state{\statevertex{i-1}{j}{p}}} \geq \weight{\state{\statevertex{i}{j}{p}}}.
                  \end{equation}
        \end{description}

        Next, we use the fact that all vectors of $A$ have the same $\sigma$-weight $\selWeight$, and the same weight.
        Hence, each vector of $A$ also has the same $\rho$-weight $\unselWeight$.
        Let $U^i_{j,\sigma}$ be the selected vertices of $U^i_j$, and $U^i_{j,\rho}$ be the unselected vertices of $U^i_j$.
        Similarly define $U^{i-1}_{j,\sigma}$ and $U^{i-1}_{j,\rho}$.
        By the relations we utilize, for each state vertex $\statevertex{i'}{j'}{p'}$ it holds that the number of selected neighbors of the state vertex in $L^{i'}_j$ is exactly $\weight{\state{\statevertex{i'}{j'}{p'}}}$, and that the number of selected neighbors of the state vertex in $\compl{L}^{i'}_{j'}$ is exactly $\weight{\rightState{\statevertex{i'}{j'}{p'}}}$.
        Hence, we must have
        \begin{equation}
            \label{eq:intermediate_nonpartial_weight_sum}
            \sum_{v \in U^{i-1}_{j,\rho}} |S \cap L^i_j \cap N(v)| = \sum_{v \in U^{i-1}_{j,\rho}}\weight{\state{v}} = \unselWeight{} = \sum_{v \in U^{i}_{j,\rho}}\weight{\state{v}},
        \end{equation} and $\sum_{v \in U^{i-1}_{j,\sigma}} \weight{\state{v}} = \sum_{v \in U^{i}_{j,\sigma}} \weight{\state{v}}$.
        Recall that $\rho$ is simple cofinite.
        \cref{eq:intermediate_nonpartial_simple_cofinite} holds for any $p \in \range{\blocksPerVertexGadget}$.
        Then, \cref{eq:intermediate_nonpartial_weight_sum} implies that $\weight{\statevertex{i-1}{j}{p}}$ is exactly the same as $\weight{\statevertex{i}{j}{p}}$ for all $v = \statevertex{i}{j}{p} \in U^i_{j, \rho}$.
        The argument for vertices of $U^i_{j,\sigma}$ is analogous, using \cref{eq:intermediate_nonpartial_finite} if $\sigma$ is finite and \cref{eq:intermediate_nonpartial_simple_cofinite} if $\sigma$ is simple cofinite.

        Now, let us define a variable assignment $\delta$ for the input instance.
        For each variable $x_j$, let $X_i$ be an arbitrary bag such that $x_j \in X_i$.
        We then set $\delta(x_j) = \state{\statevertex{i}{j}{1},\dots,\statevertex{i}{j}{p}}$.
        Observe that the choice of the bag does not matter beyond the fact that it must contain $x_j$, because the states of the state-vertices are consistent across these bags.
        Then, for each constraint $\constraint{\ell}$ it holds that $\delta$ fulfills the constraint.
        Otherwise, the first condition of the relation $\constraintRelation{\ell}$ would be violated.
    \end{claimproof}

    Now, we have all the tools ready to prove the correctness of the backwards direction.
    \begin{claim}
        \label{claim:non_partial_high_lvl_correctness_backwards}
        If the output instance $I'$ is a yes-instance, then input instance $I$ is a yes-instance.
    \end{claim}
    \begin{claimproof}
        Let $S$ be a vertex subset that has size at most $k$ and that fulfills all relations of $\mathcal{R}$.
        Then, by \cref{claim:non_partial_high_lvl_small_set_fulfills_relations}, $S$ actually fulfills all relations of $\mathcal{R}_0$.
        The proof then follows from \cref{claim:non_partial_high_lvl_backwards_1}.
    \end{claimproof}

    \begin{claim}
        The pathwidth of $I'$ is $\pw \cdot g + \Oh(1)$, and a corresponding decomposition can be computed in polynomial-time.
    \end{claim}
    \begin{claimproof}
        The following proof is almost identical to the proof we give for the nonpartial problem since the constructed graphs largely have the same structure.
        Recall that a path decomposition of the output instance is a path decomposition of the graph $G$, such that additionally, for each relation there exists a bag containing all vertices of the relation scope.
        Throughout the proof, we treat any set whose indices are out of bounds as the empty set.
        Also, recall that we set $\bag{0} = \bag{t+1} = \emptyset$ for convenience.

        Begin by copying the path decomposition provided with the input, denote the copy as $\outputBag{1},\dots,\outputBag{\bagCount}$.
        We begin by modifying the path decomposition by introducing two empty bags $\outputBag{0}$ and $\outputBag{\bagCount+1}$.
        In a first step, for any bag $\outputBag{i}$ and $\variable{j} \in \outputBag{i}$, we replace $\variable{j}$ with the state vertices $U^i_j$.

        Proceed as follows for any $i \in \range[1]{\bagCount+1}$.
        Order the variables of $\bag{i-1} \cup \bag{i}$ as \[\variable{\lambda_1^i},\dots,\variable{\lambda_r^i},\variable{\lambda_{r+1}^i},\dots,\variable{\lambda_q^i},\variable{\lambda_{q+1}^i},\dots,\variable{\lambda_z^i},\] such that for any $j \in \range{r}$, $\variable{\lambda_j^i}$ is in $\bag{i-1}$ but not in $\bag{i}$, for any $j \in \range[r+1]{q}$, we have that $\variable{\lambda_j^i}$ is in $\bag{i-1} \cap \bag{i}$, and for any $j \in \range[q+1]{z}$, we have that $\variable{\lambda_j^i}$ is in $\bag{i} \setminus \bag{i-1}$.
        Then, for any pair of adjacent bags $\outputBag{i-1}$ and $\outputBag{i}$, and any $j \in \range{z}$, we introduce initially empty bag $\outputBag{i}^j$.
        In the path decomposition, we put these new bags between $\outputBag{i-1}$ and $\outputBag{i}$ while keeping them in ascending order.
        That is, the bags $\outputBag{i-1},\outputBag{i}^1,\dots,\outputBag{i}^z,\outputBag{i}$ will form a subsequence of the path decomposition.
        The idea of these additional bags is that they should allow us to slowly transition from $\outputBag{i-1}$ to $\outputBag{i}$.
        So, we need to explain which vertices we put into them next.
        \begin{itemize}
            \item We first handle the vertices of variables in $\bag{i-1}$ which are not in $\bag{i}$.
                  For all $j \in \range{r}$, we add the vertices $\outputBag{i-1} \cup \compl{L}^{i-1}_{\lambda_j^i}$ to $\outputBag{i}^{j}$.
            \item Now, we want to handle the vertices of the variables in $\bag{i-1} \cap \bag{i}$.
                  For all $j \in \range[r+1]{q}$, we add the vertices $\bigcup_{j' \in \range[r+1]{j-1}} U^{i}_{\lambda_{j'}^i} \cup \bigcup_{j' \in \range[j+1]{q}} U^{i-1}_{\lambda_{j'}^{i}} \cup
                      U^i_{\lambda_j^i} \cup U^{i-1}_{\lambda_j^i} \cup \compl{L}^{i-1}_{\lambda_j^i} \cup L^i_{\lambda_j^i}$ to $\outputBag{i}^{j}$.
            \item We handle those vertices of the variables which are part of $\bag{i}$ but not part of $\bag{i-1}$.
                  For all $j \in \range[q+1]{z}$, we add the vertices $\outputBag{i} \cup L^i_{\lambda_j^i}$ to $\outputBag{i}^j$.
            \item If there is a constraint $C_\ell$ such that $\constraintBagMapping{C_\ell} = i$, then, we add all vertices of $\scp{\constraintRelation{\ell}}$ to the bag $\outputBag{i}^j$, for any $j \in \range{z}$, and also to the bag $\outputBag{i}$.
        \end{itemize}
        This concludes the description of the path decomposition.

        Now, we formally show that the sketched sequence of bags is actually a valid path decomposition.
        First, we show that any vertex appears in some bag, and that for any initial-relation, final-relation and consistency-relation there exists some bag containing the relation scope.
        Each vertex is in the scope of an initial-relation, a consistency-relation, or a final-relation.
        Observe that for any $\variable{j}$, if $\variable{j}$ is in $\bag{i} \setminus \bag{i-1}$, then, $\outputBag{i}^{j'}$ contains $U^i_j \cup L^i_j = \scp{\initialRelation{i}{j}}$, where $j'$ is so that $j = \lambda_{j'}^i$.
        Similarly, if $\variable{j}$ is both in $\bag{i}$ and $\bag{i-1}$, there exists bag $\outputBag{i}^{j'}$ that contains $U^i_j \cup L^i_j \cup U^{i-1}_j \cup \compl{L}^{i-1}_j = \scp{\consistencyRelation{i}{j}}$,
        where $j'$ is so that $j = \lambda_{j'}^i$.
        Finally, if $\variable{j}$ is in $\bag{i-1}$ but not in $\bag{i}$, then let $j'$ such that $j = \lambda^{i}_{j'}$.
        Then bag $\outputBag{i}^{j'}$ contains $U^{i-1}_j \cup \compl{L}^{i-1}_j = \scp{\finalRelation{i-1}{j}}$.
        This means that we indeed cover all initial-, consistency-, and final-relations, and thus also all vertices.
        Regarding the constraint relation $\constraintRelation{\ell}$, for any constraint $C_\ell$, observe that $\scp{\constraintRelation{\ell}}$ is contained in all bags between and including bag $\outputBag{\constraintBagMapping{C_\ell}}^1$ and $\outputBag{\constraintBagMapping{C_\ell}}$.

        It is also not difficult to see that each edge of $G$ appears in some bag.
        An edge that is within a set $U^i_j$, $L^i_j$, or $\compl{L}^i_j$ is covered as these sets appear in some bag.
        The only other types of edges are between $U^i_j$ and $L^i_j$, or between $U^i_j$ and $\compl{L}^i_j$ (observe that there are no edges between $L^i_j$ and $\compl{L}^i_j$ by the properties of the simple managers).
        The sets $U^i_j \cup L^i_j$ are always subject to an initial-relation or consistency-relation, and the sets $U^i_j \cup \compl{L}^i_j$ are always subject to a consistency-relation or a final-relation.
        Hence, we have already argued that these are covered.

        Finally, it remains to argue that for any vertex $v$, all nodes of bags containing $v$ form a connected subpath of the path decomposition.
        For the vertices of the sets $L^i_j$ or $\compl{L}^i_j$, it is also the case that either they appear in exactly one bag, or they are subject to some constraint relation, in which case they appear in exactly the bags in which the scope of the constraint relation appears, which form a connected subpath.

        So, all that remains is arguing about vertices in a set $U^i_j$, which is the set of state-vertices of a variable $\variable{j} \in \bag{i}$.
        Clearly $U^i_j$ is part of $\outputBag{i}$.
        We now show that there is some integer $D_\ell$ such that $U^i_j$ is in all bags between and including $\outputBag{i}^{D_\ell}$ and $\outputBag{i}$, but in no bag before $\outputBag{i}^{D_\ell}$.

        Consider that $U^i_j$ is subject to a constraint relation, in that case $U^i_j$ is explicitly added to all bags between and including $\outputBag{i}^1$ and $\outputBag{i}$.
        Moreover, $U^i_j$ is not part of bag $\outputBag{i-1}$.
        Therefore, in this case $D_\ell = 1$.

        Otherwise, if we have that $\variable{j}$ is not in $\bag{i-1}$, then there is some $q$ (namely, the lowest integer $q$ such that $\variable{\lambda_{q+1}^i}$ is in $\bag{i} \setminus \bag{i-1}$) such that $U^i_j$ is part of the subsequence $\outputBag{i}^{q+1},\dots,\outputBag{i}$, but $U^i_j$ is not part of any bag before $\outputBag{i}^{q+1}$.
        In this case $D_\ell  = q + 1$

        Otherwise, if we have that $\variable{j}$ is in $\bag{i-1}$, then $U^i_j$ is in all bags of the subsequence $\outputBag{i}^{j'},\dots,\outputBag{i}$, but not part of any bag before $\outputBag{i}^{j'}$, where $j'$ is such that $j = \lambda_{j'}^i$.
        In this case $D_\ell = j'$.

        Now, we show that there is some integer $D_r$ such that $U^i_j$ is in all bags between and including $\outputBag{i}$ and bag $\outputBag{i+1}^{D_r}$.

        If we have that $\variable{j}$ is not in $\bag{i+1}$, then there is some integer $r$ (namely, the highest integer $r$ such that $\variable{\lambda_{r}^{i+1}}$ is not in $\bag{i+1}$), such that $U^i_j$ is contained in the subsequence of bags $\outputBag{i},\dots,\outputBag{i+1}^{r}$, but not in any bag after $\outputBag{i+1}^r$.
        In this case $D_r = r$.

        If we have that $\variable{j}$ is in $\bag{i}$ and $\bag{i+1}$, then let $j'$ be such that $j = \lambda^{i+1}_{j'}$, and observe that $U^i_j$ is in the subsequence $\outputBag{i},\dots,\outputBag{i+1}^{j'}$, but $U^i_j$ is not contained in any bag after $\outputBag{i+1}^{j'}$.
        In this case $D_r = j'$.

        Overall, we have that $U^i_j$ is exactly in the bags $\outputBag{i}^{D_\ell},\dots,\outputBag{i},\dots,\outputBag{i+1}^{D_r}$.
        It is now not difficult to see that our approach results in a valid path decomposition of width $\pw \cdot \blocksPerVertexGadget + \Oh(1)$, because our initial step of replacing the variables with the state vertices blows up the pathwidth by a factor of $\blocksPerVertexGadget$, and all the other steps ensure that each bag contains only $\pw \cdot \blocksPerVertexGadget + \Oh(1)$ state vertices, as well as a constant number of further vertices, where we utilize the fact that the sizes of $L^i_j$ and $\compl{L}^i_j$ for any $i,j$ are constant.
        It is also clear that we can compute the sketched decomposition in polynomial-time.
    \end{claimproof}

    \subparagraph*{Finishing the proof.}

    Assume that there exists an algorithm that can solve good instances $\tilde I$ of \minParGenDomSetRel{} with arity at most $d$ in time $(|\allStatesNonPartial| - \varepsilon)^\pw \cdot |\tilde I|^{\Oh(1)}$ for some $\varepsilon > 0$ when provided with a path decomposition of width $\pw$ of $\tilde I$.

    We take an instance $I$ of 2-\problem{CSP} over an alphabet of size $|A|$ together with a path decomposition of its primal graph of width $\pw$ as input, and utilize the reduction this proof provides for this fixed $\varepsilon$.
    The reduction produces an equivalent good instance $I'$ with pathwidth $\pw \cdot g + \Oh{(1)}$ in polynomial-time.
    We apply the hypothetical algorithm on this instance.
    This whole procedure decides instance $I$ in time
    \begin{align*}
        |I|^{\Oh(1)} + (|\allStatesNonPartial| - \varepsilon)^{\pw \cdot g} \cdot |I'|^{\Oh(1)} & \leq (|\allStatesNonPartial| - \varepsilon)^{\pw \cdot g} \cdot |I|^{\Oh(1)} \\
                                                                                                & \leq (|A| - \varepsilon')^\pw \cdot |I|^{\Oh(1)},
    \end{align*}
    for a constant $\varepsilon' > 0$.
    So, the \pwseth{} would be false.
\end{proof}

\subsection{Finishing the Lower Bound Proof}

Now, we need to provide reductions from the intermediate problem to \minGenDomSet{}.
Here, we utilize that the reductions we provided for the partial problem also work for the nonpartial problem.

\mtheoremLowerBoundNonPartial*{}
\begin{proof}
    If $\sigma$ and $\rho$ are both finite, the result is proven in the thesis of Schepper \cite[Main Theorem III.3]{schepperFasterHigherEasier}.

    So, for the rest of the proof assume that at least one of $\sigma$ or $\rho$ is simple cofinite.
    Note that then $(\sigma,\rho)$ is not $m$-structured for any $m \geq 2$, so we are looking for a lower bound matching the base $\sigLargestNonPartial + \rhoLargestNonPartial + 2 = |\allStatesNonPartial|$.
    We can use the fact that \emph{all} reductions provided in \cref{sec:realizing_relations} do not change the input value $\ell$ of allowed violations.
    The reductions can handle if the input problem allows some number of violations, but they do not exploit the fact that they are for the partial problem to work properly.
    Hence, when the input value $\ell = 0$ they again output an instance that does not allow any violations.
    So, they also work for the nonpartial problem variants.
    We again utilize two different reduction chains depending on whether $\rho$ is simple cofinite or not.

    \begin{description}
        \item[$\rho$ is simple cofinite:] We utilize the intermediate lower bound of \cref{thm:minimization_non_partial_high_lvl} which works for good instances of bounded arity.
              Then, \cref{thm:partial_replace_superset_by_hw_geq_1} is used to replace all relations with $\hwGeqOneRelation$-relations via a pathwidth-preserving reduction.
              Finally, \cref{thm:partial_rho_cofinite_replace_relations} provides a reduction from the resulting instance to \minGenDomSet{}.
        \item[$\rho$ is finite:] This implies that $\sigma$ is simple cofinite since we assume at least one of $\sigma, \rho$ to be simple cofinite. We utilize the intermediate lower bound of \cref{thm:intermediate_lower_bound_sigma_cofinite_rho_finite}.
              Then, \cref{thm:replacing_arbitrary_with_hw_1} (or \cite[Corollary 8.8]{fockeTightComplexityBoundsLowerBound}) provides a reduction to the problem in which all relations are $\hwRelation$-relations.
              Finally, \cref{thm:rho_finite_partial_realize_relations_reduction} provides a reduction to \minGenDomSet{}.
    \end{description}
    All used reductions are pathwidth-preserving, and also the arity of each used instance of an intermediate problem is bounded by a constant.
    Hence, the lower bound from the intermediate problems transfer to \minGenDomSet{}.
\end{proof} 
\newpage
\phantomsection
\addcontentsline{toc}{section}{References}
\bibliography{bib}

\end{document}